\documentclass[11pt]{article}
\usepackage{fullpage}
\usepackage{amsmath,amsfonts,amsthm,mathrsfs,xspace,graphicx}
\usepackage[backref,colorlinks,citecolor=blue,bookmarks=true]{hyperref}
\usepackage{mathpazo}
\usepackage{enumitem}
\usepackage[utf8]{inputenc}
\usepackage{amssymb}
\usepackage{amsmath}
\usepackage[margin=0.7in]{geometry}
\usepackage{hyperref,color}
\usepackage{graphicx}
\usepackage{float}
\usepackage{comment}
\usepackage{braket}
\usepackage{xparse,etoolbox}
\usepackage{xspace}

\newtheorem{theorem}{Theorem}[section]
\newtheorem{lemma}[theorem]{Lemma}
\newtheorem{proposition}[theorem]{Proposition}
\newtheorem{corollary}[theorem]{Corollary}

\newcounter{claimlevel}[theorem]

\ExplSyntaxOn
\NewDocumentEnvironment{claim}{O{=}}
 {
  \str_case:nn { #1 }
   {
    {=}  { }
    {+}  { \stepcounter{claimlevel} }
    {-}  { \addtocounter{claimlevel}{-1} }
    {--} { \addtocounter{claimlevel}{-2} }
    {---}{ \addtocounter{claimlevel}{-3} }
   }
  \begin{ Claim \int_to_Roman:n { \value{claimlevel} } }
 }
 {
  \end{ Claim \int_to_Roman:n { \value{claimlevel} } }
 }
\ExplSyntaxOff

\theoremstyle{definition}
\newtheorem{definition}[theorem]{Definition}
\newtheorem{protocol}[theorem]{Protocol}
\newtheorem{scenario}[theorem]{Scenario}

\theoremstyle{remark}
\newtheorem{remark}[theorem]{Remark}

\newcommand{\beq}{\begin{eqnarray}}
\newcommand{\eeq}{\end{eqnarray}}

\newcommand{\proj}[1]{\ket{#1}\!\bra{#1}}
\newcommand{\Tr}{\mbox{\rm Tr}}
\newcommand{\Id}{\ensuremath{\mathop{\rm Id}\nolimits}}
\newcommand{\Es}[1]{\mathop{\textsc{E}}_{#1}}

\newcommand{\reg}[1]{{\mathsf{#1}}}

\newcommand{\id}{\mathsf{id}}
\newcommand{\C}{\ensuremath{\mathbb{C}}}
\newcommand{\N}{\ensuremath{\mathbb{N}}}

\newcommand{\Z}{\ensuremath{\mathbb{Z}}}

\newcommand{\mA}{\mathcal{A}}
\newcommand{\mH}{\mathcal{H}}
\newcommand{\mO}{\mathcal{O}}
\newcommand{\mR}{\mathcal{R}}
\newcommand{\mF}{\mathcal{F}}
\newcommand{\mC}{\mathcal{C}}
\newcommand{\mS}{\mathcal{S}}

\newcommand{\mI}{\mathcal{I}}
\newcommand{\mK}{\mathcal{K}}
\newcommand{\mZ}{\mathcal{Z}}

\newcommand{\setft}[1]{\mathrm{#1}}
\newcommand{\Density}{\setft{D}}

\newcommand{\Unitary}{\setft{U}}

\DeclareMathOperator{\poly}{poly}
\DeclareMathOperator{\DIS}{d}
\DeclareMathOperator{\Drho}{\DIS_{\ket{\psi}}}

\newcommand{\eps}{\varepsilon}

\newcommand{\class}[1]{\textup{#1}\xspace} %

\newcommand{\NP}{\class{NP}} %
\newcommand{\QMA}{\class{QMA}} %

\newcommand{\tnote}[1]{}

\newcommand{\znote}[1]{}

\usepackage{subfiles}

\usepackage{subfiles}

\begin{document}

\title{Classical proofs of quantum knowledge}

\author{Thomas Vidick\thanks{Department of Computing and Mathematical Sciences, California Institute of Technology, USA. \texttt{vidick@caltech.edu}} \ \and Tina Zhang\thanks{Division of Physics, Mathematics and Astronomy, California Institute of Technology, USA. \texttt{tinazhang@caltech.edu} }}
\date{}
\maketitle

\begin{abstract}
We define the notion of a \emph{proof of knowledge} in the setting where the verifier is classical, but the prover is quantum, and where the witness that the prover holds is in general a quantum state. We establish simple properties of our definition, including that, if a \emph{nondestructive} classical proof of quantum knowledge exists for some state, then that state can be cloned by an unbounded adversary, and that, under certain conditions on the parameters in our definition, a proof of knowledge protocol for a hard-to-clone state can be used as a (destructive) quantum money verification protocol. In addition, we provide two examples of protocols (both inspired by private-key classical verification protocols for quantum money schemes) which we can show to be proofs of quantum knowledge under our definition. In so doing, we introduce techniques for the analysis of such protocols which build on results from the literature on nonlocal games. Finally, we show that, under our definition, the verification protocol introduced by Mahadev (FOCS 2018) is a classical \emph{argument} of quantum knowledge for QMA relations.
In all cases, we construct an explicit quantum extractor that is able to produce a quantum witness given black-box quantum (rewinding) access to the prover, the latter of which includes the ability to coherently execute the prover's black-box circuit controlled on a superposition of messages from the verifier. 
\end{abstract}

\section{Introduction}
The notion of a \emph{proof of knowledge} was first considered in the classical setting in~\cite{goldwasser1989knowledge} and subsequently formalized in~\cite{tompa1987random,feige1988zero} and~\cite{bellare1992defining}.\footnote{These three works give inequivalent definitions, but the differences are not important for the purpose of this introduction.} 
Intuitively, a proof of knowledge protocol allows a prover to convince a verifier that it `knows' or `possesses' some piece of secret information (a `witness', $w$) which satisfies a certain relation $R$ relative to a publicly known piece of information $x$. (Symbolically, we might say that the prover wants to convince its verifier that, for a particular $x$, it knows $w$ such that $R(x, w) = 1$.) For example, the witness $w$ might be a private password corresponding to a particular public username $x$, and a proof of knowledge protocol in this setting could allow the prover to demonstrate that it possesses the credentials to access sensitive information or make monetary transactions.

The formal definition of a classical proof of knowledge for $\NP$ relations $R$ was settled in a series of works (see \cite{bellare1992defining} for a summary) in the 1990s. The standard definition is as follows: the prover $P$ is said to `know' a witness $w$ if there is an extractor $E$ which, given black-box access to $P$ (including the ability to rewind $P$ and run it again on different messages from the verifier), can efficiently compute $w$. The applications of classical proofs of knowledge include identification protocols \cite{feige1988zero}, signature schemes \cite{chase2006signatures}, and encryption schemes secure against chosen-ciphertext attack \cite{schnorr2000signed}.

\medskip

In this work, we consider a particular generalisation of the classical concept of a proof of knowledge to the quantum setting. We imagine a situation where the \emph{verifier} remains classical, but the \emph{prover} is quantum, and where the witness $w$ is in general a quantum state; and we ask the prover to `convince' the verifier that it knows that state. We call this type of protocol a \emph{classical proof of quantum knowledge}. Recently, there have been works which show how a fully classical verifier can, under cryptographic assumptions, delegate a quantum computation on encrypted data to a quantum server \cite{mahadev2017classical}, verify that such a server performed the computation correctly \cite{measurement}, delegate the preparation of single-qubit states to the server in a composable fashion \cite{gheorghiu2019computationally}, and test classically that the server prepared an EPR pair in its own registers \cite{metger2020self}. In short, as long as classical computational resources and classical communication channels remain considerably less expensive than their quantum counterparts, it will be natural to wish to use classical devices to test quantum functionality. 
Although we focus here on information-theoretic rather than computational security, the current paper can be considered part of the preceding line of work.

\emph{Quantum} proofs of quantum knowledge (i.e. proof of knowledge protocols for quantum witnesses in which quantum interaction is allowed) have recently been explored by \cite{broadbent2019zero} and \cite{coladangelo2019non}; these two papers give a definition for quantum proofs of quantum knowledge, and exhibit several examples which meaningfully instantiate the definition. Here, we consider the more challenging question of defining and constructing proofs of knowledge for quantum witnesses in which the verifier and the interaction are \emph{classical}. In this setting there is an interesting difficulty involved in constructing an extractor: how does one argue that a quantum prover `knows' a certain quantum state if the only information which the prover `reveals' is classical? A first approach, following the classical definition, would only allow the extractor to access classical transcripts from the protocol. Under such a restriction, the problem the extractor faces becomes one of reconstructing a witness $\rho$ based entirely on classical measurement outcomes. It is not hard to convince ourselves that this problem probably has no solution for any non-trivial class of quantum states, as indeed it may be as hard as quantum state tomography \cite{haah2017sample}. This observation makes it clear that, while it is necessary to give the extractor quantum capabilities (since it has to produce a quantum witness), by itself this does not suffice and some form of quantum interaction with the prover must be allowed. 

\medskip

Our first contribution in this paper is to provide an adequate definition of a proof of quantum knowledge for the setting where the communication between verifier and prover is classical. In order to circumvent the difficulty described in the preceding paragraph, we adopt a definition of `black-box access to the prover' which is naturally suited to the quantum setting. Informally speaking, we model the prover as a unitary map $U$ that acts on two quantum registers, one which is private (and which is used for storing its internal state) and one which is public (used for sending and receiving messages). In each round of the real protocol, the verifier places a classical message in the public `message register', and the prover then runs the unitary $U$, before the message register is measured in the computational basis; the measurement result is the message that the prover sends to the verifier for that round. We define `black-box access to the prover' as follows: we allow the extractor to place any quantum state in the public `message register', as well as run the prover's unitary $U$, which acts on both registers, or its inverse $U^\dagger$; we do not allow the extractor to access the prover's private register except through $U$ or $U^\dagger$. We do, however, allow it to place a coherent superposition of messages in the message register, even though the verifier (in a real protocol) would only ever put one classical message there. We make use of this latter possibility in our instantiations of this definition.

This definition matches the definition of `black-box access to a quantum machine' used in previous works \cite{unruh2012quantum}. We emphasise that, even though we consider protocols with purely classical communication, the extractor according to this definition of `black-box access' is allowed to coherently manipulate a unitary implementation of the prover, and the message registers are not necessarily measured after each round of interaction. This possibility was allowed in \cite{unruh2012quantum}, but not used; here we make essential use of it when we construct our extractors.  We note also that this definition of `black-box access' matches the definition given in prior works (e.g.~\cite{watrous2009zero}) of the `black-box access' to a malicious verifier which a zero-knowledge simulator for a post-quantumly zero-knowledge proof of knowledge is allowed to have.

Having formalised what `black-box access to the prover' means in our context, we move to the task of defining a `proof of quantum knowledge' for our setting. We have two main applications in mind for a `proof of quantum knowledge': one of them (proofs of knowledge for QMA witnesses) is natural given the standard formulation of classical proofs of knowledge for NP witnesses, but the other (proofs of knowledge for \emph{quantum money states} \cite{aaronson2012review}) is both natural and unique to the quantum setting. The quantum money application does not fit well into the standard formalism which is used for NP and QMA verification. Therefore, in order to formulate our definition of a `proof of quantum knowledge' generally enough that we can capture both applications, we introduce a broader framework that mirrors frameworks recently introduced for similar purposes in the classical literature. Formally, we base our definition of a `proof of quantum knowledge' on the notion of an `agree-and-prove scheme' introduced recently in~\cite{badertscher2019agree}. The main innovation in this framework is that it allows the instance $x$ and the proof relation $R$ to be determined dynamically through interactions between the prover, the verifier, and possibly a trusted setup entity (such as the provider of a common random string or a random oracle). This framework lends itself remarkably well to our applications. Since we do not need all the possibilities that it allows, we introduce a somewhat simplified version which is sufficient for our purposes; details are given in Section~\ref{sec:quantum-aap}.

\medskip

In Section \ref{sec:properties} we show two elementary but potentially interesting properties of our definition of a `proof of quantum knowledge'. The first property is that, if a classical proof of quantum knowledge leaves the witness state intact, then the witness state can be cloned by an unbounded adversary. This is a simple no-go result which precludes certain types of proofs of quantum knowledge in the scenarios which we consider. The second property is that, under certain conditions on the parameters in the definition, a proof of knowledge protocol for a hard-to-clone witness state can also be used as a quantum money verification protocol. This result formalises the intuition that the property of being a `proof of quantum knowledge' is stronger than the property of being a quantum money verification protocol: the latter implies that no adversary can pass verification twice given access to only one money bill, and the former formalises the notion that no adversary can pass even once unless it is possible to efficiently compute the money bill by interacting with said adversary.

\medskip

Our second main contribution is to provide several examples of protocols which can be shown to be proofs of knowledge under our definition, and in so doing introduce some techniques that may possibly find use in the analysis of such protocols. As we have mentioned, instantiating a secure quantum money scheme is a natural application for a proof of quantum knowledge protocol. Conversely, quantum money verification protocols are natural candidates for examples of proofs of quantum knowledge: in a quantum money protocol, there is a prover who holds a purported money state, and who wishes to demonstrate to the verifier (who might be the bank or an independent citizen) that it does indeed `hold' or `possess' the quantum money state. The first person to describe quantum money was Wiesner \cite{wiesner1983conjugate}, who proposed money states that are tensor products of $n$ qubits, each qubit of which is chosen uniformly at random from the set $\{\ket{0}, \ket{1}, \ket{+}, \ket{-}\}$. Wiesner's states can be described classically by $2n$ classical bits, and in a quantum money scheme this classical description is kept secret by the bank; a typical classical description is the pair of strings $(x, \theta)$, where the money state can be described (denoting by $H_i$ a Hadamard gate on the $i$th qubit of the state) as $\ket{\$}_{x, \theta} = \prod_i H_i^{\theta_i} \ket{x}$. We choose to analyse as our first example of a proof of knowledge a \emph{private-key}, destructive classical money verification protocol between a prover and the bank for Wiesner's quantum money states which has been described previously in \cite{molina2012optimal}. The protocol is  as follows: the verifier issues a uniformly random challenge string $c$ to the prover, which encodes the bases (standard or Hadamard) in which the prover should measure the money state; the prover measures the $i$th qubit of the state in the standard basis if $c_i=0$, or in the Hadamard basis if $c_i=1$, and sends all the measurement outcomes as a string $m$ to the verifier; and the verifier checks that, whenever $c_i = \theta_i$, $m_i = x_i$. The property which makes this protocol and these states interesting is that no prover who is given only one copy of the money state can pass verification twice.

Perhaps surprisingly, showing even that this simple protocol is a proof of knowledge according to our definition turns out to be a non-trivial task. We may examine the following illustration of the difficulty. Consider, firstly, the following na\"ive approach to designing an extractor for the protocol described in the preceding paragraph. Recall that, according to our model of `black-box access', the prover can be considered a unitary process; we denote by $U_c$ the unitary that the prover applies to its private register and the message register in response to challenge $c$. The extractor could pick a challenge $c$ at random, apply $U_c$, and then attempt to apply some unitary to the message register to `correct' for the challenge bases in order to recover the original money state. (For example, if $n=4$, and $c=0110$, the extractor could apply the unitary $U_{0110}$, and then apply $H_2 H_3$ to the message register in hopes of recovering the original money state. This strategy would work on the honest prover, who simply measures the real money state in the bases indicated by $c$ in order to obtain its message to the verifier; we may imagine that few meaningful deviations from this pattern are possible.) However, the prover (upon receipt of the challenge) may take its honest money state and decide to apply Pauli $X$ (bit-flip) gates to some arbitrary subset of the qubits of the money state which it was told to measure in the Hadamard basis, and Pauli $Z$ (phase-flip) gates to a subset of the qubits which it was told to measure in the standard basis. If the prover now measures the result in the bases indicated by $c$, it will pass with probability 1---but the state that it measures in the $c$ basis in this scenario is almost certainly not the correct money state. (The exception is when $c = \theta$.)

A little thought will show that this is a fairly general obstacle to the extractor's constructing the money state from the state residing in the prover's message registers immediately before it performs the measurement whose outcomes it will send to the verifier. Since we know very little about what the prover might be doing to the money state at any other stage in its execution, meanwhile, it is difficult to reason about finding the money state in the prover's registers at other points in its operation. This simple argument shows that, in order to design an effective extractor, it is crucial to consider the prover's responses to all challenges $c$ at once---the question, of course, is one of how.

Our way of overcoming these difficulties builds on results from the literature on nonlocal games. The key idea of our security proof for the Wiesner money verification protocol is as follows. Let the party which chooses and prepares the money state $\ket{\$}_{x,\theta} = \prod_i H_i^{\theta_i} \ket{x}$ that the prover receives be known as Alice, and let the prover be known as Bob. Consider the following thought experiment: instead of preparing $\ket{\$}_{x, \theta}$, Alice could prepare $n$ EPR pairs and send half of each one to Bob.
Let $E(\theta) = \{ \proj{\$}_{x, \theta} \: | \: x \in \{0,1\}^{n} \}$ be a POVM. Then, if Alice measures $E(\theta)$ on her side of the state, and obtains the outcome $x$, Alice's and Bob's joint state will collapse to two copies of $\ket{\$}_{x, \theta}$. Note that, from Bob's perspective, the protocol is the same regardless of whether Alice sent EPR pairs and then measured $E(\theta)$, or whether she chose $x$ and $\theta$ uniformly at random and sent him $\ket{\$}_{x, \theta}$ to begin with. However, if Bob succeeds with high probability in the money verification protocol, then he also succeeds with high probability at recovering a subset of the string $x$ which represents Alice's measurement outcomes after she measures the POVM $E(\theta)$, and which also forms part of the classical description of the money state $\ket{\$}_{x, \theta}$. This observation makes it possible to apply a theorem from \cite{nv16} which states that, if two noncommunicating parties exhibit correlations like those which Alice and Bob exhibit in this thought experiment, then they must once have shared EPR pairs, up to local isometry. Since Alice is honest and did nothing to her shares of the EPR pairs, the local isometry on her side is the identity map. Then, in order to recover the original money state, the proof-of-knowledge extractor simply has to execute the correct isometry on Bob's side. This isometry can be implemented efficiently using only black-box access to the prover; this step, however, crucially makes use of the fact that the extractor can implement controlled versions of the prover's unitaries on a superposition of messages of its choice. A detailed analysis is given in section \ref{sec:wiesner}.

Although the efficacy of this technique for showing that a protocol is a proof of knowledge depends strongly on the structure of the Wiesner verification protocol, we are also able to apply it to one other example. Wiesner states were the earliest and are the best-known kind of quantum money states, but there are other kinds, and one sort which has received some recent attention is the class of \emph{subspace states} introduced in a quantum money context by \cite{aaronson2012quantum}. Subspace states are states of the form $\frac{1}{\sqrt{|A|}} \sum_{x \in A} \ket{x}$ for some $n/2$-dimensional subspace $A \in \mathbb{Z}_2^n$, and they have similar no-cloning properties to those of Wiesner states; they are also of additional interest because they have been used in several schemes which make steps toward the goal of public-key quantum money \cite{aaronson2012quantum}, \cite{zhandry2019quantum}, and in constructions of other quantum-cryptographic primitives such as quantum signing tokens \cite{ben2016quantum}. We were not able to find a simple classical verification protocol for subspace states that we could show to be a proof of quantum knowledge. Nonetheless, in Section~\ref{sec:subspaces}, we propose a classical verification protocol for what we call \emph{one-time-padded subspace states} (that is, subspace states which have had random Pauli one-time-pads applied to them by the bank), and we are able to show under our new definition, using similar techniques to those which we applied to Wiesner states, that this simple verification protocol is a proof of knowledge for one-time-padded subspace states. This verification protocol is remarkable for having a challenge from the verifier that is only one bit long.

\medskip

Our final contribution is to show that, under our definition, a classical \emph{argument} of quantum knowledge exists for any relation in the class QMA.\footnote{Argument systems differ from proof systems only in that the honest prover must be efficient, and that soundness is required to hold only against efficient provers. In this case, `efficient' means quantum polynomial-time.} The notion of a \emph{QMA relation} was formalised jointly by \cite{broadbent2019zero} and \cite{coladangelo2019non}, as a quantum analogue to the idea of an \emph{NP relation} which was described in the first paragraphs of this introduction. \cite{broadbent2019zero} and \cite{coladangelo2019non} show that any QMA relation has a \emph{quantum} proof of quantum knowledge. The protocol that we show to be a \emph{classical} argument of quantum knowledge for QMA relations, meanwhile, is the classical verification protocol introduced recently by \cite{measurement}. Mahadev~\cite{measurement} shows, under cryptographic assumptions, that quantum properties (in her case, any language in BQP) can be decided by a classical polynomial-time verifier through classical interaction alone with a quantum polynomial-time prover. We note that the proofs of the main results in~\cite{measurement} include statements which can be used to make the verification protocol which~\cite{measurement} introduces into a classical argument of quantum knowledge in the sense in which we have defined the latter. The main work that needs to be done in order to show this is to establish that the quantum witness, which as shown in~\cite{measurement} always \emph{exists} for the case of a successful prover, can be \emph{extracted} from the prover in a black-box manner. While all the required technical components for establishing this are already present in~\cite{measurement}, we make the statement explicit. (In comparison, our proofs that specific quantum money schemes satisfy our definition of a proof of quantum knowledge do not use any cryptographic assumptions, and the protocols which we consider are very simple compared with the \cite{measurement} protocol.) The \cite{measurement} verification protocol can be shown to be an argument of quantum knowledge for any QMA relation; the only caveat, which was also a caveat for the \emph{quantum} proofs of quantum knowledge for QMA exhibited by \cite{broadbent2019zero} and \cite{coladangelo2019non}, is that an honest prover in the protocol may require multiple copies of a witness in order that the extractor can succeed in extracting \emph{one} copy. We refer the reader to Section~\ref{sec:qma-arguments} for details.

\paragraph{Related and further work.} Unruh~\cite{unruh2012quantum} was one of the first to consider the notion of a proof of knowledge in the quantum setting. In his work, as in the classical literature, a `proof of knowledge' is a classical protocol which aims to establish that a prover `knows' a \emph{classical} string $w$ satisfying an NP relation $R(x, \cdot)$. Unruh's work was novel because it was the first to consider the possibility that an adversarial prover may have quantum capabilities. This makes the design of an extractor more difficult, but Unruh shows that, under specific conditions, a classical proof of knowledge is automatically also sound against quantum adversaries. We use a similar formalism to Unruh's to capture the notion of `black-box access to a quantum prover'.

As we have already mentioned, the notion of a proof of knowledge for a \emph{quantum} relation was introduced recently in~\cite{broadbent2019zero} and~\cite{coladangelo2019non}. In these papers, the authors give a natural extension of the classical definition of a proof of knowledge to proofs of knowledge for \emph{quantum} states $\rho$ that `satisfy' a QMA relation $Q$. (The statement `$\rho$ satisfies $Q$' here means that $Q(\rho,x)\geq \alpha$ for some parameter $\alpha$, where $Q$ is a quantum circuit, $x$ is a classical problem instance, and $Q(\rho,x)$ is the probability that $Q$ accepts witness $\rho$ on input $x$.) These authors also show, building on earlier work by Broadbent et al.~\cite{qma}, that every QMA relation has a (quantum) proof of quantum knowledge. The main protocol from either paper can be made non-interactive (assuming the appropriate trusted setup), but all the protocols considered by these papers necessarily involve the exchange of quantum information between the verifier and the prover. 

It might be possible to show that other constructions for quantum money are also proofs of quantum knowledge; the construction based on the hidden matching problem by Gavinsky~\cite{gavinsky2012quantum} is one such candidate. As we mentioned in the introduction, however, the efficacy of the technique which we apply in sections \ref{sec:wiesner} and \ref{sec:subspaces} depends strongly on the structure of the verification protocols analysed in those sections, and we consider it an interesting problem to find other techniques for showing that protocols are proofs of quantum knowledge. If such techniques can be found, they might perhaps might be used to show how other quantum states in the space between quantum money states and witnesses for arbitrary QMA relations admit natural proofs of quantum knowledge, with or without cryptographic assumptions. Looking toward applications, meanwhile, a natural candidate application for (non-interactive) proofs of quantum knowledge would be turning CPA-secure encryption schemes for quantum data into CCA-secure schemes. (Quantum CCA-secure schemes have already been constructed directly in~\cite{alagic2016computational}.)

It is also natural to consider \emph{zero-knowledge} proofs of quantum knowledge. In our two examples in Section~\ref{sec:money-poqk}, the verifier is provided with secret classical information which completely specifies the state that the prover holds, so the notion of a zero-knowledge proof in this context is meaningless. However, for general QMA relations, and some applications such as quantum identification, the idea of safeguarding the witness state against the verifier becomes more relevant. In prior work~\cite{vidick2019classical} we showed that the protocol introduced in~\cite{measurement} can be made zero-knowledge. Since we show in the current work that~\cite{measurement} is an argument of quantum knowledge for any QMA relation, we believe zero-knowledge classical arguments of quantum knowledge can also be constructed for any QMA relation.  

Classical \emph{non-interactive} arguments for languages in $\QMA$ in the quantum random oracle model were recently introduced in~\cite{chia2019classical,alagic2019non}. Both papers prove their main result by applying the Fiat-Shamir heuristic to Mahadev's protocol. We leave it as an interesting open question whether their results can be adapted to show that their protocols are also proofs of quantum knowledge.

\paragraph{Organisation.} The organisation of this paper is as follows. In section \ref{sec:prelims}, we introduce some preliminary concepts. In section \ref{sec:quantum-aap}, we give our definitions of proofs of quantum knowledge, along with some intuition for our choices. In section \ref{sec:properties}, we prove some elementary properties of the definitions in section \ref{sec:quantum-aap}. In sections \ref{sec:wiesner} and \ref{sec:subspaces}, we give proofs that a classical private-key quantum money verification protocol for Wiesner money states and a classical private-key verification protocol for one-time-padded subspace states, respectively, are proofs of quantum knowledge. Finally, in section \ref{sec:qma-arguments}, we show that any QMA relation has a classical argument of quantum knowledge.

\paragraph{Acknowledgements.} We thank Alexandru Gheorghiu for useful feedback and Or Sattath for comments. Thomas Vidick is supported by NSF CAREER Grant CCF-1553477, AFOSR YIP award number FA9550-16-1-0495, MURI Grant FA9550-18-1-0161 and the IQIM, an NSF Physics Frontiers Center (NSF Grant PHY-1125565) with support of the Gordon and Betty Moore Foundation (GBMF-12500028). This material is based upon work supported by DARPA under Agreement No. HR00112020023. Any opinions, findings and conclusions or recommendations expressed in this material are those of the author(s) and do not necessarily reflect the views of the United States Government or DARPA.

\section{Preliminaries}
\label{sec:prelims}
\subsection{Terminology and notation}
\label{sec:notation}

\paragraph{Running time and asymptotic behaviour} For definitions related to quantum circuits and basic quantum complexity classes such as BQP and QMA we refer to~\cite{watrous2009encyclopedia}. 
We use `PPT' and `QPT' as shorthand for `probabilistic polynomial time' and `quantum polynomial time' respectively. A probabilistic (resp. quantum) polynomial-time procedure is a polynomial-time uniformly generated family of circuits that have a designated `random' input (resp. quantum circuits).
We use the notation $\mathsf{negl}(\lambda)$ to denote any negligible function of $\lambda\in \mathbb{N}$, i.e.\ a function $f$ such that for any polynomial $p$, $p(\lambda)f(\lambda)\to_{\lambda\to\infty} 0$. 

\paragraph{Hilbert spaces and quantum states} We use $\mH$ to denote a finite-dimensional Hilbert space. We sometimes index $\mH$ with a ``register'' label $\reg{A}$, $\reg{B}$ as $\mH_\reg{A}$, $\mH_\reg{B}$ to differentiate Hilbert spaces associated with distinct subsystems. 
A \emph{pure quantum state} is a unit vector $\ket{\psi} \in \mH$.
A \emph{mixed quantum state} is a positive semidefinite operator $\rho$ on $\mH$ with trace $1$. We use $\Density(\mH)$ to denote the collection of all mixed states on $\mH$; they are sometimes also called \emph{density matrices}. Given spaces $\mH_\reg{A}$ and $\mH_{\reg{B}}$ and $\rho \in \Density(\mH_\reg{A} \otimes \mH_\reg{B})$ we sometimes write $\rho_{\reg{AB}}$ for $\rho$ to clarify explicitly what space $\rho$ is a density matrix on, and write e.g. $\rho_\reg{A} = \Tr_\reg{B}(\rho_{\reg{AB}})$ for the partial trace over $\mH_\reg{B}$. Given a privileged basis $\{\ket{i}\}$ for $\mH_\reg{A}$, a \emph{CQ state}, for ``Classical-Quantum'', on $\mH_\reg{A} \otimes \mH_\reg{B}$ is an element $\rho \in \Density(\mH_\reg{A} \otimes \mH_\reg{B})$ of the form $\rho = \sum_i p_i \proj{i} \otimes \rho_i$ for $\{p_i\}_i$ a probability distribution and $\{\rho_i\}_i$ a collection of density matrices on $\mH_\reg{B}$. 

\paragraph{Distances and norms} $|u|_H$ denotes the Hamming weight of a string $u\in\{0,1\}^n$, $\| u \|$ denotes the Euclidean norm of a vector $u\in \mathbb{C}^n$, and $\|u\|_1 = \sum_i |u_i|$ denotes its 1-norm. For a matrix $A \in \C^{n\times d}$, $\|A\|_1 = \Tr\sqrt{A^\dagger A}$, where $A^\dagger$ is the conjugate-transpose, denotes the Schatten $1$-norm (sum of singular values), and $\frac{1}{2}\|\rho - \rho'\|_1$ is the trace distance between two density matrices $\rho$ and $\rho'$. We use the notation $\Drho(A,B)$ to denote the distance (pseudo)metric $\|(A-B) \ket{\psi}\|^2$ between two operators $A$ and $B$ with respect to a specific state $\ket{\psi}$.

\paragraph{Miscellaneous} We use $\circ$ to denote composition: for example, if $F$ and $G$ are two circuits, $F \circ G(x)$ denotes firstly running $G$ on $x$, and then running $F$ on the output of $G$ on $x$. We use the notation $a \equiv b$ to mean `a is defined to be equal to $b$'.

\subsection{Black-box quantum provers}
\label{sec:black-box}

In cryptography the notion of what a machine `knows' is generally formalized through an efficient procedure called the `extractor', which is a procedure that interacts with the machine whose knowledge one aims to establish and eventually succeeds in `extracting' said `knowledge' from it. 
When we define the knowledge of a (classical) `prover' in an interactive protocol, there are two generally accepted models for specifying how the extractor may interact with the prover. 
In the \emph{black-box} model, the extractor is allowed to execute the prover's `next-message' function at unit cost. The prover's `next-message' function takes as input the prover's internal state and the message it receives in any given round, as well (possibly) as some classical randomness, and it returns the prover's new internal state and prover's response to the message it received. The extractor cannot directly access the internal state of the prover nor, generally, its random input: at the beginning of the interaction between prover and extractor, the prover's private register is initialised with some state that is unknown to the extractor, and thereafter the only control the extractor has over this private state is through its power to run the next-message function. The extractor can, however, decide the message which is given as input to the next-message function when it is executed. It is also allowed to `rewind' the prover, i.e.\ it can run the next-message function on some message, obtain the prover's answer, and then backtrack and run it on some other message. 
In the \emph{white-box} model, meanwhile, the extractor is given an explicit circuit which computes the prover's next-message function, and may use this circuit in an arbitrary way to compute a witness.

In a setting where the prover may be quantum the situation is a little more complicated. In this case, it is not reasonable to give the extractor the power to query the prover's next-message function on any message history and then rewind: firstly, the prover's next-message function is inherently probabilistic, since the prover's messages might be the result of quantum measurements; and, secondly, the measurements that produce the prover's messages may irreversibly alter the prover's state, so rewinding may not be possible in a straightforward sense. To circumvent these difficulties, the black-box model for extraction from a quantum prover models the prover as an `interactive machine', which (informally) is specified by a unitary map $U$ acting on a private `internal register' as well as a public `network register' (or `message register') that is used for communication. The unitary $U$ is a stand-in for the `next-message function' from the classical setting. An extractor with black-box access to a quantum prover is allowed to execute the prover's unitary $U$ or its inverse $U^\dagger$ coherently, and although it cannot modify the prover's internal register except by applying $U$ or $U^\dagger$, it is allowed to put any state in the message register, including a state that is a quantum superposition of messages --- even though in the actual execution of any of the protocols considered in this paper we always assume that messages are measured by the classical verifier in the standard basis. As we argued in the introduction, it is necessary to give this power to the extractor in our context (see e.g.~\cite{badertscher2020security} for impossibility results in related settings).

Formally, we use a similar framework to that which is described in~\cite[Section 2.1]{unruh2012quantum} in order to capture black-box access to quantum provers. 
The following definitions of \textit{interactive quantum machines} and \textit{oracle access} to an interactive quantum machine are taken (with some modifications) from \cite{unruh2012quantum}; a similar formulation of these definitions of \cite{unruh2012quantum} appears in \cite{coladangelo2019non}. The modifications which we introduce are primarily for convenience in dealing with the situation where the verifier is known to be classical, instead of (potentially) quantum as it is in \cite{unruh2012quantum}. As we mentioned earlier, even though the possibility is not used in~\cite{unruh2012quantum}, the framework presented there explicitly allows the extractor to coherently implement controlled versions of the prover's unitaries on a superposition of messages from the verifier. 

\paragraph{Interactive quantum machines} An \textit{interactive quantum machine} is a machine $M$ with two quantum registers: a register $\textsf{S}$ for its internal state, and a register $\textsf{N}$ for sending and receiving messages (the network register). Upon activation, $M$ expects in $\textsf{N}$ a message, and in $\textsf{S}$ the state at the end of the previous activation. At the end of the current activation, $\textsf{N}$ contains the outgoing message of $M$, and $\textsf{S}$ contains the new internal state of $M$. A machine $M$ gets as input: a security parameter $\lambda \in \mathbb{N}$, a classical input $x \in \{0,1\}^*$, and a quantum input $\ket{\Phi}$, which is stored in $\textsf{S}$. Formally, machine $M$ is specified by a family of unitary circuits $\{M_{\lambda x}\}_{\lambda \in \mathbb{N}, x \in \{0,1\}^*}$ and a family of integers $\{r^M_{\lambda x}\}_{\lambda \in \mathbb{N}, x \in \{0,1\}^*}$. $M_{\lambda x}$ is the quantum circuit that $M$ performs on the registers $\textsf{S}$ and $\textsf{N}$ upon invocation. $r_{\lambda x}$ determines the total number of messages/invocations. We might omit writing the security parameter and/or the classical input $x$ when they are clear from the context. We say that $M$ is \textit{quantum-polynomial-time} (QPT for short) if the circuit $M_{\lambda x}$ has  size polynomial in $\lambda + |x|$, the description of the circuit is computable in deterministic polynomial time in $\lambda + |x|$ given $\lambda$ and $x$, and $r_{\lambda, x}$ is polynomially bounded in $\lambda$ and $x$.

\paragraph{Oracle access to an interactive quantum machine} We say that a quantum algorithm $A$ has oracle access to an interactive quantum machine $M$ (with internal register $\reg{S}$ and network register $\reg{N}$) \emph{running on $\ket{\Phi}$} to mean the following. We initialise $\reg{S}$ to $\ket{\Phi}$ and $\reg{N}$ to $\ket{0}$, we give $A$ the security parameter $\lambda$ and its own classical input $x$, and we allow $A$ to execute the quantum circuit $M_{\lambda x'}$ (for any $x'$) specifying $M$, and its inverse (recall that these act on the internal register $\textsf{S}$ and on the network register $\textsf{N}$ of $M$). Moreover, we allow $A$ to provide and read messages from $M$ (formally, we allow $A$ to act freely on the network register $\textsf{N}$). We do not allow $A$ to act on the internal register $\textsf{S}$ of $M$, except via $M_{\lambda x'}$ or its inverse.

\paragraph{Interactive classical machines} An \textit{interactive classical machine} is a machine $C$ with two classical registers: a register $\textsf{T}$ for its internal state, and a register $\textsf{N}$ for sending and receiving messages (the network register). Upon activation, $C$ expects in $\textsf{N}$ a message, and in $\textsf{T}$ the state at the end of the previous activation. At the end of the current activation, $\textsf{N}$ contains the outgoing message of $C$, and $\textsf{T}$ contains the new internal state of $M$. A machine $C$ gets as input: a security parameter $\lambda \in \mathbb{N}$, a classical input $x \in \{0,1\}^*$, a random input $u \in \{0,1\}^{p(\lambda + |x|)}$ for some function $p \in \mathbb{N}$, and a classical auxiliary input $t \in \{0,1\}^{|\mathsf{T}|}$, which is stored in $\reg{T}$. Formally, machine $C$ is specified by a function $p \in \mathbb{N}$, a family of classical circuits $\{C_{\lambda x u}\}_{\lambda \in \mathbb{N}, x \in \{0,1\}^*, u \in \{0,1\}^{p(\lambda + |x|)}}$ and a family of integers $\{r^C_{\lambda x}\}_{\lambda \in \mathbb{N}, x \in \{0,1\}^*}$. $C_{\lambda x u}$ is the classical circuit that $C$ performs on the registers $\textsf{T}$ and $\textsf{N}$ upon invocation. Without loss of generality, for convenience's sake, we assume that $C_{\lambda x u}$ is reversible. $r^C_{\lambda x}$ determines the total number of messages/invocations. We might omit writing the security parameter and/or the input when they are clear from the context. We say that $C$ is \textit{probabilistic-polynomial-time} (PPT for short) if $p$ is a polynomial, the circuit $C_{\lambda x u}$ has  size polynomial in $\lambda + |x|$, the description of the circuit is computable in deterministic polynomial time in $\lambda + |x|$ given $\lambda$, $x$ and $u$, and $r^C_{\lambda x}$ is polynomially bounded in $\lambda$ and $x$.

\paragraph{Oracle access to an interactive classical machine} We say that a quantum algorithm $A$ has oracle access to an interactive classical machine $C$ \emph{running on string $t$} to mean the following. We initialise $C$'s internal register $\reg{T}$ to $t$ and the network register $\reg{N}$ to the all-zero string. We give $A$ the security parameter and its own classical input $x$. Each time $A$ wishes to run $C$ (or its inverse), it must submit an input $x'$ on which to run $C$ (or its inverse). Upon receiving $A$'s choice of $x'$, we choose $u$ uniformly at random, and then we run the classical circuit $C_{\lambda x' u}$ (or its inverse); recall that these act on the internal register $\textsf{T}$ and on the network register $\textsf{N}$ of $C$. Moreover, we allow $A$ to provide and read messages from $C$ (formally, we allow $A$ to act freely on the network register $\textsf{N}$). We do not allow $A$ to act on the internal register $\textsf{T}$ of $C$, except via $C_{\lambda x' u}$ or its inverse.


\begin{definition}
We use the terminology \emph{interactive Turing machine} (ITM) to refer to either an interactive classical machine or an interactive quantum machine. If the ITM is bounded-time, we may refer to a PPT ITM or a QPT ITM to clarify which model is used. An interactive oracle machine is an ITM that in addition has query access to an oracle. 
\end{definition}

\paragraph{Interaction between an interactive quantum machine and an interactive classical machine}
Let $M = ( \{ M_{\lambda x} \}, \{ r^M_{\lambda x} \} )$ be an interactive quantum machine with internal register $\textsf{S}$ and network register $\textsf{N}$. Let $C = ( \{ p, C_{\lambda x' u} \}, \{ r^C_{\lambda x'} \} )$ be an interactive classical machine with internal register $\textsf{T}$ and network register $\textsf{N}$. For a given CQ state $\rho_{\reg{TS}}\in \Density(\mH_\reg{T}\otimes \mH_\reg{S})$, we define the interaction $\left (C(x'), M(x) \right )_{\rho_{\reg{TS}}}$ as the following quantum process:
initialize register $\textsf{N}$ to $\ket{0}$; initialise registers $\textsf{S}$ and $\textsf{T}$ to the CQ state $\rho_{\reg{TS}}$; alternately apply $M_{\lambda x}$ to registers $\textsf{S}$ and $\textsf{N}$ and $C_{\lambda x' u}$ (for a uniformly chosen $u \in \{0,1\}^{p(\lambda + |x'|)}$ each time) to registers $\textsf{T}$ and $\textsf{N}$, measuring $\textsf{N}$ in the computational basis after each application of either $M_{\lambda x}$ or $C_{\lambda x' u}$; stop applying $M_{\lambda x}$ after $r^M_{\lambda x}$ times and $C_{\lambda x' u}$ after $r^C_{\lambda x'}$ times, and finally output the output of the circuit $C_{\lambda x' u}$. We denote the random variable representing this output by $\left < C(x'), M(x) \right >_{\rho_{\reg{TS}}}$. We call the $r^M_{\lambda x} + r^C_{\lambda x'}$ measurement outcomes which are obtained after performing as many standard basis measurements of $\textsf{N}$ during a single execution of the interaction $\left (C(x'), M(x) \right )_{\rho_{\reg{TS}}}$ the \emph{transcript} for this execution.

\subsection{Implementing oracles}
\label{sec:oracles}

Some of our formal definitions rely on `oracles', which we generally visualise as functions $\mO:\{0,1\}^*\to\{0,1\}^*$ to which query access is given. In this paper, query access is always provided by some trusted procedure (which is called the `setup functionality' in Section~\ref{sec:quantum-aap}). We remark here briefly on how query access to these oracles (which, expressed as functions, may take an exponential number of bits to specify) can be implemented efficiently in the number of queries made to the oracle.

We always assume that our oracles are classical oracles, meaning that any query submitted to an oracle is measured in the standard basis before being answered. (However, the oracles are allowed to return quantum states.) With this requirement in place, a length-preserving random oracle (for example) can be implemented efficiently by a QPT procedure that (i) measures its query, (ii) looks up in a local database if the query has been answered before, and if so provides the answer, (iii) generates a random answer and adds it to the database if not. A more structured oracle with classical queries, such as one that returns a quantum state that can be efficiently re-generated given an appropriate classical description, can be implemented in a similar way. In case one wishes to efficiently implement a quantum random oracle one could use the ``compressed oracle'' method of Zhandry~\cite{zhandry2019record}, but we will not need this here.

\section{Quantum Agree-and-Prove schemes}
\label{sec:quantum-aap}

To define the intuitive notion of a `proof of quantum knowledge' in sufficient generality so that we can capture both quantum money verification and $\QMA$ verification we introduce a quantum variant of the `agree-and-prove' framework from~\cite{badertscher2019agree}, extending their formalism to our setting in which the prover and the witness are quantum, and simplifying some aspects of the formalism that are less important for the applications we have in mind. For convenience, we preserve much of the notation from~\cite{badertscher2019agree}; we refer the reader to that paper for additional motivation and explanations relating to the framework. 

Informally, to specify an agree-and-prove (AaP) scheme, one must specify both a \emph{scenario} and a \emph{protocol}. The scenario includes a (trusted) \emph{setup functionality} which captures the \emph{environment} in which the protocol will take place (such as the existence of a common random string, or an oracle accessible to all parties, etc.). In addition, the scenario includes an \emph{agreement relation} and a \emph{proof relation}, which are both efficient procedures that define the set of \emph{valid instances} (in the former case) and the set of \emph{valid witnesses} for any given instance (in the latter). That for any (yes) instance $x$ there is a set of valid witnesses is a familiar notion from the usual proof-of-knowledge formalism; the idea that there is also a set of valid instances is a less familiar concept, and the agreement relation exists to capture situations where the prover must \emph{approach} the verifier and suggest a problem instance for which it knows a witness, and moreover not all problem instances are allowed (for example, a particular protocol might only be secure for instances of a particular form, or there might be a database of valid user IDs, in which case the agreement relation would check that any ID for which the prover proposed to prove it had a password was in fact a valid ID).

The second part of the specification of an agree-and-prove scheme is the description of a \emph{protocol}, which specifies intended actions for the honest prover and honest verifier. These honest actions should satisfy a \emph{completeness} condition, whereas the \emph{soundness} condition applies to arbitrary actions for the prover.

In the next subsection we formalise the notion of a scenario. The following section discusses input generation algorithms; the one after that formalises protocols, and the one after that lays down the security conditions for agree-and-prove schemes.

\subsection{Scenario}

\begin{definition}[Agree-and-Prove Scenario for quantum relations]
An agree-and-prove (AaP for short) scenario for quantum relations is a triple $(\mF,\mR,\mC)$ of interactive oracle machines satisfying the following conditions: 
\begin{itemize}
\item The \emph{setup functionality} $\mF$ is a QPT ITM taking a unary encoding of a security parameter $\lambda$ as input. The ITM $\mF$ runs an initialization procedure \texttt{init}, and in addition returns the specification of an oracle (which we also model as an ITM) $\mO_\mF(i,\texttt{q},arg)$. The oracle function takes three arguments: $i\in \{I,P,V\}$ denotes a `role', $\texttt{q}$ denotes a keyword specifying a query type, and $arg$ denotes the argument for the query.

There are three different options for the `role' parameter, which exists to allow $\mF$ to release information selectively depending on the party asking for it. The roles $I$, $P$ and $V$ correspond respectively to the \emph{input generator} (Definition \ref{def:input-gen}), the prover, and the verifier.
\begin{remark}
In \cite{badertscher2019agree}, $\mO_\mF$ has an additional function: when it is called with the argument $\mathtt{QUERIES}$, $\mO_\mF(\mathtt{QUERIES})$ returns a list of tuples representing all of the queries made to $\mO_\mF$ by the prover $P$ and the replies that were given. This functionality is available only to the extractor, not to the parties $I$, $P$ and $V$, and it is necessary in order to permit the design of an efficient extractor for some protocols, particularly those in the random oracle model (see, for example, the discussion at the bottom of page 10 in \cite{badertscher2019agree}). Since we do not need to use this functionality in our protocols, we omit it here.
\end{remark}

\item The \emph{agreement relation} $\mC$ is a QPT oracle machine taking a unary encoding of the security parameter $\lambda$ and a statement as inputs, and producing a decision bit as output.\footnote{In~\cite{badertscher2019agree} the agreement relation also takes two auxiliary inputs. We will not need this.}
\item The \emph{proof relation} $\mR$ is a QPT oracle machine taking a unary encoding of the security parameter $\lambda$, a (classical) statement $x$ and a (quantum) witness $\rho_\reg{W}$ as inputs, and outputting a decision bit.
\end{itemize}
\end{definition}

\subsection{Input generation}

Before we formalise the notion of an agree-and-prove protocol, we introduce the notion of an \emph{input generation algorithm}, which is an algorithm that produces the auxiliary inputs that the prover and the verifier receive before they begin interacting. The input generation algorithm models `prior knowledge' which the prover and the verifier may possess. When the prover and the verifier are honest, we assume that they share a prescribed amount of prior knowledge that allows them to successfully carry out the protocol (for example, the prover should know the user ID for which it proposes to prove it has a password); when the prover is potentially dishonest, however, which is the case when we define the soundness experiment, the input generation algorithm can be arbitrary (and in particular is not required to run in polynomial time). Allowing an arbitrary input generation algorithm in the soundness experiment is important for composability.

\begin{remark}
Here we depart from~\cite{badertscher2019agree}, where input generation is always unrestricted (even when the verifier and the prover are honest). For the completeness experiment (Definition \ref{def:completeness}) to succeed in our setting we will sometimes require that the prover's and the verifier's auxiliary inputs are honestly generated: for example, in the context of QMA verification, we require that the prover is given as input $(x, \rho_W)$ such that $\rho_W$ is a valid witness for $x$. The prover cannot check for itself whether or not $\rho_W$ is a valid witness without potentially damaging it irreversibly. However, the verifier will reject the prover in the completeness experiment if it does not hold a valid witness. As such, in order that completeness can be satisfied, we require that the prover's input is honestly generated during the completeness experiment.
\end{remark}

\begin{definition}[Input Generation Algorithm]
\label{def:input-gen}
An input generation algorithm $I$ for an agree-and-prove scenario $\mathcal{S}$ is a machine $I$ taking a unary encoding of the security parameter $\lambda$ as input and producing a CQ state $\rho_{\reg{AUX}_V \reg{AUX}_P}$ specifying the auxiliary inputs for the verifier (in the classical register $\reg{AUX}_V$) and prover (in the quantum register $\reg{AUX}_P$) respectively as output. We may use the shorthand $\rho_{\reg{AUX}_P} \equiv \Tr_{\reg{AUX}_V}\big( \rho_{\reg{AUX}_V \reg{AUX}_P} \big)$ and $\rho_{\reg{AUX}_V} \equiv \Tr_{\reg{AUX}_P}\big( \rho_{\reg{AUX}_V \reg{AUX}_P} \big)$.
\end{definition}


\subsection{Protocol}

Once a scenario has been fixed we can define a \emph{protocol} for that scenario. Informally, the protocol specifies the actions of the honest parties. Each party, prover and verifier, is decomposed into two entities that correspond to the two phases, ``agree'' and ``prove'', of the protocol.

\begin{definition}[Agree-and-prove protocol]
An agree-and-prove protocol is a tuple $(\mI, P_1,P_2,V_1,V_2)$ consisting of a set $\mI$ of input generation algorithms together with the following four interactive oracle machines $(P_1,P_2,V_1,V_2)$:
\begin{itemize}
\item A (honest) first phase QPT prover $P_1$ taking a unary encoding of the security parameter $\lambda$ and a (quantum) auxiliary input $\rho_{\reg{AUX}_P}$ as inputs. It produces a (classical) statement $x_P$ or $\perp$ as output, as well as a (quantum) state $\rho_{st_P}$.
\item A (honest) first phase PPT verifier $V_1$ taking a unary encoding of the security parameter $\lambda$ and a (classical) auxiliary input $\reg{AUX}_V$ as inputs. It produces a (classical) statement $x_V$ or $\perp$ as output, as well as a (classical) state $st_V$.
\item A (honest) second phase QPT prover $P_2$ taking a classical instance $x$ and a quantum state $\rho_{st_P}$ as input, as well as a unary encoding of the security parameter $\lambda$, and producing as output a bit that indicates whether the proof has been accepted.
\item A (honest) second phase PPT verifier $V_2$ taking a classical instance $x$ and a state string $st_V$ as input, as well as a unary encoding of the security parameter $\lambda$, and producing as output a bit that indicates whether it accepts or rejects.
\end{itemize}
\end{definition}

Note that in this definition the verifier is required to be a classical probabilistic polynomial time ITM. In general one may extend the definition to allow for quantum polynomial time verifiers; since our focus is on classical protocols we restrict our attention to classical verifiers. We also restrict the honest prover to run in quantum polynomial time; for soundness, this restriction will be lifted for the case of proofs of knowledge and maintained for the case of arguments of knowledge. 
\subsection{Security conditions}
\label{sec:security}

%
%

We now specify the correctness and soundness conditions associated with an agree-and-prove scenario~$\mathcal{S}$.

\begin{definition}[Completeness experiment]
\label{def:completeness}
We define the following \emph{completeness experiment} for an agree-and-prove protocol $\mK = (\mI, P_1, P_2, V_1, V_2)$ in the context of a scenario $\mS = (\mF, \mC, \mR)$:
\begin{enumerate}
\item An input generation algorithm $I \in \mI$ is executed. It is allowed to query $\mO_{\mF}(I, \cdot, \cdot)$. It produces the CQ state $\rho_{\mathsf{AUX}_V \mathsf{AUX}_P}$, and passes input $\rho_{\mathsf{AUX}_P}$ to $P_1$ and $\rho_{\mathsf{AUX}_V}$ to $V_1$.
\item The interaction $(V_1, P_1)_{\rho_{\reg{AUX}_V\reg{AUX}_P}}$ is executed (during which $V_1$ and $P_1$ are allowed to query $\mO_{\mF}(V, \cdot, \cdot)$ and $\mO_{\mF}(P, \cdot, \cdot)$, respectively), and if either $V_1$ or $P_1$ returns $\perp$, or if $x_V \neq x_P$, the agree phase returns 0. Otherwise, the outputs of $V_1$ and $P_1$ are passed to $V_2$ and $P_2$, respectively, and the agree phase returns 1. If the agree phase returns 1, let the CQ state representing the joint distribution of $st_V$ and $\rho_{st_P}$ be denoted by $\rho_{st_V \: st_P}$, and let $x = x_P = x_V$ be the instance that $V_1$ and $P_1$ have agreed on.
\item The interaction $(V_2(x), P_2(x))_{\rho_{st_V \: st_P}}$ is executed (during which $V_2$ and $P_2$ are allowed to query $\mO_{\mF}(V, \cdot, \cdot)$ and $\mO_{\mF}(P, \cdot, \cdot)$, respectively), and the outcome of the proof phase is set to the value which $V_2$ returns at the end of the protocol.
\end{enumerate}

The completeness experiment returns 1 if the agree phase and the proof phase both return 1.
\end{definition}


\paragraph{Soundness.} The soundness experiment is formulated using the notion of an extractor. We highlight some differences with the classical case which is related in \cite{badertscher2019agree}. First of all, naturally the extractor is allowed to be a QPT procedure. Its goal is to produce a valid witness for the statement $x$ that is produced at the end of the setup phase. (If the honest ${V}_1$ and a malicious $\hat{P}_1$ do not produce a matching statement, then the extractor automatically `wins' the soundness experiment.) As in the classical setting, the extractor has access to $x$, the communication transcript from the setup phase (that cannot be rewound), as well as oracle access to the dishonest strategy $\hat{P}_2$ as explained in Section~\ref{sec:black-box}. 

A second difference is with respect to how success of the experiment is estimated. In the classical case, once the agreement phase has completed the agreement relation $\mR$ and statement $x$ are fixed, so that a witness $w$ is either correct or incorrect. In the quantum case, $\rho_\reg{W}$ is a quantum state and $\mR$ a QPT verification procedure. Therefore, success of the extractor is probabilistic: conditioned on extraction producing a witness, this witness has a certain probability of being deemed valid by $\mR$. For this reason the soundness condition contains an additional parameter $\delta$ that quantifies this last probability. 

\begin{definition}[Soundness experiment]
\label{def:soundness}
We define the following \emph{soundness experiment} for an agree-and-prove protocol $\mK = (\mI, P_1, P_2, V_1, V_2)$ and an extractor $E$,  in the context of a scenario $\mS = (\mF, \mC, \mR)$:
\begin{enumerate}
\item An input generation algorithm $\hat I$ is executed.  It is allowed to query $\mO_{\mF}(I, \cdot, \cdot)$. It produces the CQ state $\rho_{\mathsf{AUX}_V \mathsf{AUX}_P}$, and passes input $\rho_{\mathsf{AUX}_P}$ to $\hat P_1$ and $\rho_{\mathsf{AUX}_V}$ to $V_1$.
\item The interaction $(V_1, \hat P_1)_{\rho_{\reg{AUX}_V\reg{AUX}_P}}$ is executed (during which $V_1$ and $\hat P_1$ are allowed to query $\mO_{\mF}(V, \cdot, \cdot)$ and $\mO_{\mF}(P, \cdot, \cdot)$, respectively), and if either $V_1$ or $P_1$ returns $\perp$, or if $x_V \neq x_P$, the agree phase returns 0. Otherwise, the outputs of $V_1$ and $\hat P_1$ are passed to $V_2$ and $\hat P_2$, respectively, and the agree phase returns 1. If the agree phase returns 1, let the CQ state representing the joint distribution of $st_V$ and $\rho_{st_P}$ be denoted by $\rho_{st_V \: st_P}$, and let $x = x_P = x_V$ be the instance that $V_1$ and $\hat P_1$ have agreed on.
\item If the agree phase returns 1 in step 2, the extractor $E$ is provided with the transcript of the interaction $(V_1, \hat P_1)_{\rho_{\reg{AUX}_V\reg{AUX}_P}}$ and the instance $x$ resulting from the agree phase, along with oracle access to $\hat P_2$ running on input $\rho_{st_P}$ (where $\rho_{st_P}$ is the prover's half of the joint CQ state $\rho_{st_V \: st_P}$). In addition the extractor can access the oracle $\mO_\mF$ using any of the roles in $\{I,P\}$.   It outputs a state $\rho$.
\end{enumerate}
\end{definition}

We are now ready to give the formal definition of security. 

\begin{definition}[Security of Protocol for Quantum Agree-and-Prove Scenario]
\label{def:aap-security}
Let $\lambda$ be a security parameter. Let $c,\kappa,\delta:\mathbb{N}\to [0,1]$. A protocol $\mK = (\mI, P_1, V_1, P_2, V_2)$ for a scenario $(\mF, \mC, \mR)$ is secure with completeness $c$, up to knowledge error $\kappa$, and with extraction distance parameter $\delta$ if the following conditions hold:
\begin{itemize}
\item \emph{Correctness:} The completeness experiment (Definition \ref{def:completeness}) returns $1$ with probability at least $c$, and in addition the statement $x = x_V = x_P$ that is agreed on during the completeness experiment is such that $\mC(1^\lambda, x) = 1$, whenever the honest parties $P$ and $V$ are provided with their inputs by some input generation algorithm $I\in\mI$.
\item \emph{Soundness:} There exists a QPT ITM $E$ (called the ``extractor'') such that the following holds. Let $\hat{P}=(\hat{P}_1,\hat{P}_2)$ be a potentially dishonest prover for $\mK$ and $\hat{I}$ an arbitrary input generation algorithm. 
 Let $x$ be an instance such that, conditioned on the agree phase of $\mK$ returning 1 and the instance $x$ being agreed upon, the prover $\hat{P}_2$ succeeds with probability $p > \kappa$ in the proof phase of $\mK$. 
Then the state $\rho$ returned by the extractor in the soundness experiment (Definition \ref{def:soundness}), conditioned on the agree phase of the soundness experiment returning 1 and $x$ being agreed on, is such that $\mathrm{Pr}[\mR(1^\lambda, x, \rho) = 1] > 1 -\delta(p)$, where $\delta$, which may depend on $\lambda$, is such that $\delta(p)<1$ for all $p>\kappa$. The expected number of steps of extractor $E$ is required to be bounded by a polynomial in $\lambda/(p-\kappa)$, if executing the prover's unitary on any input counts as a unit-time procedure.
\end{itemize}
\end{definition}

When the soundness condition only holds under the restriction that $\hat{P}$ must be implemented by a QPT ITM we say that the protocol is \emph{computationally secure}, or that it is an \emph{argument system} (as opposed to a \emph{proof system}, which is sound against all possible provers).

\begin{remark}
When we wish to emphasize the connection between secure agree-and-prove protocols and the more usual notion of a `proof of knowledge', we sometimes refer to an AaP scenario that satisfies Definition~\ref{def:aap-security} as a `classical proof (or argument) of quantum knowledge'. (Formally, proofs and arguments of knowledge can be formulated as protocols for AaP scenarios which have trivial agreement phases and which have as a proof relation an NP  or a QMA relation; see Section~\ref{sec:qma-arguments}.) When we use this terminology, it will be clear from context what the `knowledge' is that we are referring to.
\end{remark}

\subsection{Amplification}

The proofs of quantum knowledge that we construct in the paper have a knowledge error that is a constant close to $1$. In general one would aim for a very small constant, or better, a knowledge error that is negligible in the security parameter. 

A natural avenue to improving the soundness error is to consider sequential or parallel repetition. We will show that each of the explicit protocols which we consider here has a natural sequential repetition that reduces its soundness error exponentially with the number of repetitions. We stop short of giving a general sequential amplification result for quantum Agree-and-Prove scenarios for the following technical reason. The drawback of repetition is that, as argued in Section~\ref{sec:nondestructive}, proofs of quantum knowledge with classical communication are generally destructive; therefore, the honest prover in a sequentially repeated protocol needs to make use of multiple copies of the witness to the relation in order to succeed. In contrast, the extractor will only be able to recover a single witness. Formulating general conditions on the input generation procedure so that it has a natural ``amplification'' that returns a single witness but multiple proofs for it seems cumbersome in general, and so we have elected to discuss sequential repetition of each of our protocols individually. 

Parallel amplification seems much more delicate. For the case of quantum proofs of classical knowledge Unruh was able to show that the knowledge error decreases under parallel repetition by making use of very specific requirements on the protocol, and in particular the condition of \emph{special soundness}.\footnote{Informally this condition states that a valid witness can be extracted from correct responses to two different challenges.}. Neither of our proof of quantum knowledge protocols satisfies this condition. For the case of the protocols for quantum money it seems plausible that amplification in parallel would be possible, but we do not show this. For the case of the argument of knowledge for $\QMA$ based on the Mahadev protocol, although soundness amplification in parallel has been shown in~\cite{alagic2019non,chia2019classical}, we do not know how the knowledge error could be similarly amplified. The question seems delicate and we leave it to further work.  

\subsection{Agree-and-Prove scenario for quantum money}
\label{sec:ap-scenario-qm}

As a first example of a concrete agree-and-prove scenario, we define an agree-and-prove scenario that captures the scenario which arises in the problem of verifying quantum money. We firstly lay down the `standard' security definitions for a quantum money scheme, and in so doing introduce some notation and some objects that will be useful in formulating quantum money in the agree-and-prove framework.

\begin{definition}\label{def:qmoney}
A ``quantum money scheme'' is specified by the following objects, each of which is parametrized by a security parameter $\lambda$:
\begin{itemize}
\item A algorithm $\mathtt{Bank}$ taking a string $r$ as a parameter which initialises a database of valid \emph{money bills} in the form of a table of tuples $(\mathsf{id, public, secret}, \ket{\$}_{\mathsf{id}})$. \textsf{id} represents a unique identifier for a particular money bill; \textsf{public} and \textsf{secret} represent, respectively, public and secret information that may be necessary to run the verification procedure for the bill labeled by \textsf{id}; and $\ket{\$}_{\mathsf{id}}$ is the quantum money state associated with the identifier \textsf{id}. The string $r$ should determine a classical map $H_r$ such that $H_r(\mathsf{id}) = (\mathsf{public, secret})$.\footnote{The string $r$ represents any random choices that $\mathtt{Bank}$ might make while generating valid bills; we make this string explicit for later convenience.}

\item A verification procedure $\texttt{Ver}(x, \textsf{public}, \textsf{secret}, \rho_W)$ that is a QPT algorithm which decides when a bill is valid.
\end{itemize}
In addition the scheme should satisfy the following conditions:
\begin{enumerate}
\item Completeness: for any valid money bill $(\mathsf{id, public, secret}, \ket{\$}_{\mathsf{id}})$ in the database created by $\mathtt{Bank}$,
\[\mathrm{Pr}\big( \texttt{Ver}(\reg{id}, \textsf{public}, \textsf{secret}, \proj{\$}_{\mathsf{id}})\big) \,\geq \, c_M(\lambda)\;,\]
for some function $c_M(\cdot)$.
We refer to $c_M$ as the \emph{completeness parameter} of the money scheme. 
\item No-cloning: 
Consider the following game played between a challenger and an adversary: the challenger selects a valid money bill $(\mathsf{id, public, secret}, \ket{\$}_{\mathsf{id}})$ and sends $(\mathsf{id, public,} \ket{\$}_{\mathsf{id}})$ to the adversary; the adversary produces a state $\sigma_{AB}$. Then for 
any~\footnote{Many quantum money schemes are information-theoretically secure; however, it is also possible to consider computationally secure schemes by replacing `any' with `any QPT'.} adversary in this game,
\begin{align*}
\underset{r}{\mathrm{Pr}} \big(
\texttt{Ver}(\reg{id}, \textsf{public}, \textsf{secret}, \Tr_B(\sigma_{AB})) = 1
\;\text{ and }\;
\texttt{Ver}(\reg{id}, \textsf{public}, \textsf{secret}, \Tr_A(\sigma_{AB})) = 1
\big) \,\leq\, \mu_M(\lambda)\;,
\end{align*}
for some function $\mu_M(\cdot)$. 
We refer to $\mu_M$ as the \emph{cloning parameter} of the money scheme. Note that the probability of the adversary's success is calculated assuming that the string $r$ which \texttt{Bank} takes is chosen uniformly at random.
\end{enumerate}
\end{definition}

Fix a quantum money scheme according to Definition~\ref{def:qmoney}, with completeness parameter $c_M$ and cloning parameter $\mu_M$. 
We call an agree-and-prove scenario $(\mF_M, \mC_M, \mR_M)$ that takes the form below a `quantum money scenario with completeness parameter $c_M$ and cloning parameter $\mu_M$'.

\begin{itemize}
\item Setup functionality $\mF_M(1^\lambda)$: The setup should run an initialization procedure $\mathtt{init}_M$ that instantiates\footnote{$\texttt{init}_M$ doesn't necessarily need to actually allocate memory for the database; since the database will only ever be accessed through the oracle $\mO_{\mF_M}$, it is possible to `instantiate' the database using the method described in Section \ref{sec:oracles}.} a database $B_M$ whose records are of the form (and the distribution) that $\mathtt{Bank}$ would have produced running on a uniformly random input $r$. The setup should also return a specification of how the following oracles should be implemented: 
\begin{itemize}
\item $\mathcal{O}_{\mF_M}(I, \texttt{id})$: returns an identifier $\mathsf{id}$ such that the bill $(\mathsf{id, public, secret}, \ket{\$}_{\mathsf{id}})$ is in $B_M$.\footnote{Which identifier is returned is at the discretion of any particular instantiation of this function. Intuitively, this oracle is used to represent identifiers of bills that have been generated in the past and are thus available in an ``environment'' that $I$ may have access to.}
\item $\mathcal{O}_{\mF_M}(\cdot, \texttt{public}, \mathsf{id})$: Returns the \textsf{public} string associated with \textsf{id}. Returns $\perp$ if no record in $B_M$ with the identifier \textsf{id} exists.
\item $\mathcal{O}_{\mF_M}(I, \texttt{getMoney}, \mathsf{id})$: If no record in $B_M$ with identifier \textsf{id} exists, returns $\perp$. Otherwise, returns $\ket{\$}_{\mathsf{id}}$ the first time it is called. If called again with the same \textsf{id} argument, returns $\perp$.
\item $\mathcal{O}_{\mF_M}(V, \texttt{secret}, \mathsf{id})$: accesses $B_M$ and returns the \textsf{secret} string associated with \textsf{id}. Returns $\perp$ if no record in $B_M$ with the identifier \textsf{id} exists.
\end{itemize}
\item Agreement relation $\mathcal{C}^{\mathcal{O}_{\mF_M}}(1^\lambda, \mathsf{id})$: outputs 1 if and only if a record in $B_M$ with identifier $\mathsf{id}$ exists.
\item Proof relation $\mathcal{R}^{\mathcal{O}_{\mF_M}}(1^\lambda, x, \rho_W)$: interprets $x$ as an \textsf{id} (outputting $\perp$ if this fails), sets $\textsf{public} \leftarrow \mathcal{O}_{\mF_M}(V, \texttt{public}, x)$ and $\textsf{secret} \leftarrow \mathcal{O}_{\mF_M}(V, \texttt{secret}, x)$, and executes $\texttt{Ver}(x, \textsf{public}, \textsf{secret}, \rho_W)$. 
\end{itemize}

\section{Simple properties}
\label{sec:properties}

In this section, we state and prove two simple properties of our definitions in section \ref{sec:quantum-aap}. The first of these, given in Section~\ref{sec:nondestructive}, is that, if it is possible to \emph{nondestructively} succeed in a secure agree-and-prove protocol for some scenario $(\mF, \mC, \mR)$, then there is a `cloning procedure' that can produce copies of a state which will be accepted by $\mR$. This is a simple no-go theorem which precludes certain types of classical agree-and-prove protocols in both the quantum money and the QMA settings. The other simple property which we prove is that, under certain assumptions on the parameters in Definition \ref{def:aap-security}, any protocol satisfying Definition \ref{def:aap-security} for the case of a quantum money scenario (see Section~\ref{sec:ap-scenario-qm} for a definition of the latter) can be used as a quantum money verification protocol for the same scenario, provided a generalized notion of interactive verification is allowed. This is shown in Section~\ref{sec:money-ver}. 

\subsection{Nondestructive proofs of quantum knowledge imply cloning}
\label{sec:nondestructive}

In this section we formalize the intuitive claim that a non-destructive proof of knowledge must imply the ability to clone the underlying witness state. To formalize this statement we make a number of assumptions that help simplify the presentation. More general statements can be proven depending on one's needs; see the end of the section for further discussion. 

We use definitions and notation from Section \ref{sec:black-box} and Section \ref{sec:quantum-aap}.

\begin{definition}[Nondestructive interaction]
\label{def:nondestructive-interaction}
Let $P = (\{P_{\lambda x}\}, \{r^{P}_{\lambda x}\})$ be an interactive quantum machine, and let $V = (p, \{V_{\lambda x u}\}, \{r^V_{\lambda x}\})$ be an interactive classical machine. Fix a security parameter $\lambda$. A \emph{nondestructive interaction} $\left (V(x) , P(x') \right )_{\rho_{\reg{TS}}}$ between $V$ and $P$ for some CQ state $\rho_{\reg{TS}}$ is an interaction in which the execution of $\left (V(x) , P(x') \right )_{\rho_{\reg{TS}}}$ is unitary (including the standard-basis measurements of the network register that take place during the execution) for all possible random inputs $u$ to $V$. More formally, for any choice of $r^V_{\lambda x}$ random strings $u_1, \dots, u_{r^V_{\lambda x}}$ used during the interaction $\left (V(x) , P(x') \right )_{\rho_{\reg{TS}}}$, there exists a unitary $U$ acting on registers $\reg{N}$, $\reg{T}$ and $\reg{S}$ such that the joint state of the registers $\reg{N}$, $\reg{T}$ and $\reg{S}$ is identical after $U$ has been applied to them (assuming they are initialised as described in Section \ref{sec:black-box}) to their joint state after the execution of $\left (V(x) , P(x') \right )_{\rho_{\reg{TS}}}$ using the random strings $u_1, \dots, u_{r^V_{\lambda x}}$.

\end{definition}

\begin{definition}[Oracle access to an interactive quantum machine with power of initialisation]
Recall the definition of oracle access to an interactive quantum machine given in Section \ref{sec:black-box}. In that section, the initial state $\ket{\Phi}$ on which the quantum machine is run is fixed. We say that a quantum algorithm $A$ has \emph{oracle access to an interactive quantum machine $M$ with power of initialisation} if $A$ can do all the things described in Section \ref{sec:black-box}, and in addition can initialise $M$'s internal register $\reg{S}$ to a state of its choosing (but cannot read $\reg{S}$, only write to it).
\end{definition}


\begin{proposition}
\label{prop:no-nondestructive}
Let $\lambda$ be a security parameter, let $(\mF, \mC, \mR)$ be an agree-and-prove scenario, and let $\mK = (\mI, P_1, P_2, V_1, V_2)$ be a protocol for $(\mF, \mC, \mR)$ with a classical honest verifier $V = (V_1, V_2)$, knowledge error $\kappa$ and extraction distance $\delta$. Let $\hat P = (\hat{P}_1,\hat{P}_2)$ be a prover for $\mK$.

Let $\hat{I}$ be any input generation algorithm, and $x$ and $\rho_{\reg{TS}}$ an instance and a CQ state respectively such that the agree phase of $\mK$, executed with $\hat{I}$, $V_1$ and $\hat{P}_1$, has positive probability of ending with $x$ being agreed on, and such that the joint state of $st_V$ and $st_P$ conditioned on $x$ being agreed on is $\rho_{\reg{TS}}$.

Suppose further that (i) the interaction $\left (V_2(x) , \hat P_2(x) \right )_{\rho_{\reg{TS}}}$ is nondestructive, (ii) the oracle $\mO_\mF$ does not keep state during the second phase of the protocol, i.e.\ any query to it by $V_2$ or $\hat{P}_2$ can be repeated with the same input-output behavior, and (iii) the success probability of $\hat P_2$ conditioned on instance $x$ being agreed on is at least $\kappa$. Then there exists a procedure $A$ \footnote{$A$ is in general not efficient. It is also allowed slightly more invasive access to $\hat P_2$ than a typical extractor (because it has power of initialisation for $\hat P_2$). This is acceptable because $A$ is not an extractor, but a cloning procedure. No-cloning security in quantum money scenarios is often information-theoretic, and in these cases it does not matter whether $A$ is efficient; its mode of access to the prover is also irrelevant as long as a circuit for $A$ exists (which is the case here if a circuit for $\hat P_2$ exists). On the other hand, we want the way that $A$ accesses the verifier $V$ to reflect the access that a cloning adversary might have to whichever party is running verification in real life (in the quantum money case, the bank). Our definition of `oracle access to a classical machine' reflects this wish.} such that the following holds. 

Let $\mR^{\mO_\mF}_{\lambda x}(\cdot)$ be the function such that $\mR^{\mO_\mF}_{\lambda x}(\rho) = \mR^{\mO_\mF}(1^\lambda, x, \rho)$, and let the single-bit-valued function $\big(\mR^{\mO_\mF}_{\lambda x}\big)^{\otimes 2}(\cdot)$ be the function whose output is the $\mathrm{AND}$ of the outcomes obtained by executing the tensor product of two copies of $\mR^{\mO_\mF}_{\lambda x}(\cdot)$ on the state that is given as argument. Then the procedure $A$, given as input $x$, a copy of a communication transcript from the agree phase that led to $x$, and black-box access to $V_2$ and $\hat{P}_2$ as interactive machines (including any calls they might make to $\mO_\mF$) running on $\rho_{\reg{TS}}$, with power of initialisation for $\hat P_2$, can produce a state $\sigma$ such that
\begin{gather}
\label{eq:A-success-prob}
\mathrm{Pr}[\big(\mR^{\mO_\mF}_{\lambda x}\big)^{\otimes 2}(\sigma) = 1] > 1 - 2\delta - \mathsf{negl}(\lambda).
\end{gather}
\end{proposition}

\begin{proof}
In each round of the proof phase, the prover $\hat P_2$ performs a unitary $\hat P_{2, \lambda x}$ on the registers $\textsf{S}$ and $\textsf{N}$ ($\textsf{S}$ is the prover's internal register, and $\textsf{N}$ is the network register). After this unitary has been performed, the network register contains a string $r_i$ in the computational basis which represents the response the prover sends to the verifier in the $i$th round.

Let the initial joint state that $V_2$ and $\hat P_2$ share be $\rho_{\reg{TS}} = \sum_i p_i \proj{t_i}_{\reg{T}} \otimes \rho_{\reg{S},i}$ for classical strings $t_i$. We can without loss of generality assume that $\rho_{\reg{S},i} = \proj{\Phi_i}$ for some pure state $\ket{\Phi_i}$, since any mixed initial state that resides in the prover's $\reg{S}$ register can be purified by extending the $\reg{S}$ register to some larger register $\reg{S}'$. Note that (according to Definition \ref{def:nondestructive-interaction}), if $\left (V_2(x) , \hat P_2(x) \right )_{\rho_{\reg{TS}}}$ is nondestructive for $\rho_{\reg{TS}} = \sum_i p_i \proj{t_i}_{\reg{T}} \otimes \proj{\Phi_i}_{\reg{S}}$, then $\left (V_2(x) , \hat P_2(x) \right )_{\ket{t_i}_{\reg{T}} \otimes \ket{\Phi_i}_{\reg{S}}}$ is nondestructive for any $i$.

We use the following claim as a stepping stone. For now (i.e. until we indicate otherwise), we assume that the prover's initial state is a pure state $\ket{\Phi}$, and also that the verifier's \textsf{T} register contains a single fixed string $t$. At the end of this proof we show how the conclusions extend to the general case by a linearity argument.

\begin{claim}
\label{claim:nondestructive-subclaim}
Suppose that $\left (V_2(x) , \hat P_2(x) \right )_{\ket{t}_{\reg{T}} \otimes \ket{\Phi}_{\reg{S}}}$ is nondestructive. There exists a procedure $B$ such that, if $B$ knows, for every round $i$, the classical response $r_i$ that the prover $\hat P_2$ would provide in round $i$ in response to every possible series of messages $m_1, \dots, m_i$ that the honest verifier could send in rounds 1 to $i$, and if $B$ is given oracle access to $\hat P_2$ with power of initialisation, $B$ can construct a state $\ket{\Phi'}$ such that 
\[\langle V_2(x) , \hat P_2(x) \rangle_{\ket{t}_{\reg{T}} \otimes \ket{\Phi'}_{\reg{S}}} = \langle V_2(x) , \hat P_2(x) \rangle_{\ket{t}_{\reg{T}} \otimes \ket{\Phi}_{\reg{S}}}\]
without needing access to $\ket{\Phi}$.
\end{claim}

\begin{proof}
Note that, because all of the measurements performed during the execution $\left (V_2(x) , \hat P_2(x) \right )_{\ket{t}_{\reg{T}} \otimes \ket{\Phi}_{\reg{S}}}$ have deterministic outcomes if $\left (V_2(x) , \hat P_2(x) \right )_{\ket{t}_{\reg{T}} \otimes \ket{\Phi}_{\reg{S}}}$ is nondestructive, the outcomes $r_i$ are entirely determined by the verifier's messages $m_i$ and the prover's initial state $\ket{\Phi}$. This means that, fixing $\ket{\Phi}$, we can write $r_i$ as a function of $m_1, \dots, m_i$, $r_i = r_i(m_1, \dots, m_i)$, which we will abbreviate to $r_i(m_{[1,i]})$. Suppose there are $n$ \footnote{Note that, using the notation of Section \ref{sec:black-box}, $n = \max(r_{\lambda x}^{V_2}, r_{\lambda x}^{\hat P_2})$.} rounds in the protocol $\mK$; then, we denote the concatenation of the functions $r_i(m_{[1,i]})$ for all $i \in \{1,\dots,n\}$ by $r(m)$, where $m = m_1, \cdots, m_n$. We refer to this function $r(m)$ as the \emph{response signature} of $\hat P_2$ running on $\ket{\Phi}$.

Observe that, if $\ket{\Phi'}$ is a quantum state such that $\hat P_2$ running on $\ket{\Phi'}$ has an identical response signature to $\hat P_2$ running on $\ket{\Phi}$, then
\[\langle V_2(x) , \hat P_2(x) \rangle_{\ket{t}_{\reg{T}} \otimes \ket{\Phi'}_{\reg{S}}} = \langle V_2(x) , \hat P_2(x) \rangle_{\ket{t}_{\reg{T}} \otimes \ket{\Phi}_{\reg{S}}}.\]

The claim follows: $B$ is the procedure that takes as input a response signature $r(m)$ and uses its powers of initialisation for $\hat P_2$ to look for a state $\ket{\Phi'}$ such that $\hat P_2$ running on $\ket{\Phi'}$ has a response signature of $r(m)$. (Such a state always exists because $\ket{\Phi}$ exists.)

\end{proof}

We now describe how the cloning adversary $A$ proceeds, assuming still that the  initial state which $V_2$ and $\hat P_2$ share is of the pure form $\ket{t}_\reg{T} \otimes \ket{\Phi}_\reg{S}$. Fix the instance $x$ for which $\left (V_2(x) , \hat P_2(x) \right )_{\ket{t}_\reg{T} \otimes \ket{\Phi}_\reg{S}}$ is nondestructive. 
$A$ does the following:
\begin{enumerate}
\item Plays the prove phase against $V_2$, behaving exactly as $\hat P_2$ would behave, and records every message $m_i$ that the verifier sends and every response $r_i$ that the prover sends.
\item Repeats step 1 until the probability that it has \emph{not} seen every possible sequence of messages $m = m_1, \dots, m_n$ from the verifier, and the prover's responses to each such sequence, is negligibly small. If the verifier is polynomially bounded, then this can be achieved (crudely) by running step 1 a doubly exponential number of times. Naturally, if a more precise bound is known on the verifier's runtime, or (even better) if the verifier's challenges are sampled from a publicly known set (which is true, for example, for the verifiers of both protocols which we consider in Section \ref{sec:money-poqk}), $A$'s running time can be improved.

(Note that $A$ can repeat step 2 as many times as it likes because every execution of the protocol is fully unitary, so $A$ can always recover the prover's initial state from its final state and play again. This part also uses that $\mO_\mF$ does not keep state and hence cannot `detect' that multiple prove phases are being played in sequence.)

From here on, we assume that $A$ succeeds in learning the prover's responses to all of the verifier's challenges in this step. We account for the probability that $A$ might fail in this step by subtracting a $\reg{negl}(\lambda)$ from $A$'s total success probability in equation \eqref{eq:A-success-prob}.
\item 
$A$ makes two copies of the function $r(m)$ which encodes the prover's responses to every possible sequence of verifier messages, and applies the procedure $B$ guaranteed by Claim \ref{claim:nondestructive-subclaim} to both copies. Recall that $B$'s output on a response signature $r(m)$ is a state $\ket{\Phi'}$, where $\hat P_2$ running on $\ket{\Phi'}$ has the response signature $r(m)$.

$A$ runs the extractor twice in parallel, using the two states output by the two parallel executions of $B$, as well as its power of initialisation for $\hat P_2$, to implement oracle access to $\hat P_2$ running on $\ket{\Phi'}$ for both copies of the extractor.
Let $\sigma$ denote the joint state output by the two copies of $E$ executed in parallel under these conditions. (Note that $\sigma = \sigma_1 \otimes \sigma_2$ is a product state, because the two executions of the procedure are independent; moreover, $\sigma_1 = \sigma_2$.) The output of $A$ is $\sigma$.
\end{enumerate}

Now we return to the regime where the initial state that $\hat V_2$ and $P_2$ share, $\rho_\reg{TS}$, is a general mixed state $\rho_\reg{TS} = \sum_i p_i \proj{t_i}_\reg{T} \otimes \proj{\Phi_i}_\reg{S}$. We argue that the output $\sigma$ of $A$ in this case is such that
\begin{gather*}
\mathrm{Pr}[\big(\mR^{\mO_\mF}_{\lambda x}\big)^{\otimes 2}(\sigma) = 1] > 1 - 2\delta,
\end{gather*}
as stated in Proposition \ref{prop:no-nondestructive}.

We use the shorthand $E ^{\hat P_2}(\proj{\Phi_i})$ to denote the output of $E$ when it is run with oracle access to $\hat P_2$ running on initial state $\ket{\Phi_i}$. By linearity, the output of $A$ when the initial joint state of $V_2$ and $\hat P_2$ is $\rho_{\reg{TS}} = \sum_i p_i \proj{t_i}_\reg{T} \otimes \proj{\Phi_i}_\reg{S}$ is of the form
\begin{equation*}
\sigma = \sum_i p_i \: E^{\hat P_2}(\proj{\Phi_i'}) \otimes E^{\hat P_2}(\proj{\Phi_i'}) \;,
\end{equation*}
where the states $\ket{\Phi_i'}$ are such that $\langle V_2(x) , \hat P_2(x) \rangle_{\ket{t_i}_{\reg{T}} \otimes \ket{\Phi_i'}_{\reg{S}}} = \langle V_2(x) , \hat P_2(x) \rangle_{\ket{t_i}_{\reg{T}} \otimes \ket{\Phi_i}_{\reg{S}}}$ for each $i$.

For notational convenience, we define (for each $i$) $\sigma_i \equiv E^{\hat P_2}(\proj{\Phi_i'}) \otimes E^{\hat P_2}(\proj{\Phi_i'})$, so that $\sigma = \sum_i p_i \sigma_i$. We also define $\sigma_{i,1}$ and $\sigma_{i,2}$ to be the two states respectively represented by the first and second terms in the tensor product $\sigma_i = E^{\hat P_2}(\proj{\Phi_i'}) \otimes E^{\hat P_2}(\proj{\Phi_i'})$, so that $\sigma_i = \sigma_{i,1} \otimes \sigma_{i,2}$.

Since linearity also implies that $\langle V_2(x) , \hat P_2(x) \rangle_{\rho_\reg{TS}} = \langle V_2(x) , \hat P_2(x) \rangle_{\rho'_\reg{TS}}$, where $\rho'_\reg{TS} = \sum_i p_i \proj{t_i}_\reg{T} \otimes \proj{\Phi_i'}_\reg{S}$, and since we have by assumption that $\langle V_2(x) , \hat P_2(x) \rangle_{\rho_\reg{TS}} > \kappa$, it follows that $\langle V_2(x) , \hat P_2(x) \rangle_{\rho'_\reg{TS}} > \kappa$. The definition of the extractor $E$ from Definition \ref{def:aap-security} gives us
\begin{equation*}
\mathrm{Pr}\Big[ \mR^{\mO_\mF}_{\lambda x} \Big( E^{\hat P_2} \Big( \sum_i p_i \proj{\Phi_i'} \Big) \Big) = 1\Big] > 1 - \delta\;.
\end{equation*}
By linearity,
\begin{equation*}
\sum_i p_i \: \mathrm{Pr}\Big[ \mR^{\mO_\mF}_{\lambda x} \Big( E^{\hat P_2} ( \proj{\Phi_i'} ) \Big) = 1\Big] > 1 - \delta\;.
\end{equation*}
Hence, for $j \in \{1,2\}$,
\begin{equation*}
\sum_i p_i \: \mathrm{Pr}\Big[ \mR^{\mO_\mF}_{\lambda x} ( \sigma_{i,j} ) = 1\Big] > 1 - \delta\;.
\end{equation*}
Then, by a union bound,
\begin{align*}
\sum_i \,p_i\,\mathrm{Pr}[\big(\mR^{\mO_\mF}_{\lambda x}\big)^{\otimes 2}(\sigma_i) = 1] = \mathrm{Pr}[\big(\mR^{\mO_\mF}_{\lambda x}\big)^{\otimes 2}(\sigma) = 1] &> 1 - 2\delta\;,
\end{align*}
as desired. 
\end{proof}

\subsubsection*{Discussion}

The most easily apparent application of Proposition \ref{prop:no-nondestructive} is that it precludes the existence of a \emph{perfectly} nondestructive, information-theoretically secure agree-and-prove protocol with a classical verifier for a quantum money scenario. Even though the statement of the lemma makes multiple restrictive assumptions, it is still applicable to natural scenarios such as the scenarios for quantum money given in Section~\ref{sec:money-poqk}.\footnote{One of these assumptions is that the setup oracle $\mO_\mF$ for the scenario $(\mF, \mC, \mR)$ does not keep state during the prove phase of the protocol execution (meaning that any query made to it by $\hat P_2$ or $V_2$ can be repeated with identical input-output behaviour). The reader might observe that the quantum money setup oracles $\mO_{\mF_M}$ (as we have defined them in Section \ref{sec:ap-scenario-qm}) do keep state, because they need to keep track of the number of money bills in circulation with a given identifier $\reg{id}$ (so that they do not give out clones of the same bill). However, these oracles only keep state with respect to $I$, the input generation algorithm, and not with respect to $\hat P$ or $V$, because only $I$ can access the function \texttt{getMoney}. Therefore, the assumptions of Proposition \ref{prop:no-nondestructive} are satisfied by the quantum money scenarios which we analyse in Sections \ref{sec:wiesner} and \ref{sec:subspaces}.}

\paragraph{Noninteractive classical proofs of quantum knowledge}
We remark that, as a corollary, Proposition \ref{prop:no-nondestructive} also precludes a \emph{noninteractive}, information-theoretically secure agree-and-prove protocol with a classical verifier for a quantum money scenario. (When we say `noninteractive' here, we mean that the honest verifier $V_2$ sends no messages in the prove phase of the protocol.) This is because, if a protocol is noninteractive in this sense, then there is a prover $\hat P$ for which the interaction between $\hat P_2$ and $V_2$ is automatically nondestructive. We compare this briefly to the more familiar classical case of proofs of knowledge for NP relations. While truly noninteractive \emph{zero-knowledge} proofs of knowledge for all NP relations are also impossible \cite{goldreich1994imposs}, in the NP case it is always trivially possible to define a noninteractive proof of knowledge in which the prover simply sends the witness to the verifier. Moreover, it is possible to construct noninteractive zero-knowledge proofs of knowledge for NP relations in which $V_2$ sends no messages if a trusted setup assumption, such as a common random string \cite{blum1991nizk} or a random oracle, is allowed. In the quantum money setting, we have shown that any kind of noninteractive classical proof of quantum knowledge is impossible (zero-knowledge or not), and moreover that adding a trusted setup cannot help: interaction in the prove phase is necessary in order to preserve security.

\paragraph{A robust version of Proposition \ref{prop:no-nondestructive}}
It is natural to ask if Proposition \ref{prop:no-nondestructive} can be extended to the case where the protocol is not perfectly nondestructive but has the property that each execution only damages the prover's state by a small amount. It is also natural to ask how far the claim extends to protocols with computational security.

\noindent \emph{Notation. \:\:} Before we initiate a discussion of these subjects, we establish some notation for convenience. Suppose there are $n$ rounds in the protocol $\mK$. Let $m_1, \dots, m_n$ be the $n$ messages sent by the verifier in rounds 1 to $n$. We reinterpret the unitary $\hat P_{2, \lambda x}$ that is applied to registers $\textsf{S}$ and $\textsf{N}$ by the prover in round $i$, together with the standard basis measurement of the network register that follows it, as a unitary $U_{m_i}$ on register \textsf{S} followed by the application of a projective (Hermitian)\footnote{This measurement can also be taken to be a standard basis measurement; we use this notation only to emphasise the projectivity of $M_{m_i}$.} measurement operator $M_{m_i}$ on \textsf{S} that produces a classical measurement outcome which the prover then copies into \textsf{N}. Note that, in a perfectly nondestructive protocol, the state which results from applying $U_{m_1}$ to the prover's initial state $\ket{\Phi}$ (assuming for simplicity of presentation, as we did in part of the proof of Proposition \ref{prop:no-nondestructive}, that the prover's initial state is pure), lies in the eigenspace of $M_{m_1}$ with eigenvalue $r_1(m_1)$, where $r_1$ is the prover's (deterministic) message in round 1. Similarly, $U_{m_2} U_{m_1} \ket{\Phi}$ lies in the eigenspace of $M_{m_2}$ with eigenvalue $r_2(m_1, m_2)$, and so on. $\hat P_2$ running on a state $\ket{\Phi'}$ which satisfies this property for all $n$ measurement operators $M_{m_1}, \dots, M_{m_n}$, for all possible sequences $m_1, \dots, m_n$ output by the honest verifier and their corresponding responses $r_1(m_1), \dots, r_n(m_1, \dots, m_n)$, will have the same response signature $r(m)$ as $\hat P_2$ running on $\ket{\Phi}$. We could therefore have formulated the objective of the procedure $B$ from Claim \ref{claim:nondestructive-subclaim} as that of finding a state $\ket{\Phi'}$ of the description given in the previous sentence. \\

\noindent \emph{Obstacles. \:\:} To consider formulating a `robust' version of Proposition \ref{prop:no-nondestructive} (which is applicable to a protocol where each execution only damages the prover's state by a small amount), it is necessary to overcome two obstacles. Firstly, the prover's state prior to the measurement it performs in the $i$th round will no longer be an exact eigenstate of the measurement operator $M_{m_i}$; and, secondly, the cloning adversary may not be able to repeat the protocol an arbitrary number of times in order to learn the prover's likely response to each verifier message, because the prover's state could become `worn out' before the cloning adversary has seen every possible sequence of messages from the verifier. The first obstacle can be overcome by observing that, for all $i$, the prover's pre-measurement state in round $i$ is an \emph{approximate} eigenstate of the operator $M_{m_i}$. Given that a state exists (namely $\ket{\Phi}$, the prover's initial state) such that $U_{m_1} \ket{\Phi}$ is an approximate eigenstate of $M_{m_1}$ with eigenvalue $r_1$ (where $r_1$ is the prover's most probable response to $m_1$), $U_{m_2} U_{m_1} \ket{\Phi}$ is an approximate eigenstate of $M_{m_1}$ with eigenvalue $r_2$, and so on, an all-powerful cloning adversary should be able to find this state $\ket{\Phi}$ (or a state satisfying the same criteria, if more than one exists) if it knows $r(m)$, and if the prover's pre-measurement states are sufficiently close to true eigenstates. That portion of the proof then goes through almost as before, with some additional parameters that capture how far the prover's pre-measurement states are from true eigenstates of the operators $M_{m_i}$. (Of course, a quantitative analysis would be more delicate, and we omit it. See e.g.~\cite{harrow2017sequential} for a possible starting point.)

The second obstacle can be overcome if the cloning adversary will see all possible sequences of verifier messages (or at least all possible sequences that occur with high probability) before the prover's state becomes worn out. This can occur if the min-entropy of the verifier's message sequence is very small compared to $\log (1/\epsilon)$, where $\epsilon$ is the `amount' by which the prover's state is damaged with each execution. As such, we anticipate that a `robust' version of Proposition \ref{prop:no-nondestructive} exists for any `damage parameter' $\epsilon$ that is much smaller than $2^{-H}$, where $H$ is the min-entropy of the verifier's message sequence. Note that this means that the min-entropy of the verifier's challenges in a protocol would determine some parameter capturing the minimal amount of damage that any prover succeeding in the protocol must do to its initial state. For example, in a money verification protocol like the one we analyse in Section \ref{sec:subspaces}, in which the verifier's only message is a single bit long, it is impossible to pass without completely destroying the money state. This is one of the properties of subspace verification that \cite{ben2016quantum} make use of in showing that the subspace money scheme can be modified into a quantum tokens scheme, in which tokens must not only be uncloneable but also equipped with a `verifiable destruction' procedure that both destroys a token and produces a certificate which can be shown afterwards to prove that the token has indeed been consumed and is no longer useable. \\

\noindent \emph{Reducing the min-entropy of the verifier's challenges. \:\:} We note that, if the prover is computationally bounded, it is possible to reduce the min-entropy of the verifier's challenges using a PRG. More specifically, suppose there is a protocol $\mK$ for a scenario $(\mF, \mC, \mR)$ for which a prover $\hat P$ is a nondestructive prover. Suppose, further, that $\hat P$ is computationally bounded. Then there exists a protocol $\mK'$ for the same scenario in which the honest verifier $V'$ chooses its random challenges using a PRG, but otherwise behaves the same way as the honest verifier $V$ for $\mK$. $\hat P$ is also a nondestructive prover for $\mK'$, since it (being computationally bounded) cannot distinguish between the two protocols. By the argument set forth in the previous two paragraphs, the interaction $(\hat P, V')$ must be destructive to some degree which is lower-bounded by a function of the min-entropy of verifier messages in $\mK'$. The prover's operations must be destructive to the same degree in $\mK$ (since otherwise it could distinguish between the two protocols by playing verification many times and noting when the verifier starts to reject). It is then possible to increase the (moral) lower bound on the `damage parameter' (alternatively, the moral upper bound on the damage parameter below which the impossibility proof goes through) to $2^{-H^\nu}$ for any constant $\nu > 0$, where $H$ is the min-entropy of verifier challenges in $\mK$. This is because, for any constant $\nu > 0$, it is possible to build an efficiently evaluable PRG (from appropriate assumptions) which takes the uniform distribution over strings of length $H^\nu$ to a distribution computationally indistinguishable from the uniform distribution over strings of length $H$, the latter of which can in turn be used by the verifier to sample from the distribution from which its challenges are drawn. Note that, in this scenario, the cloning adversary $A$ which a nondestructive prover implies is still in general not efficient; as such, this argument is only of interest in the case where the witness states are inherently uncloneable even to unbounded adversaries (like Wiesner states, for example) but we are only interested in ruling out \emph{efficient} nondestructive provers for $\mK$.

\paragraph{Application of Proposition \ref{prop:no-nondestructive} to QMA}
We also briefly discuss the extension of Proposition \ref{prop:no-nondestructive} to the case of verifying QMA. In such a setting, it is not interesting to state that an unbounded adversary is able to produce arbitrarily many witness states, since an unbounded party can generate witnesses for any QMA yes-instance given only the instance. However, the cloning adversary $A$ is efficient if the min-entropy of the verifier's message sequence is only logarithmic instead of polynomial. (If the verifier chooses its messages uniformly at random, a logarithmic min-entropy means there are only polynomially many possible message sequences the verifier could choose from.) As such, we conclude that, unless BQP = QMA, nondestructive secure agree-and-prove protocols for verifying QMA must have verifiers which issue a superpolynomial number of possible challenges.

\subsection{Proofs of quantum knowledge are also quantum money verification protocols}
\label{sec:money-ver}

The main result in this section is Proposition \ref{prop:pok-implies-money}, which formalises the intuition that the property of being a `proof of quantum knowledge' for a quantum money AaP scenario is stronger than the property of being a quantum money verification protocol: the latter implies that no adversary can pass verification twice given access to only one money bill, and the former formalises the notion that no adversary can pass even once unless it is possible to efficiently compute the money bill by interacting with said adversary.

\paragraph{Formalising interactive quantum money verification.}
Before we can state Proposition \ref{prop:pok-implies-money}, which intuitively says that any agree-and-prove protocol for a quantum money scenario is also a quantum money verification protocol for that scenario, we must formalise what it means to `be a quantum money verification protocol'. There is a standard definition of this for \emph{passive} verification procedures, but we must consider how we will generalise it to interactive procedures.

The standard quantum money security definition is given in Definition \ref{def:qmoney}. Implicit in that definition is the assumption that the verification algorithm \texttt{Ver} is an isometry having a single designated output qubit; we will call such verification procedures \emph{passive} verification procedures. In the below, we will use the notation $\mathtt{Ver}_\reg{id}(\rho)$ as a shorthand for the measurement outcome obtained by applying the isometry $\mathtt{Ver}$ to $\rho$ (fixing the first argument to $\mathtt{Ver}$ to $\reg{id}$, and omitting all other arguments to $\mathtt{Ver}$ except the last) and then measuring the designated output qubit.

\begin{remark}
\label{remark:passive-nondestructive}
Passive verification procedures often have the property that they perturb the real money state $\ket{\$}_{\mathsf{id}}$ only a negligible amount, so that verification can be performed many times on the same state (which is a requirement for public-key quantum money). For example, the verification procedure originally proposed by Wiesner for Wiesner's quantum money (which we describe at the start of Section \ref{sec:wiesner}) is nondestructive, as is the oracle-based verification which Aaronson and Christiano \cite{aaronson2012quantum} propose for their subspace-based scheme.
\end{remark}

In words, we can formulate a `no-cloning experiment' for a passive verification procedure $\mathtt{Ver}$ as follows: the prover sends a state $\sigma_{AB}$ on $n$-qubit registers $A$ and $B$ to the verifier; the verifier applies $\mathtt{Ver}_{\reg{id}} \otimes \mathtt{Ver}_{\reg{id}}$, and accepts if and only if both copies of $\mathtt{Ver}_{\reg{id}}$ accept. Under the standard definition, $\mathtt{Ver}$ is said to be a `quantum money verification procedure' if this no-cloning experiment can be won only with probability less than the cloning parameter $\mu_M$. \footnote{This is implicit in condition 2 of Definition \ref{def:qmoney}.} We now introduce a natural generalisation of this definition to the case where the verification procedure may involve interaction.

\begin{definition}[Interactive quantum money verification procedure]
\label{def:interactive-money-protocol}
Let $\lambda$ be a security parameter, and let $(\mF_M, \mC_M, \mR_M)$ be a quantum money scenario (as defined in Section \ref{sec:ap-scenario-qm}). A protocol $\mK = (I, P_1, P_2, V_1, V_2)$ for $(\mF_M, \mC_M, \mR_M)$ (see Definition \ref{def:aap-security}) is an \emph{interactive verification procedure} with completeness $c$ and cloning error $s$ for the quantum money scenario $(\mF_M, \mC_M, \mR_M)$ if the following two conditions hold.\footnote{This definition is distinct from the definition of security of a protocol $\mK$ described in Definition~\ref{def:aap-security}. While the latter is a general security definition that can apply to any AaP scenario, the present definition is a new definition tailored to quantum money that is a natural extension of the standard ``no-cloning''-based definition recalled in Section~\ref{def:qmoney}. Our aim in this section, in fact, is to show that (qualitatively speaking) Definition \ref{def:aap-security} implies Definition \ref{def:interactive-money-protocol}.}
 (Probabilities in these conditions are calculated assuming that $r$, the randomness that $\mathtt{Bank}$ takes as input, is drawn from the uniform distribution. See Definition \ref{def:qmoney} for a definition of \texttt{Bank}.)

\begin{enumerate}
\item Completeness: The protocol $\mK$ has completeness $c$ according to Definition \ref{def:aap-security}.

\item Soundness: let $\hat P_A = (\hat P_{A,1}, \hat P_{A,2})$ and $\hat P_B = (\hat P_{B,1}, \hat P_{B,2})$ be two provers for $\mK$, and let $\hat I$ be an algorithm that generates inputs for both of them. We define a no-cloning game as follows:

\begin{enumerate}
\item $\hat I$ prepares a (potentially entangled) joint state $\rho_{AB}$. During this phase, $\hat I$ is allowed to call the oracle $\mO_{\mF_M}$ using the role $I$. At the end of this phase, $\hat I$ gives $\rho_A = \Tr_B(\rho_{AB})$ to $\hat P_A$, and $\rho_B = \Tr_A(\rho_{AB})$ to $\hat P_B$.
\item Holding $\rho_A$, $\hat P_A$ executes $\mK$ with a copy of the honest verifier of $\mK$, the latter of which we denote by $V_A = (V_{A,1}, V_{A,2})$. Likewise, holding $\rho_B$, $\hat P_B$ also executes $\mK$ with a copy of the honest verifier of $\mK$, which we denote by $V_B = (V_{B,1}, V_{B,2})$. During the protocol executions, $\hat P_A$ and $\hat P_B$ are not allowed to communicate, but they are allowed to call the oracle $\mO_{\mF}$ using the role $P$.
\item $\hat P_A$ and $\hat P_B$ win the game if and only if the same instance $x$ is agreed upon in the agree phases of both copies of $\mK$ played in step 2, and in addition both $V_A$ and $V_B$ output 1 at the end of the game.
\end{enumerate}
We say that the protocol $\mK$ for the quantum money scenario $(\mF_M, \mC_M, \mR_M)$ is \emph{secure against cloning} with cloning error $s$ if any pair of provers $(\hat P_A, \hat P_B)$ with any input generation algorithm $\hat I$ wins the no-cloning game with probability less than $s$.
\end{enumerate}
\end{definition}

\begin{remark}
We comment on the need for the no-communication assumption in Definition \ref{def:interactive-money-protocol}. To illustrate the reason for it, we consider implementing a passive verification procedure $\mathtt{Ver}_{\reg{id}}$ within the framework of Definition \ref{def:interactive-money-protocol}. In this case we could define $\mK$ as follows: $\mK$ is simply the protocol where the honest prover sends the state $\ket{\$}_{\mathsf{id}}$ to the verifier, and the verifier: executes $\mathtt{Ver}_{\reg{id}}$ on it; measures the designated output qubit of $\mathtt{Ver}_{\reg{id}}$ and records the outcome; runs $\mathtt{Ver}_{\reg{id}}$ backwards on the post-measurement state and sends the result back to the prover; and finally outputs the measurement outcome it recorded earlier. (Ordinarily, it is assumed that the verifier does not return the post-measurement state to the prover after executing a verification procedure like $\mathtt{Ver}_{\reg{id}}$; however, as we will explain in the paragraph after next, having the verifier return the prover's state creates a closer analogy with the classical verification case.)

Note that, in this special case, the no-communication assumption in the soundness experiment is what guarantees that the state $\sigma_A$ which $\hat P_A$ submits to its verifier and the state $\sigma_B$ which $\hat P_B$ submits to its verifier exist in different Hilbert spaces. If the two provers are allowed to pass messages, then the soundness statement does not capture what we would wish it to capture: $\hat P_A$ could submit $\ket{\$}_{\mathsf{id}}$ to its verifier first, and then (assuming that verification does not disturb $\ket{\$}_{\mathsf{id}}$ by much---which, as we noted in Remark \ref{remark:passive-nondestructive}, is often true for passive verification procedures) give the post-measured version of $\ket{\$}_{\mathsf{id}}$ to $\hat P_B$, who would then use it in its interaction with its own verifier. This strategy wins with high probability but does not involve attempting to `clone' $\ket{\$}_{\mathsf{id}}$ in the sense in which we usually think of cloning.

As we have mentioned, it is usually assumed that the verifier does not return the post-measurement state to the prover after executing a verification procedure like $\mathtt{Ver}_{\reg{id}}$, because for some money schemes doing so creates a security risk, due to attacks modelled on the Elitzur-Vaidman bomb test \cite{elitzur1993quantum}. However, when the verification procedure is totally classical, the prover never sends the money bill (or a forgery of one) to the verifier: the prover measures the bill for itself, and so the verifier does not have the option of `not returning' the state. In all of the classical verification protocols that we know of, a genuine quantum money state becomes near totally corrupted after classical verification, which serves the same purpose as does the verifier refusing to return the state (i.e. to make sure the prover cannot use the state for a second verification, so that the message-passing and the no-communication models considered in the previous paragraph are equivalent). However, we cannot guarantee that this is true for all possible classical verification protocols, and so we adopt the no-communication model in order to capture the requirement that the prover does not `reuse' its state by playing the two verification games sequentially. \znote{note to myself: on second thoughts, though, shouldn't all classical quantum money protocols be destructive, because of the Elitzur-Vaidman attack? Our `nondestructive lemma' for PoQKs (which are stronger) only works right now for perfectly nondestructive protocols, and we think it works for protocols that damage the state by an exponentially or subexponentially small amount each time, but it doesn't say anything about inverse-poly damage in general. I wonder if we could do any better for QM protocols specifically, or maybe also for PoQKs in general, using ideas from Elitzur-Vaidman.}
\end{remark}

We are now ready to formally present Proposition~\ref{prop:pok-implies-money}, which captures the fact that a secure agree-and-prove protocol can be used as a quantum money verification procedure. Note that, in Proposition \ref{prop:pok-implies-money} below, and in Proposition \ref{prop:money-amp} after that, the condition on $\delta$ (the extraction distance parameter of $\mK$) can informally be thought of as a requirement that, for all provers passing with sufficiently high probability in $\mK$, the extraction distance which the extractor can achieve when extracting from such provers is bounded above away from some constant $\delta_0$. The reason for this requirement can intuitively be understood through the observation that the states $\sigma_A$ and $\sigma_B$ which result from running the extractor twice on the two provers $\hat P_A$ and $\hat P_B$ from the no-cloning experiment of Definition \ref{def:interactive-money-protocol} (in order to show a reduction from a pair of provers $\hat P_A$ and $\hat P_B$ that pass with high probability in said experiment to a cloning adversary) might be entangled, and it could be the case that, even though the marginal probability of $\sigma_A$ or $\sigma_B$ passing verification is high, the event that $\sigma_B$ passes conditioned on $\sigma_A$ passing is very unlikely (or vice versa). Bounding $\delta$ above by a constant ensures that the fidelity between either one of the extractor's outputs ($\sigma_A$ and $\sigma_B$) and a real, unentangled money bill is high enough that this is not a critical issue.

\begin{proposition}
\label{prop:pok-implies-money}
Let $\lambda$ be a security parameter, and let $(\mF_M, \mC_M, \mR_M)$ be a quantum money scenario (as defined in Section \ref{sec:ap-scenario-qm}). Let $\mathcal{K} = (I, P_1, P_2, V_1, V_2)$ be a protocol for a quantum money agree-and-prove scenario $(\mF_M, \mC_M, \mR_M)$. Let $\mu_M$ be the cloning parameter for the quantum money scenario $(\mF_M, \mC_M, \mR_M)$.

Define $\delta_0 \equiv \frac{2-\sqrt{3}}{2}$. Suppose there is a function $\kappa(\cdot)$ such that $\mathcal{K}$ is a $(c=1 - \mathsf{negl}(\lambda), \delta)$--secure protocol with knowledge error $\kappa(\lambda)$ and extraction distance $\delta$ (the latter of which we assume is a function of the prover's success probability $p$ as well as the security parameter $\lambda$) such that $\delta_0 - \delta(p(\lambda), \lambda) > \frac{1}{2} \frac{\mu_M(\lambda)}{\epsilon \cdot \kappa(\lambda)}$ for arbitrary $\epsilon > 0$ and sufficiently large $\lambda$ whenever $p$ is a function such that $p(\lambda) > \kappa(\lambda)$ for sufficiently large $\lambda$.

Then $\mathcal{K}$ is an interactive quantum money verification protocol for the money scenario $(\mF_M, \mC_M, \mR_M)$ (in the sense defined in Definition \ref{def:interactive-money-protocol}) with completeness $1 - \mathsf{negl}(\lambda)$ and cloning error $(1+\epsilon)\kappa$.
\end{proposition}

\begin{proof}
We prove Proposition \ref{prop:pok-implies-money} by contradiction: we assume that there are provers $\hat P_A$ and $\hat P_B$ who win the soundness experiment described in Definition \ref{def:interactive-money-protocol} with probability greater than $(1+\epsilon)\kappa$ for some $\epsilon>0$, and then we show a reduction (given that $\mK$ is a secure agree-and-prove protocol for $(\mF_M, \mC_M, \mR_M)$) to an adversary who can produce a state $\sigma_{AB}$ such that
\begin{align}
\label{eq:success}
\mathrm{Pr} [
\mathcal{R}^{\mathcal{O}_{\mathcal{F}}}(1^\lambda, \mathsf{id}, \Tr_B(\sigma_{AB})) = 1
\text{ and }
\mathcal{R}^{\mathcal{O}_{\mathcal{F}}}(1^\lambda, \mathsf{id}, \Tr_A(\sigma_{AB})) = 1
] > \mu_M(\lambda)
\end{align}
for some $\reg{id} \in B_M$, where $B_M$ is the database of valid bills associated with $(\mF_M, \mC_M, \mR_M)$.

Consider an adversary $\mA$ for no-cloning which gets black-box access to $\hat P_A$, $\hat P_B$, $V_1$ and $E$ (where $E$ is the extractor guaranteed by Definition \ref{def:aap-security}) as ITMs, and which proceeds as follows. $\mA$ firstly runs the agree phase of the protocol $\mK$ against both $\hat P_A$ and $\hat P_B$ (using its black-box access to $\hat P_A$, $\hat P_B$, and $V_1$), and obtains two communication transcripts in addition to (if both agree phases succeeded) two instances $x_A$ and $x_B$. If $x_A \neq x_B$, or if either agree phase failed, $\mA$ aborts. Otherwise, let $x_A = x_B \equiv \reg{id}$.

We have assumed (toward contradiction) that the provers $\hat P_A$ and $\hat P_B$ win the soundness experiment described in Definition \ref{def:interactive-money-protocol} with probability greater than $(1+\epsilon)\kappa$. This means the probability that 1) $x_A = x_B = \reg{id}$ for some valid $\reg{id} \in B_M$ after the agree phase of $\mK$ \emph{and} 2) both $\hat P_{A,2}$ and $\hat P_{B,2}$ pass in the prove phase of $\mK$ (after the agree phase resulting in $\reg{id}$ takes place) is at least $(1+\epsilon)\kappa$. Let $p_\reg{id}$ be the probability that $\reg{id} = x_A = x_B$ is agreed on after the agree phase of $\mK$, and let $q_\reg{id}$ be the probability that both $\hat P_{A,2}$ and $\hat P_{B,2}$ pass in the prove phase of $\mK$, conditioned on $\reg{id}$ having been agreed on in the agree phase. Then our assumption can be written as $\sum_\reg{id} p_\reg{id} q_\reg{id} > (1+\epsilon)\kappa$.

We now proceed to show the following claim as a stepping stone:

\begin{claim}
\label{claim:pok-implies-money-subclaim}
If $q_\reg{id} > \kappa$, then the probability that $\mA$ is successful in creating a state $\sigma_{AB}$ satisfying equation \eqref{eq:success}, conditioned on $x_A = x_B = \reg{id}$ having been agreed on in the agree phase, is at least $\frac{\mu_M(\lambda)}{\epsilon \cdot \kappa(\lambda)}$.
\end{claim}

\begin{proof}

After it has run the agree phase of $\mK$ as described directly below equation \eqref{eq:success}, $\mA$ runs the extractor $E$ on both $\hat P_{A,2}$ and $\hat P_{B,2}$, and places the output of the first extraction (the one acting on $\hat P_{A,2}$) in register $A$, and the output of the second extraction in register $B$. The result is a joint state $\sigma_{AB}$ which may be entangled across registers $A$ and $B$. For notational convenience, we define $\sigma_A = \Tr_B(\sigma_{AB})$, and $\sigma_B = \Tr_A(\sigma_{AB})$.

\begin{enumerate}
\item Given that $q_\reg{id} > \kappa$, we can consider $\hat P_{A,2}$ to be a prover for the prove phase of $\mK$ succeeding with probability $> \kappa$ on instance $\reg{id}$ when it is provided with some auxiliary input (i.e. the state which it shares with $\hat P_B$ conditioned on $\reg{id}$ having been agreed on). Then, since $\mathcal{K}$ has extraction distance guarantee $\delta$, it follows that the state $\Tr_B(\sigma_{AB}) = \sigma_A$ is such that $\mathrm{Pr} [
\mathcal{R}^{\mathcal{O}_{\mathcal{F}}}(1^\lambda, \mathsf{id}, \sigma_A) = 1 ] > 1-\delta$.

\item Suppose (without loss of generality) that the adversary $\mA$ submits $\sigma_A$ to $\mR$ for verification first. Let the state inside registers $A$ and $B$ conditioned on $\mR$ accepting $\sigma_A$ be $\sigma_{AB,acc}'$. Since $\mR$ accepts with probability at least $1-\delta$, by the gentle measurement lemma, $\frac{1}{2} \|\sigma_{AB} - \sigma_{AB,acc}' \|_1 \leq \sqrt{\delta}$. Then, by the convexity of the trace distance, $\frac{1}{2} \|\sigma_{AB} - \sigma_{AB}' \|_1 \leq (1-\delta)\sqrt{\delta} + \delta \leq 2 \sqrt{\delta}$.

\item Recall that $\sigma_B = \Tr_A(\sigma_{AB})$ is the state that resided in register $B$ before step 2 in this list occurred (i.e. before $\mR$ was performed on register $A$). Let $\sigma_B'$ be the state in register $B$ after $\mR$ was performed on register $A$. The adversary $\mA$ now submits $\sigma_B'$ to $\mR$ for quantum money verification. By a similar analysis to that in item 1, $\mR$ would have accepted $\sigma_B$ with probability at least $1-\delta$. Furthermore, by the contractivity of the trace distance, $\frac{1}{2}\| \sigma_B - \sigma_B' \|_1 \leq \frac{1}{2} \| \sigma_{AB} - \sigma_{AB}' \|_1 \leq 2 \sqrt{\delta}$. As such, the probability that $\mR$ accepts $\sigma_B'$ is at least $1 - \delta - 2\sqrt{\delta}$.
\end{enumerate}

By a union bound, the probability that $\mR$ accepts both $\sigma_A$ and $\sigma_B'$ is at least $1 - 2\delta - 2\sqrt{\delta}$. Let $\delta_0 - \delta \equiv \Delta$. Then $1 - 2\delta - 2\sqrt{\delta} \geq 2 \Delta$, so if $\Delta > \frac{1}{2} \frac{\mu_M(\lambda)}{\epsilon \cdot \kappa(\lambda)}$, the probability that $\mA$ is successful conditioned on $q_\reg{id}$ having been agreed on is at least $\frac{\mu_M(\lambda)}{\epsilon \cdot \kappa(\lambda)}$.

\end{proof}

Now, recall our assumption that $\sum_\reg{id} p_\reg{id} q_\reg{id} > (1+\epsilon)\kappa$. Define the random variable $X \equiv 1 - q_\reg{id}$, and define $\Es{\reg{id}} q_\reg{id} \equiv \sum_\reg{id} p_\reg{id} q_\reg{id} > (1+\epsilon)\kappa$. By a Markov bound, we have
\begin{gather}
\mathrm{Pr}[X \geq 1 - \kappa] \leq \frac{1 - (1+\epsilon)\kappa}{1-\kappa} \leq 1 - \epsilon\kappa. \nonumber \\
\implies \mathrm{Pr}[q_\reg{id} > \kappa] > \epsilon \kappa. \label{eq:qid}
\end{gather}

Define the indicator variable $\mathbb{1}_{(q_\reg{id} > \kappa)} \equiv \begin{cases} 1 & q_\reg{id} > \kappa \\ 0 & \text{otherwise}\end{cases}$. Then we have
\begin{align*}
\mathrm{Pr}[\text{$\mA$ is successful}] &= \sum_\reg{id} p_\reg{id} \: \mathrm{Pr}[\text{$\mA$ is successful on $\reg{id}$}]
\\ &> \sum_\reg{id} p_\reg{id} \mathbb{1}_{(q_\reg{id} > \kappa)} \frac{\mu_M(\lambda)}{\epsilon \cdot \kappa(\lambda)}, \quad \text{using Claim \ref{claim:pok-implies-money-subclaim}}
\\ &= \frac{\mu_M(\lambda)}{\epsilon \cdot \kappa(\lambda)} \mathrm{Pr}[q_\reg{id} > \kappa]
\\ &> \mu_M(\lambda), \quad \text{using \eqref{eq:qid}.}
\end{align*}
This shows the desired contradiction with the cloning security of $(\mF_M, \mC_M, \mR_M)$. The claim follows.

\end{proof}

\paragraph{Amplification.} 
The bound on the maximum success probability of a cloning adversary which comes out of Proposition \ref{prop:pok-implies-money} is linear in the knowledge error of the agree-and-prove protocol which is being used as a money verification protocol. A typical expectation in a quantum money scenario is that any cloning adversary will only succeed with negligible probability (see Definition 11 of \cite{aaronson2012quantum}, for example), but in our analyses of quantum money verification protocols in Sections \ref{sec:wiesner} and \ref{sec:subspaces}, we only get constant (and not negligible) knowledge error. Therefore, in Proposition \ref{prop:money-amp}, we present a sequential amplification lemma which shows that a money scheme equipped with a classical agree-and-prove protocol that has constant knowledge error (and other parameters similar to those which we obtain in Sections \ref{sec:wiesner} and \ref{sec:subspaces}) can be modified into a money scheme which admits only cloning adversaries that pass with negligible probability. \\

\begin{proposition}
\label{prop:money-amp}
Let $\lambda$ be a security parameter, and let $(\mF_M, \mC_M, \mR_M)$ be a quantum money scenario (as defined in Section \ref{sec:ap-scenario-qm}). Let $\mathcal{K} = (I, P_1, P_2, V_1, V_2)$ be a protocol for a quantum money agree-and-prove scenario $(\mF_M, \mC_M, \mR_M)$. Let $\mu_M$ be the cloning parameter for the quantum money scenario $(\mF_M, \mC_M, \mR_M)$.

Define $\delta_0 \equiv \frac{2-\sqrt{3}}{2}$, as in Proposition~\ref{prop:pok-implies-money}. Let $\mK$ be a $(c=1-\mathsf{negl}(\lambda), \delta)$--secure protocol for a quantum money agree-and-prove scenario $(\mF_M, \mC_M, \mR_M)$ with extraction distance $\delta(p, \lambda) = \delta(p) = (1-p)^{\nu}$ for some constant $\nu > 0$, and with constant knowledge error $\kappa$ such that $\delta_0 - \delta(\kappa) = \delta_0 - (1-\kappa)^\nu = \text{const.} > \frac{1}{2}\frac{\mu_M(\lambda)}{\epsilon \cdot \kappa}$ for arbitrary $\epsilon > 0$ and sufficiently large $\lambda$. Then, for any polynomial $n(\lambda)$, there exists a quantum money scenario $(\mF_M', \mC_M', \mR_M')$ and a protocol $\mK'$ for this scenario such that $\mK'$ is an interactive quantum money verification protocol for the money scenario $(\mF_M, \mC_M, \mR_M)$ (in the sense defined in Definition \ref{def:interactive-money-protocol}) with completeness $1 - \mathsf{negl}(\lambda)$ and soundness $(1+\epsilon)^n\kappa^n$.
\end{proposition}

\begin{proof}
Let $B_1, \dots, B_n$ be $n$ identical copies of the database maintained by $\mF_M$. $B'$, the database for the new quantum money scenario, is defined in the following way:
\begin{gather*}
B' = \{ ( r^{(1)}, \dots, r^{(n)} ) \: : \: r^{(i)} \in B_i \}.
\end{gather*}
$\mC'(1^\lambda, x)$ outputs 1 iff $x = (\mathsf{id}^{(1)}, \dots, \mathsf{id}^{(n)})$ such that $\mathsf{id}^{(i)}$ is a valid \textsf{id} in $B_i$, and $\mR'$ is defined in the natural way (i.e. as a tensor product of $n$ copies of $\mR$ which accepts iff all the copies of $\mR$ accept). $\mK'$, meanwhile, is defined to be the protocol that consists of $n$ sequential runs of the protocol $\mK$, which we denote by $\mK^{(1)}, \dots, \mK^{(n)}$.

Consider the two entangled but noncommunicating provers $\hat P_A$ and $\hat P_B$ from the no-cloning experiment of $\mK'$. We can WLOG assume that $\hat P_A$ consists of $n$ provers $\hat P_A^{(1)}, \dots, \hat P_A^{(n)}$ such that each $\hat P_A^{(i)}$ receives a message from $\hat P_A^{(i-1)}$ in the form of a quantum auxiliary input $\sigma_i$, and such that $\hat P_A^{(i)}$ participates in the protocol $\mK^{(i)}$. We can assume the same thing of $\hat P_B$.

Consider $(\hat P_A^{(1)}, \hat P_B^{(1)})$. We can consider them a pair of entangled noncommunicating provers for the soundness experiment of $\mK$. We have chosen our parameters such that we can apply Proposition \ref{prop:pok-implies-money} to conclude that the probability that both $\hat P_A^{(1)}$ and $\hat P_B^{(1)}$ pass is less than $(1+\epsilon)\kappa$.

Now consider $(\hat P_A^{(2)}, \hat P_B^{(2)})$. Note that each of these provers is allowed to receive an arbitrary auxiliary quantum input. Let the (possibly entangled) joint state representing the input to these two provers be the input that they would have received from $(\hat P_A^{(1)}, \hat P_B^{(1)})$ \emph{conditioned on $(\hat P_A^{(1)}, \hat P_B^{(1)})$ succeeding}. This input is in general a mixed state; note that the distribution over pure states that this mixture represents can be modelled exactly by an unbounded input generation algorithm which runs the input generation algorithm that supplies inputs for $(\hat P_A^{(1)}, \hat P_B^{(1)})$, simulates the protocols $\langle \hat P_A^{(1)}, V_1^{(1)} \rangle$ and $\langle \hat P_B^{(1)}, V_2^{(1)} \rangle$ (implementing its own setup functionality oracles and simulating the verifiers), and finally measures the verifiers' output qubits, preserving the output of $(\hat P_A^{(1)}, \hat P_B^{(1)})$ if the outcome is $(1,1)$ and restarting otherwise.\footnote{There is a chance that the input generation algorithm we have described never terminates. We can address this problem by forcing it to halt after some large exponential number of steps and output fail at that stage if it has still not obtained a $(1,1)$ outcome. Then the chance that it outputs fail is exponentially small, so that the trace distance between the state the algorithm finally outputs and the conditioned state which we described is exponentially small. We will neglect this negligible term in our accounting going forward.} We then apply Proposition \ref{prop:pok-implies-money} in the presence of auxiliary input to conclude that the probability that both $\hat P_A^{(2)}$ and $\hat P_B^{(2)}$ pass, conditioned on $\hat P_A^{(1)}$ and $\hat P_B^{(1)}$ passing, is less than $(1+\epsilon)\kappa$. Then the probability that $\hat P_A^{(1)}$, $\hat P_B^{(1)}$, $\hat P_A^{(2)}$ and $\hat P_B^{(2)}$ all pass is at most $(1+\epsilon)^2\kappa^2$.

We can continue this argument inductively until the $n$th execution of $\mK$, replacing the conditioning on the success of $(\hat P_A^{(1)}, \hat P_B^{(1)})$ with conditioning on the success of all previous pairs of provers at each step. The conclusion is that the chance that $(\hat P_A^{(i)}, \hat P_B^{(i)})$ pass for all $i$ is less than $(1+\epsilon)^n\kappa^n$, which completes the proof.
\end{proof}

\section{Proofs of quantum knowledge for quantum money states}
\label{sec:money-poqk}

In this section we apply our notion of proofs of quantum knowledge to the problem of certifying quantum money. We give protocols for two constructions from the literature, Wiesner's quantum money in Section~\ref{sec:wiesner} and Aaronson and Christiano's public-key quantum money based on hidden subspaces in Section~\ref{sec:subspaces}. 

\subsection{PoQK for Wiesner money states}
\label{sec:wiesner}

Our first concrete example of an Agree-and-Prove scheme for a quantum property is a verification protocol for Wiesner's quantum money states. As we recalled in the introduction, quantum money states in Wiesner's scheme are $\lambda$-qubit states, with $\lambda \in \mathbb{N}$ a security parameter, such that each qubit is chosen from the set $\{\ket{0}, \ket{1}, \ket{+}, \ket{-}\}$. Any such state can be described  by $2\lambda$ classical bits; a typical classical description is the pair of strings $\$ = (v, \theta)\in\{0,1\}^{2\lambda}$, where the associated money state is
\begin{equation}\label{eq:wiesner-state}
\ket{\$}_{v, \theta} \,=\, \Big(\prod_i H_i^{\theta_i}\Big) \ket{v}\;,
\end{equation}
where $\ket{v} = \otimes_i \ket{v_i}$ and $H_i$ denotes a Hadamard on the $i$-th qubit. 
 In the notation of Definition~\ref{def:qmoney}, valid bills in this scheme are quadruples $(\mathsf{id, public, secret}, \ket{\$}_{\mathsf{id}})$ such that $\mathsf{id}$ is an arbitrary string, $\mathsf{public}$ is empty, $\mathsf{secret} = (v,\theta)\in\{0,1\}^\lambda\times \{0,1\}^\lambda$ and $\ket{\$}_{\mathsf{id}}=\ket{\$}_{v, \theta}$. The verification procedure $\texttt{Ver}(x, \textsf{public}, \textsf{secret}, \rho_W)$ parses $\textsf{secret}=(v,\theta)$ and measures each qubit of $\rho_W$ in the basis indicated by $\theta$. It accepts if and only if the outcomes obtained match $v$. This scheme clearly has completeness parameter $1$, and it was shown in~\cite{molina2012optimal} that its cloning parameter is $\mu_W(\lambda)=(3/4)^\lambda$. 

\begin{scenario}
\label{scenario:wiesner}
We instantiate the generic AaP scenario for quantum money described in Section \ref{sec:ap-scenario-qm} as follows:
\begin{itemize}
\item Setup functionality $\mathcal{F}_W(1^\lambda)$:
\begin{itemize}
\item $\mathtt{init}_W$ initializes a random oracle $H$ taking strings of length $2\lambda$ to strings of length $2\lambda$.\footnote{Formally the oracle is implemented in the standard way, recalled in Section~\ref{sec:oracles}.} 
In addition, it initializes an empty database $B_W$ that is destined to contain a record of all quantum money bills in circulation. 
\item  $\mathcal{O}_{\mF_W}(I, \texttt{getId})$: generates $\reg{id} \in \{0,1\}^{2\lambda}$ uniformly at random. Sets $(v,\theta) = H(\reg{id})$, $\textsf{secret}=(v,\theta)$ and $\ket{\$} = \ket{\$}_{v,\theta}$. If $\mathsf{id}$ already appears in $B_W$, then returns $\perp$. 
 Otherwise, add $(\mathsf{id,public,secret},\ket{\$}_{v,\theta})$ to $B_W$. Return $\mathsf{id}$.
\item $\mathcal{O}_{\mF_W}(\cdot, \texttt{public}, \mathsf{id})$,  $\mathcal{O}_{\mF_W}(I, \texttt{getMoney}, \mathsf{id})$, and $\mathcal{O}_{\mF_W}(V, \texttt{secret}, \mathsf{id})$: as described in Section~\ref{sec:ap-scenario-qm}. 
\end{itemize}
\item Agreement relation $\mathcal{C}^{\mO_{\mF_W}}_W(1^\lambda, \mathsf{id})$: The agreement relation is the same as it is in Section \ref{sec:ap-scenario-qm}.
\item Proof relation $\mathcal{R}_W^{\mathcal{O}_{\mathcal{F}_W}}(1^\lambda, x, \rho_W)$: The proof relation firstly queries $\mathcal{O}_{\mathcal{F}_W}(V, \texttt{secret}, x)$ in order to get a tuple $(v, \theta)$. (If $\mathcal{O}_{\mathcal{F}}(V, \texttt{secret}, x)$ returns $\perp$, then $\mathcal{R}$ rejects.) Then it implements the Wiesner money verification procedure: it applies $\prod_i H_i^{\theta_i}$ to its quantum input $\rho_W$, measures all qubits in the computational basis, and accepts if and only if the outcome is $v$.
\end{itemize}
\end{scenario}

\begin{protocol}
\label{protocol:wiesner} We define our proof of knowledge protocol $\mK_W = (\mI_W, P_1, P_2, V_1, V_2)$ for the scenario $(\mF_W, \mC_W, \mR_W)$. An honest input generation algorithm $I \in \mI_W$ calls $\mO_{\mF_W}(I, \texttt{getId})$ repeatedly until it obtains a string $\reg{id} \in \{0,1\}^{2\lambda}$ such that $\reg{id} \neq \perp$. It then queries $\mathcal{O}_{\mathcal{F}}(I, \texttt{getMoney}, \mathsf{id})$, obtains a quantum state $\rho_W$, and gives $(\reg{id}, \rho_W)$ to the prover (it gives nothing to the verifier). In the agreement phase, the prover $P_1$ parses the auxiliary input $\rho_{\reg{AUX}_P}$ which it gets from $I$ as a classical string $\mathsf{id} \in \{0,1\}^{2\lambda}$ in addition to a quantum state $\rho_W$. (If this fails, the prover halts.) Then the prover sends $\mathsf{id}$ to $V_1$ and outputs the statement $x_P=\mathsf{id}$ and the quantum state $\rho_{st_P} = \rho_W$. $V_1$, upon receiving $\mathsf{id}$ from $P_1$, queries $\mO_{\mF_W}(V,\texttt{public},\mathsf{id})$. If this returns $\perp$ the verifier aborts. Otherwise, $V_1$ outputs $x_V = \mathsf{id}$ and $st_V = \bot$.

This completes the description of the (honest) first phase prover and verifier. We now describe the interaction between the (honest) second phase prover and verifier, $P_2$ and $V_2$:
\begin{enumerate}
\item $V_2$ queries $\mO_{\mF_W}(V,\texttt{secret},\mathsf{id})$. If it obtains $\perp$, $V_2$ aborts. Otherwise, let $\$=(v,\theta)$ be the classical description obtained.
\item $V_2$ sends a uniformly random $c\in\{0,1\}^\lambda$ to the prover.
\item For each $i\in\{1,\ldots,n\}$ if the $i$th bit of $c$ is 0, $P_2$ measures the $i$th qubit of $\rho_{st_P}$ in the standard basis; and if it is 1, it measures the $i$th qubit in the Hadamard basis. Let $\beta\in\{0,1\}^\lambda$ denote the outcomes obtained. $P_2$ sends $\beta$ to $V_2$.
\item Let $s = c \cdot \theta + \bar{c}\cdot\bar{\theta}$, where $\cdot$ denotes componentwise multiplication. In other words, $s_i=1$ if and only if $c_i=\theta_i$. $V_2$ checks that, whenever $s_i = 1$, it holds that $v_i=\beta_i$. If not, then it returns $0$. 
\end{enumerate}
\end{protocol}

\begin{lemma}\label{lem:wiesner}
There is a constant $\kappa <1 $ such that Protocol $\mK_W$ (Protocol~\ref{protocol:wiesner}) is a secure agree-and-prove protocol for the Wiesner money AaP scenario $(\mF_W, \mC_W, \mR_W)$ (Scenario~\ref{scenario:wiesner}) with completeness $1$, knowledge error $\kappa$, and extraction distance $\delta=O(\mu^{1/4})$, where $\mu=1-p$ and $p$ is the prover's success probability.
\end{lemma}

\begin{proof}
The completeness condition is clear, and follows from completeness of the Wiesner quantum money verification scheme. For soundness, we need to define the extractor $E$. Towards this we first introduce notation to model the actions of the prover as an ITM, and derive constraints on the resulting ITM that follow from the assumption that the prover succeeds in the protocol with some probability $p$. (See Section~\ref{sec:black-box} for terminology relating to ITMs.) 

\paragraph{Step 1: Modeling the prover.}
We start by introducing notation to model an arbitrary prover $\hat{P} = (\hat{P}_1,\hat{P}_2)$ in the protocol. By definition of the soundness experiment, to define the extractor we may assume without loss of generality that the interaction of $\hat{P}_1$ and $V_1$ results in a statement $x_P=x_V=\mathsf{id}$ such that $\mathsf{id}$ appears in the database $B_W$, and such that moreover the execution of $\hat{P}_2$ and $V_2$ on $\reg{id}$, when provided auxiliary inputs $\rho_{st_P}$ and $st_V$ respectively, results in acceptance with probability $p=1-\mu\geq \kappa$, for some $\kappa$ to be determined at the end of the proof. Note that for an honest verifier $st_V = \perp$ always. 

In the second phase of the protocol $\hat{P}_2$ expects a query $c\in\{0,1\}^\lambda$ and is required to return an outcome $\beta\in\{0,1\}^\lambda$. As an ITM the action of $\hat{P}_2$ is realized by a unitary $U$ acting on $\mH_\reg{M} \otimes \mH_\reg{P}$, where $\mH_\reg{M} = \mH_{\reg{M}_1} \otimes \mH_{\reg{M}_2}$ with $\mH_{\reg{M}_1}=\mH_{\reg{M}_2}=(\C^2)^{\otimes 2\lambda}$ a message register (the first part of which is associated with the query $c$ and the second part the prover's response $\beta$) and $\mH_\reg{P}$ the prover's secret space (which contains $\rho_{st_P}$ as well as some ancilla qubits if needed). 
Since $c$ is always classical we can assume that $U = \sum_c \proj{c} \otimes U_c$ where for each $c$, $U_c$ is a unitary that acts on $\mH_{\reg{M}_2}\otimes \mH_\reg{P}$. After the prover applies $U$ (in the real protocol), it performs a measurement in the computational basis of the $\lambda$ qubits in $\mH_{\reg{M}_2}$ to get the outcome $\beta$ which constitutes its response.
For notational convenience, for $c\in\{0,1\}^\lambda$ we define a unitary $V_c$ as  $V_c = \big( \prod_i H^{c_i}_i \big) U_c$, where the Hadamard $H_i$ acts on the $i$-th qubit in register $\reg{M}_2$. Intuitively, $V_c$ can be thought of as the unitary the prover would apply to $\mH_{\reg{M}_2} \otimes \mH_{\reg{P}}$ if it were to measure $\mH_{\reg{M}_2}$ in the bases determined by $c$ (instead of in the standard basis) in order to get the outcome $\beta$.

We write $\sigma_X$ and $\sigma_Z$ for the Pauli matrices on a single qubit. We use the shorthand $\sigma_Z(a)$, for some string $a \in \{0,1\}^\lambda$, to denote the $\lambda$-qubit observable that is a tensor product of single qubit observables, and which is $\sigma_Z$ on those qubits $i$ such that $a_i = 1$, and $I$ otherwise. (In particular, for the case $\lambda=1$ we have that $\sigma_Z(0)=I$ and $\sigma_Z(1)=\sigma_Z$.) Let $\sigma_X(a)$ be defined similarly. 
For any $c\in \{0,1\}^\lambda$, define the following observables on $\mH_{\reg{M}_2} \otimes \mH_{\reg{P}}$:
\begin{equation}\label{eq:def-prover-obs}
 Z^B(c) \,=\, V_{{c}}^\dagger (\sigma_Z(\bar{c})\otimes \Id) V_{{c}}\qquad\text{and}\qquad X^B(c) \,=\, V_{c}^\dagger (\sigma_X(c)\otimes \Id) V_{c}\;,
\end{equation}
where $\bar{c}$ denotes the bitwise complement of $c$, and where the identities act on $\mH_\reg{P}$.  Intuitively, for a given challenge string $c$, $Z^B(c)$ denotes the observable which measures the parity of all the bits $\beta_i$ of the prover's response $\beta$ such that $c_i = 0$, while $X^B(c)$ denotes the observable which measures the parity of all the bits $\beta_i$ of $\beta$ such that $c_i = 1$. Note that the observables $Z^B(c)$ and $X^B(c)$ can be implemented with black-box access to an ITM implementing the prover, by initializing $\mH_\reg{M}$ to $\ket{c}\otimes \ket{0^\lambda}$, executing the prover's unitary $U$ (which acts on $\mH_{\reg{M}} \otimes \mH_{\reg{P}}$) followed by the controlled-Hadamard $\prod_i H_i^{c_i}$ on $\mH_{\reg{M}_2}$, measuring $\sigma_Z(\bar c)$ or $\sigma_X(c)$ on $\mH_{\reg{M}_2}$, and executing $ (\prod_i H_i^{c_i} )U^\dagger$. 

\paragraph{Step 2: The prover's observables.} We now state some properties of the observables $Z^B(c)$ and $X^B(c)$ defined in~\eqref{eq:def-prover-obs} that follow from the assumption that the prover succeeds with probability $p=1-\mu$, for some $\mu\geq 0$, in the second phase of the protocol. Towards this we introduce a ``purified'' description of the operation of the oracle $\mO_{\mF_W}$ which we denote $\mO_{\mF'_W}$ and operates as follows:
\begin{itemize}
\item $\mathcal{O}_{\mF'_W}(I, \texttt{getId})$: generate a uniformly random $\mathsf{id}\in \{0,1\}^{2\lambda}$. Create $2\lambda$ EPR pairs $\ket{\phi}_{\mathsf{id}} = \ket{\mathrm{EPR}}_{\reg{A}_{\mathsf{id}}\reg{B}_{\mathsf{id}}}^{\otimes 2\lambda}$ and store $(\mathsf{id},\texttt{unmeasured},\ket{\phi}_{\mathsf{id}})$ in the (quantum) database $B_{W}'$. Return $\mathsf{id}$.
\item $\mathcal{O}_{\mF'_W}(\cdot, \texttt{public}, \mathsf{id})$: Not needed since there is no $\mathsf{public}$ record in this scheme. 
\item $\mathcal{O}_{\mF'_W}(I, \texttt{getMoney}, \mathsf{id})$: If no record in $B_{W}'$ with identifier \textsf{id} exists, returns $\perp$. Otherwise, return the second half of each EPR pair in $\ket{\phi}_{\mathsf{id}}$, i.e. the state in register $\reg{B}_{\mathsf{id}}$, the first time it is called. If called again with the same \textsf{id} argument, return $\perp$. 
\item $\mathcal{O}_{\mF'_W}(V, \texttt{secret}, \mathsf{id})$:  If no record in $B_{W}'$ with identifier \textsf{id} exists, returns $\perp$. If there is a record of the form  $(\mathsf{id},\texttt{measured},x)$ in $B'_W$ then return $x$. Otherwise, the record is of the form $(\mathsf{id},\texttt{unmeasured},\sigma_{\reg{A}_\id\reg{B}_\id})$. In this case, 
select a uniformly random $\theta\in\{0,1\}^\lambda$ and measure register $\reg{A}_{\mathsf{id}}$ of $\sigma$ in the basis indicated by $\theta$ to obtain a string of outcomes $v$. Return $(v,\theta)$. Update the database entry to $(\mathsf{id},\texttt{measured},(v,\theta))$.
\end{itemize}
We observe that the operation of $\mathcal{O}_{\mF'_W}$ is indistinguishable from $\mO_{\mF_W}$ from the point of view of the parties $I$ or $P$. Indeed the only formal differences between the two lie in the implementation of the following operations:
\begin{itemize}
\item $\mO_{\mF_W}$: when any identifier $\reg{id}$ is added to $B_W$, a state $\ket{\$}_{\mathsf{id}}$ such that $\$ = (v,\theta)$ has been generated uniformly at random is also added to $B_W$ and associated with the identifier $\reg{id}$. This is the state that  $\mathcal{O}_{\mF_W}(I, \texttt{getMoney}, \mathsf{id})$ returns to $I$, and this operation can take place only once. 
\item $\mathcal{O}_{\mF'_W}$: when any identifier $\reg{id}$ is added to $B'_W$, EPR pairs $\ket{\phi}_{\mathsf{id}} = \ket{\mathrm{EPR}}_{\reg{A}_{\mathsf{id}}\reg{B}_{\mathsf{id}}}^{\otimes 2\lambda}$ are also added to $B'_W$ and associated with the identifier $\reg{id}$. Half of these EPR pairs are returned to $I$ when 
 $\mathcal{O}_{\mF'_W}(I, \texttt{getMoney}, \mathsf{id})$ is called for the first time. We argue in the paragraph below that these half-EPR pairs are indistinguishable from a uniformly random money state $\ket{\$}_{v,\theta}$ unless one has access to $(v,\theta)$. (In the scenario $(\mF'_W, \mC'_W, \mR'_W)$, $(v,\theta)$ is decided as the classical outcome of the measurement which
$\mF'_W$ performs on register $\reg{A}_{\mathsf{id}}$ of the joint state $\ket{\mathrm{EPR}}_{\reg{A}_{\mathsf{id}}\reg{B}_{\mathsf{id}}}^{\otimes 2\lambda}$ after the verifier calls $\mathcal{O}_{\mF'_W}(V, \texttt{secret}, \mathsf{id})$. The result of this measurement is that both registers $\reg{A}_{\mathsf{id}}$ and $\reg{B}_{\mathsf{id}}$ collapse to $\ket{\$}_{v, \theta}$.)
 
 One can see that half EPR pairs are indistinguishable to $I$ or $P$ from $\proj{\$}_{v, \theta}$ for uniformly random $(v, \theta)$ by the following argument. Local operations on different halves of EPR pairs commute. If the measurement which $\mF'_W$ performs on register $\reg{A}_{\mathsf{id}}$ (the one which collapses registers $\reg{A}_{\mathsf{id}}$ and $\reg{B}_{\mathsf{id}}$ to two copies of $\ket{\$}_{v, \theta}$ for uniformly random $(v, \theta)$) occurs \emph{before} $I$ calls \texttt{getMoney}, then it is easy to see that $\mO_{\mF_W}$ and $\mO_{\mF'_W}$ are identical from the perspective of $I$ or $P$. By the commutativity of local operations on half EPR pairs, there is no measurement $I$ or $P$ could make which would have a different output distribution in the scenario $(\mF'_W, \mC'_W, \mR'_W)$ as compared with its output distribution in the scenario $(\mF_W, \mC_W, \mR_W)$, even if $\mF'_W$ performs its measurement on $\reg{A_{id}}$ after $I$ calls \texttt{getMoney}---as long, of course, as $\mF'_W$ does not disclose $(v, \theta)$ to $I$ or $P$, so that the assumption of local operations is valid.
\end{itemize}

At the end of the first phase of the protocol, let $\mathsf{id}=x_P=x_V$ be the identifier upon which $V_1$ and $\hat{P}_1$ agreed (if there is no agreement then we do not need to specify the extractor). For $a,b\in\{0,1\}^\lambda$ define $Z^A(b) = \sigma_Z(b)$ and $X^A(a) = \sigma_X(a)$, where the $\sigma_Z$ and $\sigma_X$ operators act on the register  $\reg{A}_{\mathsf{id}}$ associated with the entry $\mathsf{id}$ in the database $B_W'$ maintained by $\mF'_W$ in the ``purified'' version of the scenario described above. Note that there must exist such an entry, as otherwise by definition the agreement phase does not succeed. Let $\ket{\psi} \in \mH_{\reg{A}_{\mathsf{id}}}\otimes \mH_{\reg{M}_2} \otimes \mH_\reg{P}$ be the joint state of the registers $\reg{A}_{\mathsf{id}}$ (for the fixed $\mathsf{id}$ above) and the prover at the start of the second phase of the protocol. Without loss of generality we assume that this state is pure (if not, we use $\ket{\psi}$ to denote a purification of it such that $\hat{P}_2$ is in addition given the purifying system). We also consider that the part of $\ket{\psi}$ that lies on $\reg{M}_2$ is in state $\ket{0}$, although this is not important for the proof. 

\paragraph{Step 3: Correlations.}
 Let $c$ denote the query sent by $V_2$ to $\hat{P}_2$ and $\beta$ the answer given by $\hat{P}_2$. Let $(v,\theta)$ denote the outcome obtained by $V_1$ when it queries $\mO_{\mF'_W}(V, \texttt{secret}, \mathsf{id})$ at the end of the first phase of Protocol~\ref{protocol:wiesner}. According to the purified description given in the previous step, $\theta$ is chosen uniformly at random and $v$ is obtained as the outcome of a measurement of the qubits of $\ket{\psi}$ in register $\reg{A}_{\mathsf{id}}$ in the bases indicated by $\theta$. The check performed at the end of the protocol is that $\beta_i = v_i$ for all $i$ such that $\theta_i = c_i$. By assumption, this check is passed with probability $1-\mu$, on average over the choice of $c$ and $\theta$, by the prover. 
It follows that with probability at least $1-\mu$, 
\[ \oplus_{i: \, c_i = \theta_i = 1}  \,\beta_i \,=\, \oplus_{i: \, c_i = \theta_i = 1} \, v_i \qquad\text{and}\qquad \oplus_{i: \, c_i = \theta_i = 0}  \,\beta_i \,=\, \oplus_{i: \, c_i = \theta_i = 0} \, v_i\;.\]
Using the notation for the prover's observables introduced at the previous step this condition is re-written as the condition that 
\begin{equation}
\Es{c,\theta} \:  \bra{\psi} Z^A(c\cdot \theta) \otimes Z^B(c\cdot \theta) \ket{\psi} \,\geq\, 
 (1)(1-\mu) + (-1)\mu \,=\, 1 - 2\mu\;,\label{eq:n0-1}
\end{equation}
where the expectation is taken under the uniform distribution over $c,\theta\in\{0,1\}^\lambda$. Similarly, 
\begin{equation}
\Es{c, \theta} \: \bra{\psi} X^A(\bar{c}\cdot \bar{\theta}) \otimes X^B(\bar{c}\cdot \bar{\theta}) \ket{\psi} \,\geq\, 1 - 2\mu\;.\label{eq:n0-2}
\end{equation}

\paragraph{Step 4: Rigidity.}
In preparation for the definition of the extractor $E$ we invoke the following lemma, whose proof is given at the end of the section. The lemma is an adaptation of results that appeared in~\cite{nv16}; informally, it states that observables $X^B(a)$ and $Z^B(b)$ satisfying~\eqref{eq:n0-1} and~\eqref{eq:n0-2} are necessarily isometric to Pauli observables. 

\begin{lemma}\label{lem:gh}
Let $\ket{\psi}_{AB} \in (\C^2)^{\otimes n}_\reg{A} \otimes \mH_\reg{B}$, where $\mH_\reg{B}$ is arbitrary. 
Suppose that for every $a\in\{0,1\}^n$ there exist observables $X^B(a)$ and $Z^B(a)$ on $\mH_{B}$ such that
\begin{equation}\label{eq:n2-0}
\forall W \in \{X, Z\}\;,\qquad \Es{a}  \big\|\big(\sigma_W^A(a) -W^B(a)\big)\ket{\psi} \big\|^2 \,\leq\,\eps\;,
\end{equation}
for some $0\leq \epsilon\leq 1$ and where the expectation is under the uniform distribution over $a\in \{0,1\}^n$. Then there exists an isometry 
\begin{equation*}
\Phi_B \: : \: \mH_{B} \rightarrow  ((\mathbb{C}^2)^{\otimes n})_{B'} \otimes \mH_{\hat{B}}
\end{equation*}
 and a state $\ket{aux}_{\hat{B}}$ on $\mH_{\hat{B}}$ such that
\begin{equation}
\Tr\Big( \big(\bra{\mathrm{EPR}}_{AB'}^{\otimes
      n}\otimes \bra{aux}_{\hat{B}}\big)\,\big( \Id_\reg{A}\otimes \Phi_B(\proj{\psi}_{AB})\big)\big( \ket{\mathrm{EPR}}_{AB'}^{\otimes n} \otimes \ket{aux}_{\hat{B}}\big) \Big) = 1 - O\big(\eps^{1/2}\big).\label{eq:n2-0b}
\end{equation}
Moreover, the isometry $\Phi_B$ can be implemented as an $O(n)$-size quantum circuit acting on $\mH_B$ as well as some ancilla qubits, and that uses controlled gates for $X^B(a)$ and $Z^b(b)$ (controlled on an ancilla register of dimension $2^n$ that contains $a$ or $b$) as black boxes. Precisely, $\Phi^B$ is defined by 
\begin{equation}\label{eq:def-iso}
\Phi^B \ket{\varphi}_\reg{B} \,=\, \Big(\frac{1}{2^n} \sum_{a,b} X^B(a)Z^B(b) \otimes \sigma_X(a)\sigma_Z(b) \otimes \Id\Big)\ket{\varphi}_\reg{B} \ket{\mathrm{EPR}}^{\otimes n}\;,
\end{equation}
and letting $\mH_{\hat{B}} = \mH_{\reg{B}} \otimes (C^2)^{\otimes n}$, with register $\reg{B}'$ associated with the last $n$ copies of $\C^2$. 

Finally, the lemma also holds, with the same conclusion (but weaker implied constants), if~\eqref{eq:n2-0} holds when the expectation is 
\begin{itemize}
\item[(i)] uniform over those strings $a$ such that $|a|_H=\frac{n}{2}$, where $|\cdot|_H$ denotes the Hamming weight, or
\item[(ii)] taken under the biased distribution where each coordinate of $a$ is sampled independently such that $a_i=1$ with probability $\frac{1}{4}$ and $a_i=0$ with probability $\frac{3}{4}$. 
\end{itemize}
\end{lemma}

We apply Lemma~\ref{lem:gh} with the following identifications. We let $n$ in the lemma be $\lambda$ here. The observables $W^B(a)$ are as defined in~\eqref{eq:def-prover-obs}, with the space $\mH_\reg{B}$ in Lemma~\ref{lem:gh} identified with $\mH_{\reg{M}_2}\otimes \mH_\reg{P}$ here. The space $\mH_\reg{A}$ in Lemma~\ref{lem:gh} is identified with the space $\mH_{\reg{A}_{\mathsf{id}}}$ here. Setting $\epsilon= 4\mu$, it follows from~\eqref{eq:n0-1} and~\eqref{eq:n0-2} that equation~\eqref{eq:n2-0} of Lemma \ref{lem:gh}, in its version considered in the ``Finally,'' part (ii) of the lemma,  is satisfied. 
Eq~\eqref{eq:n2-0b} implies (by the contractivity of the trace distance, and the equality $\| \rho - \sigma \|_1 = \sqrt{1 - \Tr(\rho \sigma})$ that holds when $\rho, \sigma$ are pure) that on average over $(v,\theta)$ chosen uniformly at random (equivalently, $\theta$ chosen uniformly at random and $v$ the outcome of a measurement of the $\lambda$ qubits in $\reg{A}_{\mathsf{id}}$ in the basis indicated by $\theta$, as would be done by $\mathcal{O}_{\mF'_W}(V, \texttt{secret}, \mathsf{id})$---this is equivalent because we know that the reduced density of $\ket{\psi}$ on $\reg{A}_{\mathsf{id}}$ is totally mixed),
\begin{equation}\label{eq:wiesner-close-2}
\big\| \Phi_B \Tr_A \big( (\proj{\$}_{v, \theta} \otimes  \Id_{\reg{M}_2\reg{P}}) \cdot \proj{\psi} \big)- \proj{\$}_{v, \theta}\otimes  \proj{aux}_{\hat{B}} \big\|_1 = O(\eps^{1/4})\;.
\end{equation}

\paragraph{Step 5: Definition of $E$.}
We are ready to define the extractor. Recall that $E$ gets black-box access to the ITM $\hat{P}_2$, which operates on the same initial state as $\hat{P}_2$. ($E$ also has the same access to $\mO_\mF$ as $\hat{P}_2$ does, but we will not need this here.) Consider $E$ that uses its access to $\hat{P}_2$ in order to implement the isometry $\Phi_B$. As stated in Lemma~\ref{lem:gh} this is possible using black-box access to $\hat{P}_2$. Then, $E$ traces out the register $\hat{B}$ and returns the state contained in $\reg{B}'$. Using~\eqref{eq:wiesner-close-2} it follows that the resulting state is $O(\mu^{1/4})$ close to $\ket{\$}_{v, \theta}$, on average over $v$ and $\theta$. In particular this state must be accepted by the verification procedure $\mR_W$ with probability $1-O(\mu^{1/4})$, where the probability is taken over the specification of the oracle $H$.
Choosing $\kappa$ to be any constant $\kappa < 1 $ such that this quantity is positive for all $p=1-\mu \geq \kappa$, this completes the proof.
\end{proof}

We conclude the section by giving the proof of Lemma~\ref{lem:gh}. 
The proof uses the following general lemma, that is based on~\cite{gowers2017inverse}. See e.g.~\cite[Lemma 4.7]{coladangelo2017robust} for a proof. In the lemma, $\Unitary(\mH)$ denotes the group of unitaries acting on Hilbert space $\mH$. 

\begin{theorem}\label{thm:gh}
Let $G$ be a finite group. Let $f:G\mapsto \Unitary(\mH_{B'})$ and $\ket{\psi}_{AB'}\in \mH_A \otimes \mH_{B'}$ be such that
\begin{equation}\label{eq:main-ass}
\Es{x,y\in G} \bra{\psi} \Id_A \otimes f(x)f(yx)^\dagger f(y) \ket{\psi} \geq 1-\delta\;,
\end{equation}
for some $\delta>0$. Then there exists an isometry $V:\mH_{B'} \to \mH_{{B}''}$ and a representation $g:G \mapsto \Unitary(\mH_{{B}''})$  such that
\begin{equation}\label{eq:main-ccl}
\Es{x\in G} \big\|\Id_A \otimes \big(f(x) - V^\dagger g(x) V\big) \ket{\psi} \big\|^2 \,\leq 2\delta\;.
\end{equation}
\end{theorem}

In this paper we are particularly concerned with the Pauli group $G=\mathcal{P}_n$, which can be defined as the $2\cdot 4^n$-element matrix group generated by the $n$-qubit Pauli matrices $\sigma_X$ and $\sigma_Z$, i.e.\ 
\begin{equation}\label{eq:def-pn}
\mathcal{P}_n = \big\{ \pm \sigma_X(u)\sigma_Z(v)\big|\, u,v\in\{0,1\}^n\big\}\;.
\end{equation}
We represent each element of $\mathcal{P}^n$ by a triple $(\eps,u,v)\in \{\pm 1\} \times \{0,1\}^n \times \{0,1\}^n$ in the natural way. In particular, the identity element is $e=(1,0^n,0^n)$. For $x=(\eps,u,v)\in \mathcal{P}_n$, let $f(x)=\epsilon{X}^B(u){Z}^B(v)$. 
In this case it can be verified (see  e.g.\ the proof given in~\cite{coladangelo2017robust}) that  the isometry $V$ promised in Theorem~\ref{thm:gh} takes the form~\eqref{eq:def-iso}. 
In particular, an efficient circuit for $V$ is easily constructed from an efficient circuit for applying $f$ controlled on $x$. 

\begin{proof}[Proof of Lemma~\ref{lem:gh}]
We first show how either cases (i) or (ii) of the ``finally'' part of the lemma can be reduced to  the initial statement of the lemma. 

Consider first case (i), where~\eqref{eq:n2-0} is only promised to hold on average over strings of Hamming weight $\frac{n}{2}$. 
For a string $u\in\{0,1\}^n$ define two random strings $(a,b)$ of Hamming weight $n/2$ as follows. First assume that the Hamming weight of $u$ is even. Let $S\subseteq\{1,\ldots,n\}$ denote the location of the '$1$' entries in $u$. Let $S_1$ be a uniformly random subset of $S$ of size $|S|/2$, and $S_2$ a uniformly random subset of $\overline{S}$ of size $|\overline{S}|/2$. Set all entries of $a$ in $S_1 \cup S_2$ to '$1$', and all other entries to '$0$'. Set $b=u+ a$. Then it is clear that $a$ and $b$ both have Hamming weight $\frac{n}{2}$. Furthermore, if $u$ is uniformly distributed then $a$ and $b$ are each uniformly distributed over all strings of Hamming weight $\frac{n}{2}$. If the Hamming weight of $u$ is odd, a similar construction, flipping an additional coin to decide if $|S_1|=(|S|\pm 1)/2$, applies. 

A similar construction can be done in case (ii), where~\eqref{eq:n2-0} is promised to hold under a biased distribution. For a string $u\in\{0,1\}^n$ define two biased strings $a$ and $b$ as follows. For every $i$ such that $u_i=0$ set $a_i=b_i=0$. For every $i$ such that $u_i=1$ select $c\in\{0,1\}$ uniformly at random and set $a_i= c$ and $b_i=(1-c)$. Then $u=a+b$ and if $u$ is uniformly distributed the marginal distribution of $a$ and $b$ is the biased distribution. 

For $u\in \{0,1\}^n$ define $\hat{X}^B(u)=X^B(a)X^B(b)$, where $a$ and $b$ are generated according to the process described above (for either case (i) or (ii)), independently for each $u$. Similarly, for $v\in \{0,1\}^n$ define $\hat{Z}^B(v)=Z^B(a)X^B(b)$, with $a,b$ generated independently for each $v$. Then $\{\hat{X}^B(u)\}$ and $\{\hat{Z}^B(v)\}$ are unitaries on $\mH_{B'}$. 

For $x=(\eps,u,v)\in \mathcal{P}_n$, let $f(x)=\epsilon\hat{X}^B(u)\hat{Z}^B(v)$. 

\begin{proposition}\label{prop:n1-1}
The following holds:
\begin{equation}\label{eq:n1-1e}
 \Es{x,y,z\in \mathcal{P}^n }  \bra{\psi} \Id_A \otimes f(x)f(yx)^\dagger f(y) \ket{\psi} \geq 1-O(\eps)\;,
\end{equation}
where the expectation is over three uniformly random group elements. 
\end{proposition}

\begin{proof}
First we observe that the claim holds, with no error, in case $f$ is replaced by $f^A(x)=\epsilon\sigma_X^A(u)\sigma_{Z}^A(v)$, for $x=(\eps,u,v)$. Next we note that for all $W\in \{X,Z\}$ it holds that 
\begin{equation}\label{eq:n1-1a}
 \Es{u\in\{0,1\}^n} \big\| \big( \sigma_W^A(u) - \hat{W}^B(u) \big)\ket{\psi} \big\|^2 = O(\eps)\;.
\end{equation}
Indeed, this follows by two applications of~\eqref{eq:n2-0}, the triangle inequality, and the fact that the marginal distributions of $a$ and $b$ for $u$ chosen uniformly at random are both uniform. Applying~\eqref{eq:n1-1a} twice, once for $W=Z$ and once for $W=X$, we get 
\begin{equation}\label{eq:n1-1b}
 \Es{\epsilon\in \{\pm 1\}, u,v\in\{0,1\}^n} \big\| \big( f^A(\eps,u,v) - f(\eps,u,v) \big)\ket{\psi} \big\|^2 = O(\eps)\;.
\end{equation}
Given that $f^A$ satisfies~\eqref{eq:n1-1e}, this concludes the proof. 
\end{proof}

Proposition~\ref{prop:n1-1} shows that the function $f$ satisfies the assumption of Theorem~\ref{thm:gh}, for some $\delta = O(\eps)$. Let $g$ be the representation promised by the theorem. Let $g = g_+ \oplus g_-$ where $g_+(-e)=\Id$ and $g_-(-e)=-\Id$ (recall that $e$ denotes the neutral element of $\mathcal{P}_n$). The only irreducible representation of $\mathcal{P}_n$ that does not send $-e$ to the identity is the Pauli matrix representation $\tau_\mathcal{P}$ that we used to define the group in~\eqref{eq:def-pn}. Thus $g_- = \tau_{\mathcal{P}} \otimes \Id_d$, for some integer $d$.  Let $d'$ be the dimension of $g_+$ and $\Pi_+$ the projection on its range. 
Using that by definition $f(-x)=-f(x)$ for all $x\in G$, we get from~\eqref{eq:main-ccl} using the triangle inequality that 
\begin{align}\label{eq:n2-1}
\Es{x\in \mathcal{P}_n} \big\| \Id_A \otimes V^\dagger 2\Pi_+ V \ket{\psi}\big\|^2
&= \Es{x\in \mathcal{P}_n} \big\| \Id_A \otimes V^\dagger \big( g(x)+g(-x) \big) V \ket{\psi} \big\|^2\notag\\
&= O(\eps)\;.
\end{align}
In particular, it then follows from~\eqref{eq:main-ccl} that
\begin{equation}\label{eq:n2-1f}
\Es{x\in \mathcal{P}_n} \big\| \Id_A \otimes \big( f(x) - V^\dagger\big((\tau_\mathcal{P}(x) \otimes \Id_d) \oplus 0_{d'} \big) V \ket{\psi} \big\|^2
= O(\eps)\;.
\end{equation}
Let $\ket{\psi'} = (\Id \otimes V)\ket{\psi}$. 
Using the assumption~\eqref{eq:n2-0} twice and~\eqref{eq:n2-1f} we get that for $W\in \{X,Z\}$,
\begin{equation}\label{eq:n2-3}
\Es{u\in\{0,1\}^n} \bra{\psi'} \sigma_W^A(u) \otimes \big( (\sigma_W^B(u) \otimes \Id_d ) \oplus 0_{d'} \big)\ket{\psi'}\,=\, 1-O\big(\sqrt{\eps}\big)\;.
\end{equation}
It is easy to verify by direct calculation that 
\[ \Es{u,v} \sigma_X^A(u)\sigma_Z^A(v)\otimes \sigma_X^B(u)\sigma_Z^B(v) \,=\, \proj{\textrm{EPR}}^{\otimes n}\;.\]
It then follows from~\eqref{eq:n2-3} that $\ket{\psi'} = \ket{\textrm{EPR}}^{\otimes n} \ket{aux} + \ket{\psi''}$ for some $\ket{\psi''}$ such that $\|\ket{\psi''}\|^2 = O(\eps)$. The lemma follows, setting $\Phi_B(X) = VXV^\dagger$ for all $X$.
\end{proof}

\subsection{PoQK for subspace money states}
\label{sec:subspaces}

Our second example of a proof of quantum knowledge protocol is a verification protocol for a modification of Aaronson and Christiano's \emph{subspace states} \cite{aaronson2012quantum}. Aaronson and Christiano present a quantum money scheme in which a $\lambda$-qubit money state, with $\lambda\in\mathbb{N}$ a security parameter, is specified by a  (secret) $(\lambda/2)$-dimensional subspace $A \subseteq \mathbb{Z}_2^\lambda$, and defined as $\ket{A} =\frac{1}{\sqrt{|A|}} \sum_{x \in A} \ket{x}$.  In the notation of Definition~\ref{def:qmoney}, valid bills in this scheme are quadruples $(\mathsf{id, public, secret}, \ket{\$}_{\mathsf{id}})$ such that $\mathsf{id}$ is an arbitrary string, $\mathsf{public}$ is empty\footnote{What we describe here is actually a \emph{private-key} version of the Aaronson-Christiano scheme, equipped with a verification procedure which is similar to a verification procedure used in \cite{ben2016quantum}. Aaronson and Christiano originally proposed this subspace scheme with the idea of making progress towards public-key quantum money. As such, in their original scheme, $\mathsf{public}$ is not empty.}, $\mathsf{secret} = \mZ = \{z_1,\ldots,z_{\lambda/2}\}$ is a basis for a $(\lambda/2)$-dimensional subspace $A$ of $\Z_2^\lambda$, and $\ket{\$}_{\mathsf{id}}=\ket{A}$. One possible (quantum-verifier) verification procedure $\texttt{Ver}(x, \textsf{public}, \textsf{secret}, \rho_W)$ for these bills parses $\textsf{secret}=\mZ$ and then performs the projective measurement $H^{\otimes \lambda} \mathbb{P}_{A^\perp} H^{\otimes \lambda} \mathbb{P}_A$ on $\rho_W$ (where $\mathbb{P}_A$ is a projection onto all standard basis strings in $A$, i.e. $\mathbb{P}_A = \sum_{x \in A} \proj{x}$, and $\mathbb{P}_{A^\perp}$ is a projection onto all standard basis strings in $A^\perp$), and accepts if and only if the outcome is 1. The scheme (when equipped with this verification procedure) has completeness parameter $1$, and it was shown in~\cite{aaronson2012quantum} that its cloning parameter is $\mu_{AC}(\lambda)\leq c^\lambda$ for some constant $c<1$. \footnote{In fact Aaronson and Christiano show the stronger result that this bounds holds even if the adversary is given black-box access to a pair of measurement operators that respectively implement projections on $A$ and $A^\perp$. In our formalization of the cloning game the adversary is not given any access to the verification procedure.}

As we mentioned in the introduction, we do not know if it is possible to devise a natural proof of quantum knowledge for the Aaronson-Christiano subspace states as they have thus far been described. What makes finding such a protocol challenging is that, in contrast with Wiesner's money scheme, there is no obvious classical verification protocol for subspace states in which there can be a single `right answer' to a challenge. We may consider, for example, a classical verification protocol for these states similar to a protocol which was considered in \cite{ben2016quantum}, where the prover is asked to measure all the qubits of a subspace state in either the standard or the Hadamard basis, and any vector $x\in A$ (resp. $y\in A^\perp$) is a valid outcome for a measurement in the standard (resp. Hadamard) basis. It is difficult to argue that such a protocol is a PoQK for subspace states using similar techniques to those which we used for Wiesner states, because the large number of possible `right answers' which the verifier would accept means that the correlations that the prover and the verifier must share if the prover passes are much weaker than those for which we can argue in Protocol \ref{protocol:wiesner}. Nonetheless, we are able to give a proof of knowledge for a version of the subspace scheme in which a (secret) quantum one-time pad has been applied to every subspace state. Due to the use of this one-time pad we are able to apply an analysis that is very similar to the one for Wiesner's scheme. In fact, as we will see, the analysis works even if the subspace $\mZ$ that underlies the construction of the money state is fixed but the one-time pad is chosen uniformly at random. Nevertheless, the resulting scheme is not identical to Wiesner's. In particular, it has the interesting property that the challenge issued by the verifier is only a single bit long.

\paragraph{Notation.}
Before we define the associated AaP scenario, we introduce some notation:
\begin{itemize}
\item Let $\ket{\$}_{v, \theta}$ be a Wiesner money state representing the string $v$ encoded in bases $\theta$, as in~\eqref{eq:wiesner-state}.
\item Let $\{s_i : i \in \{1,\ldots,\lambda\}\}$ $ = \{100...0, 010...0, 001...0, \dots , 000...1\}$ be the standard basis for $\mathbb{Z}_2^\lambda$.
\item Let $\mathcal{Z} = \{z_i : i \in \{1,\ldots,\lambda\}\}$ be a basis for $\mathbb{Z}_2^\lambda$.
\item Let $W$ be the unitary on $(\C^2)^{\otimes \lambda}$ defined as follows:
\begin{align}
W : W\ket{x} &= W\ket{x_1 s_1 + \dots + x_\lambda s_\lambda} \notag\\
&= \ket{x_1 z_1 + \dots + x_\lambda z_\lambda}\;.\label{eq:def-v}
\end{align}
\item Let $L_{\theta}$ for a string $\theta \in \{0,1\}^\lambda$ be the subspace of $\mathbb{Z}_2^\lambda$ whose elements are always 0 in the positions where $\theta_i = 0$, and can be either 0 or 1 in the positions where $\theta_i = 1$.
\item Let $X(a)$ for some vector $a = (a_1, \dots, a_\lambda) \in \mathbb{Z}_2^\lambda$ denote the tensor product of $\lambda$ single-qubit gates which is Pauli $X$ in those positions $i$ where $z_i=1$, and $I$ otherwise. Define $Z(b)$ similarly. Let $X_{\cal{Z}}(d)$, for a basis $\mathcal{Z} = \{z_j\}$, denote the operator
\begin{equation*}
\prod_j \,\big(X(z_j)\big)^{d_j}\;,
\end{equation*}
where $z_j$ denotes a particular vector from the basis set $\cal{Z}$, and $d_j$ denotes the $j$th bit of $d$. Define $Z_{\cal{Z}}(e)$ similarly.
\item Let
\begin{equation}\label{eq:ac-money-state}
\ket{\$}_{v,\theta,\mathcal{Z}} \equiv \frac{1}{\sqrt{|L_\theta|}} \sum_{\lambda \in L_\theta} X_{\cal{Z}}(d) Z_{\cal{Z}}(e) \ket{\lambda_1 z_1 + \dots + \lambda_n z_n}\;,
\end{equation}
with $d_i = v_i$ for $i$ such that $\theta_i = 0$ (and $d_i = 0$ for all other $i$), and $e_i = v_i$ for $i$ such that $\theta_i =1$ (and $e_i = 0$ for all other $i$). Note that the distribution of $\ket{\$}_{v,\theta,\mathcal{Z}} $ over uniform $v, \theta, \mathcal{Z}$ is identical (ignoring global phase) to that of a uniformly random subspace state with a uniformly random Pauli one-time-pad applied to it, because the coordinates of $d$ and $e$ which we have forced to be zero (instead of uniformly random) would only add a global phase to the state.
\end{itemize}

\begin{scenario}
\label{scenario:subspaces}
We instantiate the generic AaP scenario for quantum money described in Section \ref{sec:ap-scenario-qm} as follows:
\begin{itemize}
\item Setup functionality $\mathcal{F}_{AC}(1^\lambda)$: 
\begin{itemize}
\item $\mathtt{init}_{AC}$ initializes a random oracle $H$ taking strings of length $2\lambda+\lambda^2$ to strings of length $2\lambda+\lambda^2$.\footnote{Formally the oracle is implemented in the standard way, recalled in Section~\ref{sec:oracles}.} 
In addition, it initializes an empty database $B_{AC}$ that is destined to contain a record of all quantum money bills in circulation. 
\item  $\mathcal{O}_{\mF_W}(I, \texttt{getId})$: generates $v\in\{0,1\}^\lambda$ and $\theta\in \{0,1\}^\lambda$ such that $|\theta|_H=\frac{n}{2}$ uniformly at random and selects $\mathcal{Z} = \{z_i : i \in \{1,\ldots,\lambda\}\}$ a uniformly random basis for $\mathbb{Z}_2^\lambda$.\footnote{As will be clear from the analysis, we could also publicly fix a basis $\mZ$ a priori and define a scheme that is parameterized by the fixed $\mZ$. The properties of the resulting scheme would be identical to the one that we analyze here. We make the choice of a uniformly random $\mZ$ for convenience.} Sets $\mathsf{id} = H^{-1}((v,\theta,\mZ))$, \footnote{We use $H^{-1}$ and not $H$ here because we specified in Section \ref{sec:ap-scenario-qm} that $H$ maps $\reg{id}$s to $(\reg{public, secret})$ pairs.} $\textsf{secret}=(v,\theta,\mZ)$ and 
 $\ket{\$}_{\mathsf{id}} = \ket{\$}_{v,\theta,\mZ}$ defined in~\eqref{eq:ac-money-state}.
 Adds $(\mathsf{id,public,secret},\ket{\$}_{\mathsf{id}})$ to $B_{AC}$. Returns $\mathsf{id}$.
\item $\mathcal{O}_{\mF_{AC}}(\cdot, \texttt{public}, \mathsf{id})$,  $\mathcal{O}_{\mF_{AC}}(I, \texttt{getMoney}, \mathsf{id})$, and $\mathcal{O}_{\mF_{AC}}(V, \texttt{secret}, \mathsf{id})$: as described in Section~\ref{sec:ap-scenario-qm}. 
\end{itemize}
\item Agreement relation $\mathcal{C}^{\mO_{\mF_{AC}}}_W(1^\lambda, \id)$: The agreement relation is the same as it is in Section \ref{sec:ap-scenario-qm}.
\item Proof relation $\mathcal{R}_{AC}^{\mathcal{O}_{\mathcal{F}_W}}(1^\lambda, x, \rho_W)$: The proof relation firstly queries $\mathcal{O}_{\mathcal{F}_W}(V, \texttt{secret}, x)$ in order to get a tuple $(v, \theta,\mathcal{Z})$. (If $\mathcal{O}_{\mathcal{F}}(V, \texttt{secret}, x)$ returns $\perp$, then $\mathcal{R}$ rejects.) Then it applies $Z(e)X(d)$ to its quantum input $\rho_W$, where $d$ and $e$ are defined in terms of $(v, \theta)$ the same way that they are below equation~\eqref{eq:ac-money-state}. Following that, it defines $A$ to be the subspace generated by the vectors $z_i \in \mZ$ such that $\theta_i = 1$, and then it follows the subspace money verification procedure: it performs the projective measurement $H^{\otimes \lambda} \mathbb{P}_{A^\perp} H^{\otimes \lambda} \mathbb{P}_A$ on $Z(e) X(d) \rho_W$ (where $\mathbb{P}_A$ is a projection onto all standard basis strings in $A$, i.e. $\mathbb{P}_A = \sum_{x \in A} \proj{x}$, and $\mathbb{P}_{A^\perp}$ is a projection onto all standard basis strings in $A^\perp$), and accepts if and only if the outcome is 1.

\end{itemize}
\end{scenario}

\begin{protocol}
\label{protocol:subspaces}
We define a proof of knowledge protocol $\mK_{AC}$ for the scenario $(\mF_{AC}, \mC_{AC}, \mR_{AC})$. The agreement phase is identical to that in Protocol~\ref{protocol:wiesner}, except that now $\mathsf{id}$ has length $2\lambda+\lambda^2$. The second phase is similar but not identical, as the verifier's challenge now consists of a single bit:   
\begin{enumerate}
\item $V_2$ queries $\mO_{\mF_{AC}}(V,\texttt{secret},\mathsf{id})$. If it obtains $\perp$, $V_2$ aborts. Otherwise, let $\$=(v,\theta,\mZ)$ be the classical description obtained.
\item $V_2$ sends a uniformly random bit $c\in\{0,1\}$ to the prover.
\item If $c=0$ the prover $P_2$ measures the state $\rho_{st_P}$ it received from $P_1$ in the standard basis, obtaining a $\lambda$-bit string of outcomes $m\in\{0,1\}^\lambda$, and sends  $m$ to the verifier. If $c=1$ then $P_2$ measures in the Hadamard basis and likewise sends the outcomes $m$ to $V_2$.
\item If $c = 0$ the verifier $V_2$ checks that $m + Wd$ is in the subspace $A$ spanned by $\{z_i:\,\theta_i=1\}$, where $\mZ=\{z_1,\ldots,z_\lambda\}$. If $c = 1$ then $V_2$ checks that $m + W e$ is in $A^\perp$.
\end{enumerate}
Finally, the class of input generation algorithms $\mI_{AC}$ used for completeness is the same as the class $\mI_W$ in Protocol~\ref{protocol:wiesner}.
\end{protocol}

\begin{lemma}
\label{lem:subspaces}
There exists a constant $\kappa < 1$ such that Protocol $\mK_{AC}$ (Protocol~\ref{protocol:subspaces}) is secure with completeness $1$, up to knowledge error $\kappa$, and with extraction distance $\delta=O(\mu^{1/4})$, where $\mu=1-p$ and $p$ the prover's success probability, for the subspace AaP scenario $(\mF_{AC}, \mC_{AC}, \mR_{AC})$ (Scenario~\ref{scenario:wiesner}).
\end{lemma}

\begin{proof}
Completeness of the protocol is clear. For soundness we proceed in a manner analogous to the proof of Lemma~\ref{lem:wiesner}. We give the main steps below and omit some details that are identical to the former proof. 

\paragraph{Step 1: Modeling the prover.} Similarly to the proof of Lemma~\ref{lem:wiesner}, as an ITM the action of $\hat{P}_2$ is realized by a unitary $U$ acting on $\mH_\reg{M} \otimes \mH_\reg{P}$, where $\mH_\reg{M} = \mH_{\reg{M}_1} \otimes \mH_{\reg{M}_2}$ with $\mH_{\reg{M}_1}=\C^2$ a message register associated with the query $c$, $\mH_{\reg{M}_2}=(\C^2)^{\otimes \lambda}$ a message register associated with the response $m$, and $\mH_\reg{P}$ is the prover's secret space. Since $c\in\{0,1\}$ we can assume that $U_c = \proj{0} \otimes U_0+\proj{1}\otimes U_1$ where for each $c$, $U_c$ is a unitary that acts on $\mH_{\reg{M}_2}\otimes \mH_P$. Let $V_0=U_0$ and $V_1 = \big( \prod_i H_i \big) U_1$, where $H_i$ acts on the $i$-th qubit in $\mH_{\reg{M}_2}$. A measurement in the computational basis (if $c=0$) or the Hadamard basis (if $c=1$) of the $\lambda$ qubits in $\mH_{\reg{M}_2}$ after the application of $V_c$ yields the outcome $m$.  
For any $a,b\in \{0,1\}^\lambda$ define observables on $\mH_{\reg{M}_2} \otimes \mH_{\reg{P}}$ as follows:
\begin{equation}\label{eq:def-prover-obs-subspace} 
 X^B(a) \,=\, V_{1}^\dagger (\sigma_X(a)\otimes \Id) V_{{1}}\qquad\text{and}\qquad Z^B(b) \,=\, V_{0}^\dagger (\sigma_Z(b)\otimes \Id) V_{0}\;,
\end{equation}
where the identities act on $\mH_\reg{P}$. The observables $X^B(a)$ and $Z^B(b)$ are observables that can be implemented with black-box access to an ITM implementing the prover, by initializing $\mH_\reg{M}$ to $\ket{c}\otimes \ket{0^\lambda}$ for $c=0$ to implement $Z^B(a)$ and $c=1$ to implement $X^B(b)$, executing the prover's unitary $U$ (which acts on $\mH_{\reg{M}} \otimes \mH_{\reg{P}}$) followed by $\prod_i H_i^c$ on $\mH_{\reg{M}_2}$, measuring $\sigma_Z(a)$ or $\sigma_X(b)$ on $\mH_{\reg{M}_2}$, and executing $(\prod_i H_i^c)U^\dagger$. 

\paragraph{Step 2: The prover's observables.} 
Similarly to the proof of Lemma~\ref{lem:wiesner}, we introduce a purified description of the oracle $\mO_{\mF_{AC}}$, which we denote $\mO_{\mF'_{AC}}$. The only relevant changes are: 
\begin{itemize}
\item $\mathcal{O}_{\mF'_{AC}}(I, \texttt{getId})$: generate a uniformly random $\mathsf{id}\in \{0,1\}^{2\lambda}$. Create $2\lambda$ EPR pairs $\ket{\phi}_{\mathsf{id}} = \ket{\mathrm{EPR}}_{\reg{A}_{\mathsf{id}}\reg{B}_{\mathsf{id}}}^{\otimes 2\lambda}$ and store $(\mathsf{id},\texttt{unmeasured},\ket{\phi}_{\mathsf{id}})$ in the database $B_{AC}'$. Return $\mathsf{id}$.
\item $\mathcal{O}_{\mF'_{AC}}(R, \texttt{getMoney}, \mathsf{id})$ for $R\in\{I,P\}$: If no record in $B_{AC}'$ with identifier \textsf{id} exists, returns $\perp$. Otherwise, return the second half of each EPR pair in $\ket{\phi}_{\mathsf{id}}$, i.e. register $\reg{B}_{\mathsf{id}}$, the first time it is called. If called again with the same \textsf{id} argument, returns $\perp$. 
\item $\mathcal{O}_{\mF'_{AC}}(V, \texttt{secret}, \mathsf{id})$:  If no record in $B_{AC}'$ with identifier \textsf{id} exists, returns $\perp$. If there is a record of the form  $(\mathsf{id},\texttt{measured},x)$ in $B'_{AC}$ then return $x$. Otherwise, the record is of the form $(\mathsf{id},\texttt{unmeasured},\sigma_{\reg{A}_\id\reg{B}_\id})$. In this case, 
select a uniformly random $\theta\in\{0,1\}^\lambda$ such that $|\theta|_H=\frac{n}{2}$ and a basis $\mZ$ for $\mathbb{Z}_2^\lambda$ and measure register $\reg{A}_{\mathsf{id}}$ of $\sigma$ using the POVM $E(\theta,\mZ)=\{W\proj{\$}_{v,\theta} W^\dagger|v\in\{0,1\}^\lambda\}$, with $W$ defined in~\eqref{eq:def-v} and $\ket{\$}_{v,\theta}$ as in~\eqref{eq:wiesner-state}
 to obtain an outcome $v$. Return $v$. Update the database entry to $(\mathsf{id},\texttt{measured},(v,\theta,\mZ))$.
\end{itemize}
The operation of $\mathcal{O}_{\mF'_{AC}}$ is indistinguishable from $\mO_{\mF_{AC}}$ from the point of view of the parties $I$ or $P$, because measuring  the half-EPR pairs in $\reg{A}_{\mathsf{id}}$ using the POVM $E(\theta,\mZ)$ for uniformly chosen $\theta,\mZ$ has the effect of projecting the state in $\reg{A}_{\mathsf{id}}$ to $\ket{\$}_{v,\theta,\mZ}$ for uniformly random $v,\theta,\mZ$. In more detail, observe that for any $\theta$ and $v$,
\begin{align*}
W \ket{\$}_{v,\theta}  &= W \Big(\prod_i H_i^{\theta_i}\Big) \ket{v}\\
&= W \Big(\prod_i H_i^{\theta_i}\Big) X(v) \ket{0}  \\
&= W  X(v\cdot \bar{\theta}) Z(v\cdot \theta) \Big(\prod_i H_i^{\theta_i}\Big)  \ket{0} \\
&= W  X(d) Z(e) \ket{L_\theta}  \\
&=   X_\mZ(d) Z_\mZ(e) W \ket{L_\theta}  \\
&= X_\mZ(d) Z_\mZ(e) \ket{A} \\
&= \ket{\$}_{v,\theta,\mZ}\;,
\end{align*}
where we let $\ket{L_\theta} = |L_\theta|^{-1/2} \sum_{\lambda\in L_\theta} \ket{\lambda}$ and $d = v\cdot \bar{\theta}$, $e = v\cdot \theta$ with $\cdot$ denoting entrywise product. 

At the end of the first phase of the protocol, let $\mathsf{id}=st_P=st_V$ be the identifier upon which $V_1$ and $\hat{P}_1$ agreed (if there is no agreement then we do not need to specify the extractor). For $a,b\in\{0,1\}^\lambda$ define $Z^A(b) = \sigma_Z(b)$ and $X^A(a) = \sigma_X(a)$, where the $\sigma_Z$ and $\sigma_X$ operators act on the register  $\reg{A}_{\mathsf{id}}$ associated with the entry $\mathsf{id}$ in the database $B_W'$ maintained by $\mF'_W$ in the purified version of the scenario. Let $\ket{\psi} \in \mH_{\reg{A}_{\mathsf{id}}}\otimes \mH_{\reg{M}_2} \otimes \mH_\reg{P}$ be the joint state of the registers $\reg{A}_{\mathsf{id}}$ (for the fixed $\mathsf{id}$ above) and the prover at the start of the second phase of the protocol, with $\reg{M}_2$ the register destined to receive the prover's answer $m$. Without loss of generality we assume that this state is pure.

\paragraph{Step 3: Correlations.}
Suppose that the probability that $\hat{P}_2$ succeeds in the proof phase of the protocol is at least $1 - \mu$. Then the probability that $\hat P_2$ succeeds conditioned on $c=0$ is at least $1 - 2\mu$. `Success' in the $c=0$ case means the following. Let $m$ be the outcome obtained by the prover after applying $V_0$ and measuring in the standard basis. Let $(v,\theta,\mathcal{Z})$ be as above. Let $d = v\cdot \bar{\theta}$. Then the condition checked by the verifier in the protocol is that $m + Wd \in A = \text{span}\{ z_i:\, \theta_i = 1\}$. Equivalently, $ W^\dagger (m + W d) \in W^\dagger A = L_\theta$, which using the definition of $d$ is equivalent to the condition that $(W^\dagger m)_i = v_i$ for all $i$ such that $\theta_i = 0$. Observe that if the oracle, on call $\mathcal{O}_{\mF'_{AC}}(V, \texttt{secret}, \mathsf{id})$ made by $V_2$, had measured all qubits directly in the computational basis instead of using the POVM $E(\theta,\mZ)$, it would have obtained a string $w$ such that $(W^\dagger w)_i = v_i$ whenever $\theta_i=0$. Using the above, it follows that $(W^\dagger m)_i = (W^\dagger w)_i$ for all $i$ such that $\theta_i=0$. 

Recall the definition of $X^B(a)$ and $Z^B(b)$ in~\eqref{eq:def-prover-obs-subspace}, and let $\tilde{X}^B(a') = X^B(W^\dagger a')$ and $\tilde{Z}^B(b')=Z^B(W^\dagger b')$. Similarly let $\tilde{X}^A(a')= X^A(W^\dagger a')=W^\dagger \sigma_X(a')W$ and $\tilde{Z}^A(b')= Z^A(W^\dagger b')=W^\dagger\sigma_Z(b')W$. The preceding discussion shows that on average over uniformly random $\theta$ of weight $\frac{\lambda}{2}$ it holds that
\begin{equation}\label{eq:cor-sub-1}
\Es{\theta}\, \bra{\psi} \tilde{Z}^A(\bar{\theta}) \otimes \tilde{Z}^B(\bar{\theta}) \ket{\psi} \geq (1)(1-2\mu)+(-1)(2\mu) \,\geq\, 1-4\mu\;,
\end{equation}
and similarly
\begin{equation}\label{eq:cor-sub-2}
\Es{\theta}\,\bra{\psi} \tilde{X}^A(\theta) \otimes \tilde{Z}^B(\theta) \ket{\psi} \geq (1)(1-2\mu)+(-1)(2\mu) \,\geq\, 1-4\mu\;.
\end{equation}
In both equations the left-hand side should be understood as an average  over the random choice $\mZ$. By an averaging argument, for any given $\mZ$ the same equations will hold with right-hand side $1-4\mu_\mZ$ for some parameter $\mu_{\mZ}$ such that $\Es{\mZ} \mu_\mZ \leq 2\mu$. For the time being, fix any $\mZ$. 

\paragraph{Step 4: Rigidity.}
Setting $\eps=8\mu_\mZ$ it follows from~\eqref{eq:cor-sub-1} and~\eqref{eq:cor-sub-2} that the assumption~\eqref{eq:n2-0} of Lemma~\ref{lem:gh} is satisfied by the observables $\tilde{Z}^B$ and $\tilde{X}^B$ acting on the state $\ket{\tilde{\psi}} = W_\reg{A} \otimes \Id \ket{\psi}$. Note that here we use case (i) of the ``finally'' part of the lemma, which allows us to restrict the condition to uniform expectation over strings of Hamming weight $\frac{\lambda}{2}$.
The conclusion of the lemma implies that there is an isometry $\tilde{\Phi}^B_\mathcal{Z} : \mH_\reg{B} \to ((\C^2)^{\otimes \lambda})_{\reg{B'}} \otimes \mH_{\reg{\hat{B}}}$, where we used $\reg{B}$ to denote the concatenation of $\reg{M}_2$ and $\reg{P}$, such that
\begin{equation}
\label{eq:EPR}
\Tr\Big( \big(\bra{\mathrm{EPR}}_{\reg{AB'}}^{\otimes
      \lambda}\otimes \bra{aux}_{\hat{\reg{B}}}\big)\,\big( W_\reg{A} \otimes \tilde{\Phi}^B_\mathcal{Z}(\proj{\psi}_{\reg{AB}})\big)\big( \ket{\mathrm{EPR}}_{\reg{AB'}}^{\otimes \lambda} \otimes \ket{aux}_{\hat{\reg{B}}}\big) \Big) = 1 - O\big(\sqrt{\mu_\mZ}\big)\;.
\end{equation}
Moreover, the isometry $\tilde{\Phi}^B_\mZ$ is defined as  
\[ \tilde{\Phi}^B_\mZ:\,\ket{\varphi}_\reg{B} \mapsto \Big(\frac{1}{2^\lambda} \sum_{a,b}\, \Id_{\reg{B}'} \otimes   \sigma_X(a) \sigma_Z(b) \otimes \tilde{X}^B(a) \tilde{Z}^B(b) \Big) \ket{\mathrm{EPR}}_{\reg{B}'\reg{A}'}^{\otimes \lambda}\ket{\varphi}_{\reg{B}} \;,\]
so $\reg{\hat{B}}$ in the output of $\tilde{\Phi}^B$ is the concatenation of $\reg{A}'$ and $\reg{B}$.
We have
\begin{align*}
\tilde{\Phi}^B_\mZ \ket{\varphi}_\reg{B} &=\Big( \frac{1}{2^\lambda} \sum_{a,b} \,\Id_{\reg{B}'}\otimes\sigma_X(a) \sigma_Z(b) \otimes \tilde{X}^B(a) \tilde{Z}^B(b)  \Big) \ket{\mathrm{EPR}}_{\reg{B}'\reg{A}'}^{\otimes \lambda}\ket{\varphi}_\reg{B}\\
&= \Big( \frac{1}{2^\lambda} \sum_{a,b} \,\Id_{\reg{B}'}\otimes \sigma_X(a) \sigma_Z(b) \otimes  {X}^B(W^\dagger a) {Z}^B(W^\dagger b)  \Big) \ket{\mathrm{EPR}}_{\reg{B}'\reg{A}'}^{\otimes \lambda} \ket{\varphi}_\reg{B}\\
&= \Big( \frac{1}{2^\lambda} \sum_{a',b'} \,\Id_{\reg{B}'}\otimes  \sigma_X(W a') \sigma_Z(W b') \otimes {X}^B(a') {Z}^B(b')  \Big)\ket{\mathrm{EPR}}_{\reg{B}'\reg{A}'}^{\otimes \lambda}\ket{\varphi}_\reg{B}\\
&=  \Big( \frac{1}{2^\lambda} \sum_{a',b'} \, W_{\reg{B}'}\otimes   W\sigma_X( a') \sigma_Z( b') \otimes{X}^B(a') {Z}^B(b') \Big) \ket{\mathrm{EPR}}_{\reg{B}'\reg{A}'}^{\otimes \lambda}\ket{\varphi}_\reg{B}\;,
\end{align*}
where for the last equality we used that $\sigma_X(W a')= W \sigma_X(a') W^\dagger$ and 
\[ W^\dagger \otimes \Id \ket{\mathrm{EPR}}^{\otimes \lambda} = \Id \otimes \overline{W} \ket{\mathrm{EPR}}^{\otimes \lambda} = \Id \otimes W \ket{\mathrm{EPR}}^{\otimes \lambda}\]
since $W$ has coefficients in $\{0,1\}$. This shows that if we let $\Phi^B$ be defined as 
\[ {\Phi}^B:\,\ket{\varphi}_\reg{B} \mapsto \Big(\frac{1}{2^\lambda} \sum_{a,b} \,   \Id_{\reg{B}'}\otimes  \sigma_X(a) \sigma_Z(b) \otimes {X}^B(a) {Z}^B(b) \Big) \ket{\mathrm{EPR}}_{\reg{B}'\reg{A}'}^{\otimes \lambda} \ket{\varphi}_\reg{B}\]
then $\Tr_{\reg{A}'}\Phi^B$ does not depend on $\mZ$ and moreover 
\begin{equation}\label{eq:phi-tildephi}
\Tr_{\reg{A}'}\Phi^B(\cdot) \,=\, W_{\reg{B}'}\Tr_{\reg{A}'} \tilde{\Phi}^B_\mZ(\cdot) W_{\reg{B}'}^\dagger\;.
\end{equation}
 Selecting a random $\theta\in\{0,1\}^\lambda$ and measuring register $\reg{A}$ of state $\ket{\psi}$ using the POVM $E(\theta,\mZ)$ as would be done by $\mathcal{O}_{\mF'_W}(V, \texttt{secret}, \mathsf{id})$, by contractivity of the trace distance~\eqref{eq:EPR} implies that on average over uniformly random $\theta$ and $v$ ($v$ is uniformly random because we know that the reduced density of $\ket{\psi}$ on $\reg{A}_{\mathsf{id}}$ is totally mixed),
\begin{equation}\label{eq:subspace-close-2}
\big\| \Tr_{\reg{A'}} \tilde{\Phi}^B_\mathcal{Z} \Tr_{\reg{A}} \big( (\proj{\$}_{v, \theta, \mathcal{Z}} \otimes \Id)\proj{\psi} \big)- \ket{\$}_{v, \theta} \ket{aux}_{\reg{B}} \big\|_1 \,=\, O\big(\epsilon^{1/4}\big)\;.
\end{equation}
Equivalently, using~\eqref{eq:phi-tildephi},
\begin{equation}\label{eq:subspace-close-2b}
\big\| \Tr_{\reg{A}'} {\Phi}^B \Tr_{\reg{A}} \big( (\proj{\$}_{v, \theta, \mathcal{Z}} \otimes \Id)\proj{\psi} \big)- \ket{\$}_{v, \theta,\mathcal{Z}} \ket{aux}_{\reg{B}} \big\|_1 \,=\, O\big(\epsilon^{1/4}\big)\;.
\end{equation}
since $\ket{\$}_{v, \theta,\mathcal{Z}} = W \ket{\$}_{v,\theta}$, where $\ket{\$}_{v, \theta}$ is the Wiesner money state associated with $(v, \theta)$. 

\paragraph{Step 5: Definition of $E$.}
Consider $E$ that uses its access to $\hat{P}_2$ in order to implement the isometry $\Phi^B$. This is possible using black-box access to $\hat{P}_2$. Then, $E$ traces out the register $\reg{\hat{B}}$ and returns the state contained in $\reg{B}'$. Using~\eqref{eq:subspace-close-2b} it follows that the resulting state is $O(\mu_\mZ^{1/4})$ close to $\ket{\$}_{v, \theta,\mZ}$, and in particular must be accepted by the verification procedure $\mR_{AC}$ with probability $1-O(\mu_\mZ^{1/4})$. Averaging over a uniformly random $\mZ$ and using Jensen's inequality and $\Es{\mZ} \mu_\mZ \leq 2\mu$ we deduce that the same holds on average over $\mZ$, with probability $1-O(\mu^{1/4})$.
 Choosing $\kappa$ to be any constant $\kappa < 1 $ such that this quantity is positive for all $p=1-\mu \geq \kappa$ completes the proof.
\end{proof}

\section{Arguments of Quantum Knowledge for QMA relations}
\label{sec:qma-arguments}

The main result of this section is Theorem~\ref{thm:meas}, which gives a protocol (with classical verifier) for the verification of any $\QMA$ relation, a natural analogue of the notion of $\NP$ relation the definition of which we recall in Section~\ref{sec:qma-scenario} below. Since the soundness of this protocol in general only holds against QPT provers we refer to it as a ``classical argument of quantum knowledge'' for any $\QMA$ relation. In addition we note that for general $\QMA$ relations the completeness property of the associated AaP protocol requires the honest prover to be given multiple copies of the $\QMA$ witness in order to succeed with high probability. It is still possible for completeness to hold with a single witness if one is willing to assume that the $\QMA$ relation has been pre-processed to a specific form; see the statement of Theorem~\ref{thm:meas} below.

Our construction is based on the classical verification protocol for $\QMA$ introduced in~\cite{measurement}, which we review in Section~\ref{sec:meas-protocol}. Before doing so we introduce the Agree-and-Prove scenario.

\subsection{Agree-and-Prove scenario for $\QMA$ relations}
\label{sec:qma-scenario}

We first recall the quantum extension of an $\NP$ relation $\mR$, following~\cite{coladangelo2019non,broadbent2019zero}. 

\begin{definition}[$\QMA$ relation]
  \label{def:qma-relation}
	A \emph{$\QMA$ relation} is specified by a triple $(Q,\alpha,\beta)$ where $\alpha,\beta:\N\to[0,1]$ satisfy $\beta(n)\leq \alpha(n)$ for all $n\in\N$ and $Q = \{Q_n\}_{n\in\N}$ is a uniformly generated family of quantum circuits such that for every $n$, $Q_n$ takes as input a string $x\in\{0,1\}^n$ and a quantum state $\ket{\psi}$ on $p(n)$ qubits (i.e.\ $Q_n$ takes $n+p(n)$ input qubits for some polynomial $p$ that is implicitly specified by $Q$, and is assumed to immediately measure its first $n$ input qubits in the computational basis) and returns a single bit.
	\end{definition}
	
	To a $\QMA$ relation $(Q,\alpha,\beta)$ we associate two sets
\[R_{Q,\alpha} = \bigcup_{n\in \N} \Big\{(x,\sigma)\in \{0,1\}^n\times \Density(\C^{p(n)})\,\big|\;  \Pr(Q_{n}(x, \sigma)=1)\geq \alpha\Big\}\;
\]
and
\[N_{Q,\beta} = \bigcup_{n\in \N} \Big\{x\in \{0,1\}^n\,\big|\; \forall \sigma \in  \Density(\C^{p(n)})\,,\;  \Pr(Q_{n}(x, \sigma)=1)< \beta\Big\}\;.\]
 We say that a (promise) \emph{language $L=(L_{yes},L_{no})$ is specified by the $\QMA$ relation $(Q,\alpha,\beta)$} if
	\begin{equation}\label{eq:l-qma}
	L_{yes} \subseteq \Big\{x \in \{0,1\}^* \big|\; \exists \sigma \in \Density(\C^{p(n)})\,, \; (x,\sigma)\in R_{Q,\alpha}\Big\}\;,
	\end{equation}
	and $L_{no} \subseteq N_{Q,\beta}$. Note that, whenever $\alpha - \beta > 1/\poly(n)$, any language $L$ that is specified by $(Q,\alpha,\beta)$ lies in $\QMA$. Conversely, any language in $\QMA$ is specified by some $\QMA$ relation in a straightforward (non-unique) way.

\paragraph{The local Hamiltonian problem}
In the following, we make use of Kitaev's circuit-to-Hamiltonian construction~\cite{kitaev2002classical,kempe20033}, which associates with any promise language $L=(L_{yes},L_{no})\in \QMA$ and $x \in L_{yes}\cup L_{no}$ an instance of the \emph{local Hamiltonian problem}. An instance of the local Hamiltonian problem is specified by a local Hamiltonian operator $H$  and two real numbers $\alpha > \beta$. The instance is a `YES instance' if $H$ has smallest eigenvalue at most $\alpha$, and a `NO instance' if it is at least $\beta$. We give the name `ground state' to those states $\rho$ such that $\Tr(H \rho)$ takes on its minimum value, and we may refer to the space of all ground states of $H$ (if there is more than one) as the `ground space' of $H$. We call the minimum eigenvalue of $H$ its `ground energy'. 

\paragraph{Agree-and-Prove scenario.}
Fix a $\QMA$ relation $(Q,\alpha,\beta)$. We associate an AaP scenario to $Q$ as follows. 

\begin{itemize}
\item Setup functionality $\mF_Q(1^\lambda)$. We consider a ``trivial'' setup, i.e.\ the initialization procedure does nothing and there is no associated oracle $\mO_{\mF_Q}$. 
\item Agreement relation $\mC_Q(1^\lambda,x)$: returns $1$ for any $\lambda$ and $x$.\footnote{The agreement relation does not even require that $x\in R_{Q,\alpha}\cup N_{Q,\beta}$, as in general this cannot be efficiently verified.} 
\item Proof relation $\mR_Q(1^\lambda,x,\rho)$: executes the verification circuit $Q_{|x|}$ on the pair $(x,\rho)$ and returns the outcome. 
\end{itemize}

We end by discussing some assumptions on a $\QMA$ relation under which our results will hold. Let $(Q,\alpha,\beta)$ be a $\QMA$ relation. We require that the relation satisfies the following properties: 
\begin{enumerate}[label=(\textbf{Q.\arabic*})]
\item \label{enu:qma-a1}The completeness parameter $\alpha$ is negligibly close to $1$, and the soundness parameter $\beta$ is bounded away from $1$ by an inverse polynomial. 
	\item \label{enu:qma-a2} For any $x\in \{0,1\}^n$ there is a local Hamiltonian $H=H_x$ that is efficiently constructible from $x$ and satisfies the following. First, we assume that $H$ is expressed as a linear combination of tensor products of Pauli operators with real coefficients chosen such that $-\Id \leq H\leq \Id$. Second, whenever there is $\sigma$ such that $(x,\sigma)\in R_{Q,\alpha}$, then  $\Tr(H\sigma)$ is negligibly close to $-1$ and moreover any $\sigma$ such that  $\Tr(H\sigma)\leq -1 + \delta$ satisfies $\Pr(Q_{|x|}(x,\sigma)=1)\geq 1-r(|x|)q(\delta)$ for some polynomials $q,r$ depending on the relation only. Third, whenever $x\in N_{Q,\beta}$ then the smallest eigenvalue of $H$ is larger than $-1+1/s(|x|)$, where $s$ is another polynomial depending on the relation only. 
	\end{enumerate}
The first assumption~\ref{enu:qma-a1} is benign and can be made without loss of generality by applying standard amplification techniques such as~\cite{marriott2005quantum} to the $\QMA$ verification circuit, followed by the circuit-to-Hamiltonian construction to obtain $H$. The second assumption~\ref{enu:qma-a2} is somewhat more restrictive. For any QMA relation $(Q, \alpha, \beta)$, the existence of Hamiltonians $H=H_x$ satisfying all claimed properties follows from Kitaev's circuit-to-Hamiltonian construction, in its amplified form used in~\cite[Protocol 8.3]{mahadev2018classical}. However, this construction is not in general ``witness-preserving'' in the sense described in the second assumption above: to construct an eigenstate of the Hamiltonian $H$ with small enough eigenvalue one may need to use many copies of a witness $\sigma$ for the QMA verification procedure $Q(x, \cdot)$. Hence depending on the $\QMA$ relation $Q$ one is interested in, completeness of our Agree-and-Prove scenario for it will hold with respect to the class of input generation algorithms $\mI_Q^{(\ell)}$ described in Section~\ref{sec:qma-agree-protocol} for sufficiently large (polynomial) $\ell\geq 1$. 

	\subsection{The protocol}
	\label{sec:qma-protocol}
	
	Before giving the protocol associated with the Agree-and-Prove scenario introduced in the previous section, in the following subsection we recall the high-level structure of the verification protocol from~\cite{measurement}, on which our protocol will be based. 
	
	\subsubsection{The verification protocol from~\cite{measurement}}
	\label{sec:meas-protocol}
	
	In the protocol from~\cite{measurement}, which we will refer to as the \emph{verification protocol}, the input to the verifier is an $n$-qubit Hamiltonian $H$ that is expressed as a linear combination of tensor products of $\sigma_X$ and $\sigma_Z$ Pauli operators. The input to the prover is a ground state of $H$. Both parties also receive a security parameter $\lambda$. 
At a high level, the verification protocol has the following structure: 
\begin{enumerate}
\item The verifier selects a \emph{basis string} $h\in\{0,1\}^n$ according to a distribution that depends on $H$. The verifier then randomly samples a pair of keys, $(pk,sk)$, consisting of a public key $pk$ and secret key $sk$. (The distribution according to which $(pk,sk)$ is sampled depends on $h$.) The choice of keys specifies an integer $w$ of size $\poly(n,\lambda)$. The verifier sends $pk$ to the prover. 
\item The prover returns an $n$-tuple of \emph{commitment strings} $y=(y_1,\ldots,y_n)$, where each $y_i$ lies in some alphabet $\mathcal{Y}$. 
\item The verifier selects a \emph{challenge bit} $c\in \{0,1\}$ and sends $c$ to the prover. 
\item If $c=0$ (``test round''), the prover returns a string $b\in\{0,1\}^n$ and $x_1,\ldots,x_n\in \{0,1\}^w$. If $c=1$ (``Hadamard round''), the prover returns a string $b\in\{0,1\}^n$ and $d_1,\ldots,d_n\in\{0,1\}^w$. 
\item In case $c=0$ the verifier uses $pk$, $y$, $b$ and $x_1,\ldots,x_n$ to make a decision to accept or reject. (In a test round the verifier never checks any properties of the prover's state; it only checks that the prover is, loosely speaking, doing the correct operations.) In case $c=1$ the verifier uses $sk$ to decode $y, b$ and $d_1,\ldots,d_n$ into \emph{decoded measurement outcomes} $(m_1,\ldots,m_n)\in\{0,1\}^n$. (For the case of a honest prover, the decoded outcomes $m$ correspond to the outcomes of measuring a ground state of $H$ in the bases indicated by $h$, with $h_i=0$ indicating that the $i$th qubit should be measured in the computational basis and $h_i=1$ that the $i$th qubit should be measured in the Hadamard basis. The prover remains ignorant throughout the entire protocol of the verifier's choice of $h$.)
\item In case $c=1$ the verifier makes a decision based on the decoded measurement outcomes and the instance $x$, as described in~\cite[Protocol 8.1]{mahadev2018classical}. 
\end{enumerate}

To model the verifier and prover in the protocol as ITM, which will be required to establish the extractor needed to show soundness, we introduce registers associated with each party and the messages that they exchange in the protocol. Let $\reg{K}$ and $\reg{C}$ denote registers that contain the verifier's first and second messages respectively, i.e.\ the key $pk$ and the challenge bit $c$. Let $\reg{T}$ denote the verifier's private space. Let $\reg{Y}$ denote the register measured by the prover to obtain the prover's first message $y$, and $\reg{M}$ the register measured to obtain the prover's second message $(b,x_1,\ldots,x_n)$ or $(b,d_1,\ldots,d_n)$, depending on $c=0$ or $c=1$ respectively. Let $\reg{S}$ denote the prover's private space.

The natural description of the prover as an ITM consists of 
 (i) its initial state $\sigma \in \Density(\mH_{\reg{YMS}})$, (ii) a unitary $V_0$ acting on $\reg{KYMS}$, and (iii) two unitaries $V$ and $V'$ acting on $\reg{MS}$, where $V$ is the action of the prover on challenge $c=0$ and $V'$ its action on challenge $c=1$. In either case the register $\reg{M}$ is measured in the computational basis to obtain the prover's answer.\footnote{This description slightly departs from the ``canonical'' formalism introduced in Section~\ref{sec:black-box} by using different notations for the prover's unitaries associated with different rounds as well as different challenges. This is for convenience and it is not hard to find an equivalent description to the one given here that uses the language from Section~\ref{sec:black-box}. In this case, the four registers $\reg{KYCM}$ would all be considered network registers, and are thus accessible to the extractor.} 

For convenience we introduce a slightly different representation of the prover, that matches the presentation from~\cite{measurement} and which can be straightforwardly simulated given black-box access to the natural representation described in the previous paragraph. First, we replace $V_0$ by the unitary $U_0 = VV_0$. Note that this is well-defined and does not change the prover's first message, since $V$ does not act on $\reg{Y}$. Second, we define $U = H^{\otimes (n+nw)} V'V^\dagger $, where the Hadamard gates act on the $(n+nw)$ qubits in register $\reg{M}$. It is then immediate that given a natural representation of the prover as three unitaries $(V_0,V,V')$ the pair of unitaries $(U_0,U)$ provides a different representation of the same prover, who now behaves as follows: 
\begin{enumerate}
\item Upon reception of $pk$, the prover applies $U_0$ to its initial state (to which $\ket{pk}$ has been appended), measures the first $n\log|\mathcal{Y}|$ qubits in the computational basis and returns the outcome;
\item Upon reception of $c=0$, the prover directly measures the first (remaining) $n+nw$ qubits in the computational basis and returns the outcome;
\item Upon reception of $c=1$, the prover applies the unitary $U$,  measures the first (remaining) $n+nw$ qubits in the Hadamard basis and returns the outcome.
\end{enumerate}
In both cases $c=0$ and $c=1$ we denote the first $n$ qubits measured by the prover (in step 2 or in step 3, respectively), whose associated measurement outcomes are denoted by $b$ in the protocol, the \emph{committed qubits}. 

\subsubsection{The Agree-and-Prove protocol}
\label{sec:qma-agree-protocol}

In this section we define a protocol $\mK_Q$ for the AaP scenario $(\mF_Q,\mC_Q,\mR_Q)$ associated to a $\QMA$ relation $(Q,\alpha,\beta)$ as in Section~\ref{sec:qma-scenario}. 
Recall that an Agree-and-Prove protocol consists of two phases, an ``agree'' phase and a ``prove'' phase. The agree phase in protocol $\mK_Q$ is simple:
\begin{itemize}
\item The prover $P_1$ takes as input $1^\lambda$ and a CQ state $\rho_{\reg{AUX}_P}$. It interprets the classical part of $\rho$ as a string $z\in \{0,1\}^n$ and the quantum part as $\ell$ witnesses $\sigma_1,\ldots,\sigma_\ell$ each of the same number of qubits.  (We assume that the integers $n$ and $\ell$ are both encoded in a canonical way in the state $\rho_{\reg{AUX}_P}$.) It sends $z$ to the verifier and outputs the statement $x_p=z$ and the quantum state $\rho_{st_P}=(\sigma_1,\ldots,\sigma_\ell)$ (which may in general be entangled).
\item The verifier $V_1$ takes as input $1^\lambda$ and a classical auxiliary input $\rho_{\reg{AUX}_v}$. It parses $\rho_{\reg{AUX}_v}$ as the specification (in binary) of an input length $n$ followed by a string $x\in\{0,1\}^n$. It receives $z$ from $P_1$. If $z\neq x$ it aborts. Otherwise, it produces the statement $x_v=x$. 
\end{itemize}
For the proof phase $V_2$ and $P_2$ behave exactly as the verifier and prover in the verification protocol described in Section~\ref{sec:meas-protocol}, first defining the Hamiltonian $H_v$ and $H_p$ from their respective statements $x_v$ and $x_p$ according to assumption~\ref{enu:qma-a2}. Note that $H_v$ (resp. $H_p$) acts on $\poly(n)$ qubits, with $n=|x_v|$ (resp. $n=|x_p|$). Of course, in case of a honest implementation of the protocol it holds that $x_v=x_p$.

To complete the description of the protocol we define a class of of input-generation algorithms under which completeness holds. 
We consider only input generation algorithms that generate positive instances of the language, accompanied with $\ell$ copies of a valid proof, where $\ell\geq 1$ is a parameter. That is, for any $\ell\geq 1$, $\mI_Q^{(\ell)}$ contains any input generation algorithm $I$ that returns a CQ state of the form
\begin{equation}\label{eq:input-qma}
\sum_{x\in \{0,1\}^*} p_x \proj{|x|,x}_{\reg{AUX}_V} \otimes \big(\proj{x} \otimes \sigma_x^{\otimes \ell} \big)_{\reg{AUX}_P}\;,\footnote{For clarity we omit explicitly writing out $|x|$ in both registers.}
\end{equation}
 where $(p_x)$ is any distribution over positive instances for the $\QMA$ relation, i.e.\ the set 
\[  \big\{x:\exists \sigma, (x,\sigma)\in R_{Q,\alpha}\big\}\;,\]
and moreover for each $x$, $\sigma_x$ is such that $(x,\sigma_x) \in R_{Q,\alpha}$.\footnote{Our definition requires the $\ell$ copies of the witness to be identical. One could generalize this to tensor products of different witnesses, or even entangled states that satisfy certain conditions; since we cannot think of a natural setting where the generalization is useful we do not consider it here.}

\subsection{Construction of an extractor}
\label{sec:qma-extractor}

	To show soundness of the protocol described in the previous section we ultimately have to describe an extractor that succeeds in the soundness experiment from Section~\ref{sec:security}. Before giving the final extractor, which will be done in the next section, in this section we review parts of the analysis given in~\cite{measurement} that will be relevant for our construction. 
	
	We first quote a claim from~\cite{mahadev2018classical}. To make the claim comprehensible we need the following definition. 
	
	\begin{definition}
	\label{def:trivial-prover}
A prover in the verification protocol described in Section~\ref{sec:meas-protocol}, represented as described in that section by an initial state $\sigma$ and a pair of unitaries $(U_0,U)$, is called ``trivial'' (implicitly, for a given input Hamiltonian to the protocol) if two conditions hold: (i) the prover's probability of being accepted in a test round (case $c=0$) is negligibly close to $1$, and (ii) the prover's unitary $U$ commutes with a computational basis measurement of the first $n$ qubits in register $\reg{M}$ (i.e.\ the committed qubits).
	\end{definition}
	
	\begin{claim}[Claim 5.7 in~\cite{mahadev2018classical}]\label{claim:57}
	For all trivial provers $P$, there exists an $n$-qubit state $\rho$ [which can be created from the prover's initial state using a polynomial-size quantum circuit] such that for all $h\in\{0,1\}^n$, the distribution over measurement results produced in the protocol with respect to $P$ for basis choice $h$ is computationally indistinguishable from the distribution which results from measuring $\rho$ in the basis determined by $h$ [with $h_i=0$ corresponding to a computational basis measurement and $h_i=1$ a Hadamard basis measurement].
	\end{claim}
	
	The statements in brackets in the claim has been inserted by us. 
	As will soon be made clear, the ``polynomial-size circuit'' referred to in this statement is a circuit that is obtained by small modifications of the provers' actions in the protocol (i.e.\ the unitaries $U_0$ and $U$ described in Section~\ref{sec:meas-protocol}). The claim of ``computational indistinguishability'' is based on a cryptographic assumption that underlies soundness of the verification protocol from~\cite{measurement}: informally, the distribution of outcomes is ``indistinguishable'' from measurements on $\rho$ \emph{to any (classical or quantum) computationally bounded adversary}, assuming that the Learning with Errors (LWE) problem is intractable for quantum polynomial-time procedures. We refer to~\cite{measurement} for more details on this computational assumption. 
	
	For our purposes we need a slightly modified version of Claim~\ref{claim:57}, stated below. 
	
		\begin{claim}\label{claim:57-new}
		There exist constants $c_1,C_1>0$ such that the following holds. 
		Let $H$ be a Hamiltonian. Let $P$ be a prover that is accepted in the verification protocol associated with $H$, conditioned on a test round ($c=0$), with probability $1-\eps$, for some $\eps \geq 0$. Then there exists an $n$-qubit state $\rho$ (which can be created from the prover's initial state using a polynomial-size quantum circuit) such that on average over $h\in\{0,1\}^n$ sampled by the verifier at step 1 of the protocol, the distribution over decoded measurement outcomes obtained by the verifier at step 5 of the protocol (in case $c=1$) is computationally indistinguishable from some distribution on $\{0,1\}^n$ that is within statistical distance at most $C_1 \eps^{c_1}$ from the distribution which results from measuring $\rho$ in the basis determined by $h$.
	\end{claim}
	
	We sketch a proof of this claim by relying on the proof of Claim~\ref{claim:57}  given in~\cite{mahadev2018classical}. In the course of the proof we explain the terminology used in the statement of Claim~\ref{claim:57}. 
	
\begin{proof}[Proof of Claim~\ref{claim:57-new}]
As shown in~\cite{mahadev2018classical}, it is straightforward to construct a trivial prover from an arbitrary prover that succeeds with probability sufficiently close to $1$ in the ``test'' part of the protocol without affecting the prover's answers in the ``Hadamard'' part by too much (assuming the prover has a high overall probability of success in the first place). First, we can obtain (i) in Definition \ref{def:trivial-prover} because the prover itself can check if the answer it would provide in a test round would be accepted,\footnote{This is not the case for answers in the Hadamard round, that require the secret key $sk$ to be verified.} so a modified prover, if it so desires, can repeatedly simulate the original prover until it obtains a commitment string $y$ such that the resulting post-measurement state would lead to acceptance in a test round. The argument to obtain (ii) is slightly more complicated, and appears as the proof of~\cite[Claim 7.3]{mahadev2018classical}. There it is argued that, for the case of a prover that satisfies (i), replacing $U$ by a twirling (conjugation) of it by random Pauli $\sigma_Z$ operators on the first $n$ qubits does not affect the decoded outcomes obtained by the verifier in a way that would be efficiently noticeable. 

In our context we would like an extractor with black-box access to the ITM representing $P$ (as described in Section~\ref{sec:qma-agree-protocol})
 to be able to modify the ITM for $P$ into an ITM for a trivial prover, so that we can then apply Claim~\ref{claim:57}. Unfortunately, the modification required for (i) that we have described requires polynomially many copies of a ground state of $H$, to which the prover may not have access (so that the extractor may not be able to convert a given prover in a black-box manner to a trivial prover). 
Fortunately we note that the requirement that the prover succeeds with probability negligibly close to $1$ in the test round is not necessary to proceed with the proof; instead, it is sufficient to require that the prover succeeds with probability sufficiently close to $1$, where ``sufficiently'' can be an inverse polynomial in the instance size. To see why this is the case we observe that performing the modification for (ii) on a prover that only succeeds with probability $1-\eps$ in the test round, as opposed to negligibly close to $1$, will only affect the verifier's decoded measurement outcomes up to statistical distance $\poly(\eps)$. This is because, as shown in the proof of~\cite[Claim 7.2]{mahadev2018classical} the transformation from a general prover to one that satisfies (i) only affects the internal state of the prover at the start of step 4 of the protocol described in Section~\ref{sec:meas-protocol} to within trace distance at most $\poly(\eps)$. Since the twirling operation performed for (ii) has a negligible effect on the verifier's decoded measurement outcomes for the case of a prover satisfying (i), the same operation cannot lead to outcomes that are distinguishable from the original outcomes with advantage more than $\poly(\eps)$ when applied to a prover that succeeds with probability only $1 - \eps$, instead of negligibly close to $1$. 

Removing the assumption that the prover is trivial also requires us to weaken the statement ``for all $h$'' in Claim~\ref{claim:57} to ``on average over $h$'' here, because it is possible that a non-trivial prover fails with probability $1$ on a small fraction of the possible choices for $h$; however, this is unimportant as the verifier can use a distribution on bases $h$ of its choice at the first step of the protocol (and in particular the distribution can depend on the structure of the local Hamiltonian instance), and the ``average'' referred to in the claim is then taken with respect to the verifier's choice of distribution.
\end{proof}
	
	Implicit in the proof of Claim~\ref{claim:57-new} is the description of the procedure that is able to generate the claimed state $\rho$ by interacting with an ITM implementation of any prover that satisfies the assumptions of the claim. For concreteness we now make this description explicit. In doing so we follow~\cite[Protocol 5.8]{mahadev2018classical}, adapting that procedure to the case of a prover that is not necessarily ``trivial''; the only modification is the twirling performed at step 3.
	\begin{enumerate}
	\item Sample $(pk,sk)$ as the honest verifier $V$ would for the case when $h$ is the all-$1$ string, $h=1^n$;
	\item Provide $pk$ to $P$ and execute its first unitary $U_0$ on the registers $\reg{K}$ (containing $pk$) and $\reg{S}$ (the prover's private register). Measure the first $n\log|\mathcal{Y}|$ qubits (register $\reg{Y}$) in the computational basis to obtain  commitment strings $(y_1,\ldots,y_n)$; 
	\item Apply a random $\sigma_Z$ operator on each of the first $n$ qubits of register $\reg{M}$. Apply the prover's second unitary $U$ (recall that this corresponds to the case $c=1$) on registers $\reg{MS}$. Apply the same $\sigma_Z$ operators on $\reg{M}$. Do not measure the first $n$ qubits. (These qubits are referred to as \emph{committed qubits}.) Measure the following $nw$ qubits in the Hadamard basis to obtain a string $d$. 
	\item For $i\in\{1,\ldots,n\}$, apply $\sigma_Z^{d_i\cdot (x_{0,i} \oplus x_{1,i})}$ to the $i$-th committed qubit. Here $x_{0,i}$ and $x_{1,i}$ are $w$-bit strings that can be inferred from $y_i$ using $sk$, the inner product is taken modulo $2$, and $\sigma_Z^0=\Id$. 
	\item Return the state $\rho$ of the $n$ committed qubits. 
	\end{enumerate}
	Note that given black-box access to an ITM for $P$, the extractor can be implemented efficiently, as it only performs each of the prover's unitaries $U_0$ and $U$ once, as well as a small number of single-qubit gates on the unmeasured message register $\reg{M}$ and some classical polynomial-time pre- and post-processing.
	
	\subsection{Arguments of Quantum Knowledge for QMA relations}

We show the main result of this section. 
	
	\begin{theorem}\label{thm:meas}
	Let $(Q,\alpha,\beta)$ be a $\QMA$ relation that satisfies properties~\ref{enu:qma-a1} and~\ref{enu:qma-a2} described in Section~\ref{sec:qma-scenario}. There exists a polynomially bounded $\ell=\ell(n)$ such that the following holds. 
	Under the Learning with Errors assumption the protocol presented in Section~\ref{sec:qma-agree-protocol} is secure with completeness $c$ (for the class of input generation algorithms $\mI_Q^{(\ell)}$), up to knowledge error $\kappa$ and with extraction distance $\delta$ for the Agree-and-Prove scenario $(\mF_Q,\mC_Q,\mR_Q)$, where: $c$ is negligibly close to $1$; $\kappa$ is bounded away from $1$ by an inverse polynomial; $\delta = \poly(1-p)\poly(n)$ (for any prover success probability $p > \kappa$). 
		\end{theorem}
	
	\begin{proof}
	The completeness requirement follows immediately from completeness of the protocol from Section~\ref{sec:meas-protocol}, as shown in~\cite{measurement}, and provided $\ell$ is a large enough polynomial in $n$. For soundness, the construction of the extractor is described at the end of Section~\ref{sec:qma-extractor}. Fix an $(x,\sigma) \in R_{Q,\alpha}$ and let $n=|x|$. 
	Suppose that $P^*$ is a prover that succeeds with probability $p = 1-\eps > \kappa$ in the associated protocol, where $\kappa$ is as in the theorem statement. In particular, the prover is accepted in a test round with probability at least $1-2\eps$. By Claim~\ref{claim:57-new}, measurement outcomes obtained on the extracted state $\rho'$ are computationally indistinguishable from a distribution on $n$-bit strings that is within statistical distance at most $\poly(\eps)$ from those obtained by the verifier at step 5 of the verification protocol. We claim that the extractor returns a state $\rho'$ such that 
	\begin{equation}\label{eq:qn-1}
	\Pr(Q_n(x,\rho')=1) \geq 1-\poly(\eps)\cdot\poly(|x|)\;.
	\end{equation}
	The reason is that, as shown in~\cite[Section 8]{mahadev2018classical} (based on~\cite{fitzsimons2018post}), the verifier's decision at step 6 of the verification protocol involves an efficient computation on the decoded measurement outcomes. Moreover, whenever the measurement outcomes are truly obtained from measurements on a state $\rho'$, then the verifier's acceptance probability at step 6. is $\frac{1}{2}(1-\Tr(H\rho'))$. Therefore,	the aforementioned computational indistinguishability implies that the quantity $1-\Tr(H\rho')$ is within $\delta=\poly(\eps)$ of the verifier's acceptance probability in a Hadamard round. 
By the second property assumed in the theorem statement, this implies the claimed bound~\eqref{eq:qn-1}. Assuming $\kappa$ chosen to be sufficiently close to $1$, $\eps$ is small enough that the right-hand side of~\eqref{eq:qn-1} is positive. 
	\end{proof}
	
	\subsection{Sequential amplification}
		
	The protocol presented in the previous sections (Theorem~\ref{thm:meas}) has a knowledge error that is inverse polynomially close to $1$. In order to amplify this error we consider a modified protocol that essentially amounts to sequential repetition of the one we have presented. Precisely, given a function $t:\N\times\N \to \N$ that parametrizes the protocol, $V_2$ and $P_2$ engage in $t(\lambda,n)$ iterations of the verification protocol from Section~\ref{sec:meas-protocol}, where $\lambda$ is the security parameter provided as input and $n=|x|$ is the length of the statement. The only difference is that instead of making a decision at the end of any given iteration of the protocol, the verifier collects the decisions that it would have made. At the end of the $t$ iterations, it accepts if and only if in the original protocol, it would have accepted in each of the iterations.
	For this protocol, we show the following. 
	
	\begin{corollary}\label{cor:meas-rep}
Let $(Q,\alpha,\beta)$ be a $\QMA$ relation that satisfies properties~\ref{enu:qma-a1} and~\ref{enu:qma-a2} described in Section~\ref{sec:qma-scenario}. Let $\kappa,\eta:\N\to [0,1]$. Then there exists functions $\ell(n)=\poly(n)$ and $t(\lambda,n) = \log(1/\kappa(\lambda))\cdot \poly(n)/\poly(\eta)$ such that the following holds. 
	Under the Learning with Errors assumption the $t$-fold sequential repetition of the protocol for the Agree-and-Prove scenario $(\mF_Q,\mC_Q,\mR_Q)$  is secure with completeness $c$ (for the class of input generation algorithms $\mI_Q^{(\ell)}$), up to knowledge error $\kappa$ and with extraction distance $\delta$, where: $c$ is negligibly close to $1$, $\kappa$ is as specified, and $\delta = 1-p(1-\eta)^2$ where $p>\kappa$ is the prover success probability.
		\end{corollary}
		
		\begin{proof}
Completeness of the protocol follows from the fact that completeness of the base $t=1$ protocol is negligibly close to $1$, as stated in Theorem~\ref{thm:meas}. For soundness, let $P^*$ be a prover that succeeds with probability $p>\kappa$ in the sequentially repeated protocol. By a standard martingale argument, provided $t$ has been chosen large enough, a fraction at least $(1-\eta)$ of rounds $i^*\in\{1,\ldots,t\}$ are such that, conditioned on succeeding in previous rounds, $P^*$ succeeds in the $i^*$-th round with probability $1-\eps$, where $\eps$ is such that the extraction distance function  $\delta$ from Theorem~\ref{thm:meas} satisfies $\delta(1-\eps) \leq \eta$. Call any such round ``good''. We define the extractor as follows. $E$ selects an $i \in \{1,\ldots,t\}$ uniformly at random. It executes the interaction of $P^*$ and $V$ up to the $(i-1)$-st round included. Then, it attempts to extract a witness by using the strategy of the single-round extractor from Theorem~\ref{thm:meas}. Conditioned on $i$ being a good round and on $P^*$ being accepted by $V$ in the first $(i-1)$ rounds, the extractor returns a state that passes verification with probability $1-\delta(1-\eps)$. 
		\end{proof}



\newcommand{\etalchar}[1]{$^{#1}$}

\end{document}